\documentclass[acmsmall, nonacm, authorversion, screen]{acmart}

\usepackage{xspace, nicefrac, ctable}

\newcommand{\Opt}{\ensuremath{\texttt{OPT}}\xspace}
\newcommand{\opt}{\Opt}
\newcommand{\Alg}{\ensuremath{\texttt{ALG}}\xspace}
\newcommand{\OPT}{\opt}

\newcommand{\TAF}{\ensuremath{\texttt{AnyFit}_\tau}\xspace}
\newcommand{\TAAF}{\ensuremath{\texttt{AlmostAnyFit}_\tau}\xspace}
\newcommand{\TNF}{\ensuremath{\texttt{NextFit}_\tau}\xspace}
\newcommand{\TWF}{\ensuremath{\texttt{WorstFit}_\tau}\xspace}
\newcommand{\TFF}{\ensuremath{\texttt{FirstFit}_\tau}\xspace}
\newcommand{\TBF}{\ensuremath{\texttt{BestFit}_\tau}\xspace}
\newcommand{\GHAR}{\ensuremath{\texttt{Harmonic}_\tau}\xspace}

\newcommand{\TAAFop}{\ensuremath{\texttt{AlmostAnyFit}_{\tau^*}}\xspace}
\newcommand{\TNFop}{\ensuremath{\texttt{NextFit}_{\tau^*}}\xspace}
\newcommand{\TWFop}{\ensuremath{\texttt{WorstFit}_{\tau^*}}\xspace}

\newcommand{\GHARop}{\ensuremath{\texttt{Harmonic}_{\tau^*}}\xspace}

\newcommand{\AF}{\ensuremath{\texttt{AnyFit}}\xspace}
\newcommand{\AAF}{\ensuremath{\texttt{AlmostAnyFit}}\xspace}
\newcommand{\NF}{\ensuremath{\texttt{NextFit}}\xspace}
\newcommand{\WF}{\ensuremath{\texttt{WorstFit}}\xspace}
\newcommand{\FF}{\ensuremath{\texttt{FirstFit}}\xspace}
\newcommand{\BF}{\ensuremath{\texttt{BestFit}}\xspace}
\newcommand{\HAR}{\ensuremath{\texttt{Harmonic}}\xspace}

\newcommand{\gbp}{\texttt{GBP}\xspace}
\newcommand{\black}{black\xspace}
\newcommand{\CR}{\texttt{CR}\xspace}
\newcommand{\WPC}{WPC\xspace}

\usepackage{subcaption}
\usepackage{thmtools}
\usepackage{nicefrac}

\usepackage[dvipsnames]{xcolor}
\definecolor{darkblue}{RGB}{0 0 100}

\newcommand{\shahin}[1]{  
	{\textcolor{orange}{(Shahin says:  #1)}}{}} 
\newcommand{\shahinadd}[1]{  
	{\textcolor{black}{#1}}{}} 

\setcopyright{acmlicensed}
\copyrightyear{2026}
\acmYear{2026}
\acmDOI{XXXXXXX.XXXXXXX}

\acmConference[SIGMETRICS '26]{the ACM on Measurement and Analysis of Computing Systems}{June 08--12,
  2026}{Ann Arbor, Michigan}

\title{Green Bin Packing}

\author{Jackson Bibbens}
\affiliation{%
  \institution{University of Massachusetts Amherst}
  \country{USA}
}
\email{jbibbens@umass.edu}

\author{Cooper Sigrist}
\affiliation{%
  \institution{University of Massachusetts Amherst}
  \country{USA}
}
\email{csigrist@umass.edu}

\author{Bo Sun}
\affiliation{%
  \institution{University of Ottawa}
  \country{Canada}
}
\email{bo.sun@uottawa.ca}

\author{Shahin Kamali}
\affiliation{%
  \institution{York University}
  \country{Canada}
}
\email{kamalis@yorku.ca}

\author{Mohammad Hajiesmaili}
\affiliation{%
  \institution{University of Massachusetts Amherst}
  \country{USA}
}
\email{hajiesmaili@cs.umass.edu}

\begin{abstract}
     The online bin packing problem and its variants are regularly used to model server allocation problems. Modern concerns surrounding sustainability and overcommitment in cloud computing motivate bin packing models that capture costs associated with highly utilized servers. In this work, we introduce the \textit{green bin packing} problem, an online variant with a linear cost $\beta$ for filling above a fixed level $G$. For a given instance, the goal is to minimize the sum of the number of opened bins and the linear cost. We show that when $\beta G \le 1$, classical online bin packing algorithms such as FirstFit or Harmonic perform well, and can achieve competitive ratios lower than in the classic setting. However, when $\beta G > 1$, new algorithmic solutions can improve both worst-case and typical performance. We introduce variants of classic online bin packing algorithms and establish theoretical bounds, as well as test their empirical performance.

\end{abstract}


\begin{document}

\maketitle

\section{Introduction}
\label{sec:intro}
Online bin packing is a classic problem in the field of online algorithms: a sequence of items with sizes in $(0,1]$ must be packed into the minimum number of unit capacity bins in an online fashion. This basic setting finds applications in diverse areas such as hospital room scheduling, packing physical trucks, and cloud computing. This last case has been of particular interest in recent literature, with a number of works studying variants of online bin packing which better capture aspects of resource allocation in the cloud \cite{Hong2023SBP, Cohen2019Overcommitment, Kamali2021FaultTolerant}. 

Motivated by two practical scenarios in cloud computing: (1) carbon-aware geographical load balancing~\cite{lechowicz2025learning,liu2011greening, Lin2012GLB, Lin2023GLB, Murillo2024GLB, Sukprasert2024} and (2) overcommitment in cloud resource allocation~\cite{Cohen2019Overcommitment,cortez2017resource,reidys2025coach,kumbhare2021prediction,bashir2021take}, this paper introduces and studies the online \textit{green bin packing} problem (\gbp). In \gbp, each bin has fixed cost to open and a normalized capacity of 1, partitioned into two regions: a \textit{green} portion of size $G \le 1$, representing carbon-free resources in scenario (1) or guaranteed-availability resources in scenario~(2), and a remaining \textit{black} portion of size $1 - G$. Packing items into the green region incurs no cost, while using the black region, i.e., placing items beyond the threshold $G$, incurs a linear cost at a rate $\beta$. We refer to \autoref{sec:model-preliminaries} for more details on the formal statement of the problem and details of the motivating applications. 

In this paper, we tackle \gbp and provide foundational results for both its offline and online variants. Our findings reveal both intriguing similarities to and important distinctions from the classical bin packing setting, while also opening new directions for future research. A central motivation for \gbp arises from a key algorithmic challenge: balancing two distinct cost functions---the fixed cost of opening a new bin and the variable cost incurred when usage exceeds a predefined ``green'' threshold in existing bins. This trade-off introduces a novel design dimension that is absent in traditional formulations. We show that directly applying classical bin packing algorithms without accounting for this new cost structure can lead to unbounded competitive ratios (see \autoref{prop:classicUnbounded}), underscoring the need for new algorithmic approaches tailored to the \gbp setting.

\subsection{Contribution}
We provide a comprehensive algorithm design and analysis for \gbp in both offline and online settings. 
In~\autoref{sec:fundamentals}, and for the completeness of the result in terms of complexity analysis of \gbp, we focus on offline \gbp and prove that it is NP-hard in \autoref{thm:np-hard} and further show that an asymptotic polynomial-time approximation scheme (APTAS) exists in \autoref{thm:aptas}. 

\paragraph{Fundamental algorithmic ideas.}
In the online setting, the core question of \gbp is that for each item, \textit{should we open a new bin at cost $1$ to use its green space $G$, or should we use the black space in opened bins at unit cost $\beta$?} As shown concretely in~\autoref{subsec:approaches}, answering this question clearly depends on the values of $\beta$ and $G$. In particular, with ``cheap'' black cost, i.e., small values of $\beta G$ ($\beta G \leq 1$), it is best always to avoid opening a new bin if possible. For expensive black costs, i.e., larger values of $\beta G$,  ($\beta G > 1$), we design simple competitive algorithms for all given $G$ and $\beta$. This separation guides our algorithm design in this paper. 
Specifically, we consider a class of threshold-based algorithms analogous to classic bin packing algorithms. For a given threshold $\tau \in [0,1-G]$, we define $G+\tau$ as the effective capacity, i.e., the maximum capacity that can be used to pack items in a bin, and run the classic \NF, \WF, \AAF, and \HAR algorithms using this effective capacity. The main contribution of this work is to determine the threshold $\tau$ for each given $G$ and $\beta$, and to analyze the competitive ratios of the corresponding algorithms. We summarize our results in \autoref{tab:less1} and \autoref{tab:greater}, and visualize these results in \autoref{fig:allGB} (with elaboration of the behaviors in~\autoref{Appendix-Visual}). In particular, our results can be divided into two cases.

\paragraph{Algorithms and results when $\beta G \leq 1$.}
In~\autoref{sec:GBleqOne}, we focus on the case when the unit cost of black space is relatively low, i.e., $\beta G \le 1$, and we set $\tau = 1 - G$. Thus, the algorithms target to fully utilize the black space, filling each bin up to capacity $1$. We derive (nearly)-tight competitive ratios for all algorithms %
as functions of $\beta$ and $G$. Unlike classic bin packing, where \NF and \WF both have a competitive ratio of 2~\cite{Johnson16}, for \gbp we establish a clear separation between \NF and \WF, where the performance of \WF is much better. Besides this, the competitive ratio of other algorithms exhibits a similar pattern and separation to classic bin packing, that is \BF and \FF have similar competitive ratios (which is better than \WF), and \HAR has a slightly improved competitive ratio.
Compared to classic bin packing, in \gbp, the competitive ratios of these algorithms are generally smaller. 
This is because the additional cost of black space affects both the online algorithms and the offline optimum, leading to improved ratios. %

\paragraph{Algorithms and results when $\beta G > 1$.}
In~\autoref{sec:GBgtOne}, we focus on the case when the unit cost of black space is high, i.e., $\beta G > 1$, and the threshold $\tau$ depends on both $\beta$ and $G$. We derive upper and lower bounds for all algorithms as functions of $\tau$, which are complex piecewise functions (see~\autoref{fig:allGB} for an illustration). Interestingly, the competitive ratios with optimized thresholds exhibit new performance structures. In particular, \NF still has a competitive ratio of $2$ when the threshold is set $\tau^* = 0$ for large $\beta G$. 
In contrast, all other algorithms can better utilize opened bins compared to \TNF and set the threshold $\tau^* = \Theta({1}/{\beta})$. 
Interestingly, \TWF achieves a competitive ratio of $2 - \Omega(1/(\beta G))$, which clearly separates it from \TNF in a departure from the classic behavior.
This improvement is  due to the load-balancing nature of \TWF, which tends to fill the bins with the smallest load first. This effectively prioritizes green space usage and minimizes the use of expensive black space. Finally, \TAAF and \GHAR attain a ratio of $1.75 - \Omega(1/(\beta G))$, which degrades with increasing $\beta G$, as \gbp reduces to classic bin packing with capacity $G$ when $\beta G \to\infty$. 

\paragraph{Numerical results}
Finally, we conduct extensive experiments to evaluate performance in both cases of small and large $\beta G$. \BF and \FF generally achieve the best performance among all algorithms. The most notable observation is that \WF demonstrates strong empirical performance, especially when $G$ is small and $\beta$ is large. This behavior contrasts sharply with that of classic bin packing and suggests new design principles for \gbp.

\begin{table}
\centering
\caption{Summary of results for $\beta G \le 1$. For each algorithm and the online deterministic lower bound, we list competitive ratio bounds and compare the (upper) bound value when $\beta \rightarrow0$ to the bound in classic bin packing. Note that when $\beta \rightarrow 0$, \gbp becomes equivalent to classic bin packing where bins have capacity 1.}
 \small
\begin{tabular}{c c c c c c} 
 \specialrule{.1em}{.05em}{.05em} 
 Algorithm & Competitive Ratio & $\beta \rightarrow 1/G$ & $\beta \rightarrow 0$ & Classic & Theorem \\ [0.5ex] 
 \hline
 \\[-2.1ex]
 \NF & $\frac{2+\beta(1-G)}{\beta+(1-\beta G)}$ & $G+1$ & 2 & 2 & \ref{thm:nextfit} \\ [0.6ex] 
 \WF & $\frac{2+\beta(1-2G)^+}{\beta+(1-\beta G)}$ & $\max\{1, 2G\}$ & 2 & 2 & \ref{thm:worstfit} \\[0.6ex] 
 \AAF & piecewise (\ref{eq:craaf}) & $\max\{1,\frac{2+\beta}{2\beta},7/4\beta\}$ & 1.75 &  1.7 & \ref{thm:anyfitLB}, \ref{thm:anyfitUB} \\[0.6ex] 
 \HAR & piecewise (\ref{eq:crhar}) & $\max\{1,\frac{2+\beta}{2\beta},\frac{11.146-\beta}{6\beta}\}$ & 1.691 & 1.691 & \ref{thm:anyfitLB}, \ref{thm:harUB} \\ [0.2ex]
\hline
 \\[-2.1ex]
  \texttt{GeneralLB} & $\frac{248/161+\beta(1-248G/161)^+}{1+\beta(1-G)}$& $\max\{1,1.5403G\}$&1.5403 & 1.5403 & \ref{thm:GenLB-gbLtOne}\\[0.6ex] 
 \specialrule{.1em}{.05em}{.05em}
\end{tabular}
 \normalsize
\label{tab:less1}
\end{table}

\begin{table}
\centering
 \caption{Summary of results for $\beta G > 1$. For each algorithm and online deterministic lower bound, we list competitive ratio bounds and compare the (upper) bound values at $\beta G  \rightarrow \infty$ to the \CR for classic bin packing. Note that as $\beta G\rightarrow \infty$, \gbp becomes equivalent to classic bin packing where bins have capacity $G.$}
 \small
\begin{tabular}{c c c c c c} 
 \specialrule{.1em}{.05em}{.05em} 
 Algorithm & Competitive Ratio & $\beta G \rightarrow 1$ & $\beta G\rightarrow\infty$ & Classic & Theorem \\ [0.5ex] 
 \hline
 \\[-2.1ex]
 \TNFop & piecewise (\ref{eq:nextfittau}) & $G+1$ & 2  & 2 & \ref{thm:tnf} \\ [0.6ex] 
 \TWFop & $\max\{\frac{2\beta G}{\beta G+1},2G\}$ & $\max\{1,2G\}$ & $2 - \Omega(\frac{1}{\beta G})$ & $2$ & \ref{thm:taf-lb}, \ref{thm:TWFUB} \\[0.6ex] 
 \TAAFop & piecewise (\ref{Cor:optimalTauOne}) & $\max\{1,\frac{6G+1}{4},\frac{4-\beta}{2}\}$ & $1.75 - \Omega(\frac{1}{\beta G})$ & 1.7 & \ref{thm:taf-lb}, \ref{thm:taf-UB} \\[0.6ex] 
 \GHARop & piecewise (\ref{Cor:optimalTauOne}) & $\max\{1,\frac{6G+1}{4},\frac{4-\beta}{2}\}$ & $1.75 - \Omega(\frac{1}{\beta G})$ & 1.691 & \ref{thm:taf-lb}, \ref{thm:taf-UB} \\ [0.2ex]
\hline
 \\[-2.1ex]
  \texttt{GeneralLB} & piecewise (\ref{Thm:GenLB_BGgt1}) & $\max\{1, \frac{2.5403-G}{1+\beta(1-G)}\}$&  1.5 & 1.5403  & \ref{Thm:GenLB_BGgt1}\\[0.6ex] 
 \specialrule{.1em}{.05em}{.05em}
\end{tabular}\normalsize
\label{tab:greater}
\end{table}

\subsection{Overview of Techniques}
Our analysis builds on several classical techniques from the bin packing literature, but requires certain modifications to address the new cost structure introduced by \gbp. 
While the general strategy follows familiar lines---combining approximation schemes, adversarial analysis, and weighting arguments---the coexistence of fixed and variable costs adds a layer of complexity that fundamentally changes algorithm design and analysis.

\paragraph{Extending classical frameworks.}
We first extend established theoretical tools to the \gbp setting. 
On the offline side, we adapt the APTAS framework of de la Vega and Lueker~\cite{Vega81apt} and the NP-hardness reduction from \textsc{Partition} to handle the two-tier cost model. 
For the online setting, we generalize Yao’s general lower-bound framework~\cite{Yao1980genLB} to account for the black-space penalty. 
These extensions preserve the spirit of the classical analyses but reveal a key difference: the optimal fill level and structure of feasible packings now depend on the parameters $\beta$ and $G$, in contrast to the fixed threshold 1 of traditional bin packing. 
This dependency propagates to all aspects of the competitive analysis and leads to substantial complexities in the characterization of the adversarial patterns as described below. 

\paragraph{Multiple adversarial patterns and piecewise competitive ratios.}
A defining feature of \gbp is that no single adversarial input achieves the worst-case performance across all parameter values. 
Instead, the interplay between the fixed cost of opening bins and the linear black-space cost produces multiple adversarial regimes, each defining the competitive ratio for different ranges of~$\beta G$. 
For example, when~$\beta G$ is small, the adversary may exploit the algorithm’s inefficiency in utilizing green space---forcing it to open many bins, much like in the classical setting. 
However, when~$\beta G$ is large, the adversary may instead construct sequences of tiny ``sand'' items that compel the algorithm to fill bins beyond $G$ and incur large black-space costs.  There are other adversarial patterns, dependent on the specific algorithm.
The transition between these behaviors occurs at distinct thresholds of~$\beta G$, yielding \textit{piecewise competitive ratios} whose defining input sequences change discontinuously with~$(\beta, G)$. 
Identifying these adversarial patterns for each algorithm requires a substantial departure from proof techniques in the classic bin packing and forms the core of our lower-bound analysis and constitutes one of the main technical contributions of this work.

\paragraph{Modified weighting arguments for upper bounds.}
To obtain upper bounds, we introduce a new version of the classical weighting argument. 
In the standard analysis of bin packing, weights are assigned so that the total \textit{weight per bin} can be bounded, which in turn bounds the number of bins used. 
This approach no longer works in \gbp, where the cost is not determined solely by the number of bins but also by how much black space each bin uses. 
We therefore generalize the idea to measure the \textit{weight per cost}, accounting for both the cost of opening bins and the variable cost of black-space usage. 
This framework allows us to reason about how different parts of the cost contribute to the overall performance of an algorithm and to identify the trade-offs between green and black utilization. 
For algorithms such as \textit{AlmostAnyFit}$_\tau$ and \textit{Harmonic}$_\tau$, we further extend this approach into a parameterized analysis that traces a family of bounding curves as functions of parameter $\tau$.

\paragraph{Analytical implications.}
These methodological departures lead to several phenomena absent from classical analyses: 
(i) competitive ratios become continuous, piecewise functions of~$\beta G$; 
(ii) equivalences such as the identical performance of \NF and \WF no longer hold---\WF performs strictly better due to its load-balancing behavior. 
Together, these insights form the technical foundation of our study and point toward new directions for analyzing online algorithms under hybrid cost structures.

\section{Related Work}

Bin packing is a foundational optimization problem, with a rich literature covering both the original statement and many variants under offline, online, and stochastic regimes. Across these regimes, some common variants include multiple-dimensional bin packing \cite{Karp1984MDBP,Kou1977MDBP,panigrahy2011heuristics}, bin packing with item migrations \cite{Mellou2024, Jansen2019Migration, Gambosi2000Migration, Balogh2014Migration}, bin packing with departures \cite{Mellou2024, Liu2022DBP, Hong2023SBP}, and bin packing with predictions of item sizes \cite{Liu2022DBP, angelopoulos2024BPpredictions}. 
The regime of analysis also separates the literature of bin packing; this paper resides in the literature of online competitive analysis, distinguishing it from offline or stochastic approaches. In the offline setting, the entire set of items is presented simultaneously, and the goal is typically to approximate the optimal solution in a reasonable time \cite{Vega81apt, BALOGH2015offline, Borgulya2021}. The techniques used in the offline setting require full information over the set of bins, making them infeasible for the online regime. Meanwhile, in the stochastic setting, item sizes are sampled iid from some distribution. In this case, algorithms are evaluated typically by their expected loss, which is the difference between the number of opened bins and the total size of all items (normalized to the bin capacity) \cite{Csirik2006, Rhee1993, Gupta_2020}. While these algorithms can perform well under stochastic inputs, they do not provide good worst-case guarantees in terms of competitive ratio~\cite{Banerjee2020, Csirik2006}. In contrast, the \gbp problem considered in this work builds upon the standard online bin packing framework \textit{without} incorporating these additional extensions or analysis regime. Extending \gbp to include such dimensions presents compelling directions for future research in green bin packing.

Focusing on more relevant literature, \gbp has similarities to existing works, which either allow overfilling a bin for additional cost~\cite{PerezSalazar2022,Denton2010OR,Coffman2006, Levin2022GEPB} or are motivated by sustainable computing~\cite{Song2014AdaptiveRP,Kamali2015EfficientBP}. \citet{PerezSalazar2022} provide a model of bin packing with an overfilling cost and where item sizes are unknown at the time of packing. In this model, the overfilling cost is a flat penalty regardless of the amount the bin is overfilled, and at most one item is allowed to be inserted above capacity. Not only do these restrictions differ from the \gbp model, but \cite{PerezSalazar2022} considers the stochastic regime while we consider the online setting. The concept of overfilling is also seen in previous research on operating room scheduling with overtime costs \cite{Denton2010OR} and (generalized) extensible bin packing \cite{Coffman2006, Levin2022GEPB}. Similarly to \gbp, these models allow for some form of overfilling a bin at an additional cost linear to the amount overfilled. However, these models differ by allowing the ability to overfill a bin by an arbitrarily large amount (we are limited to $1-G$ in \gbp), and are all studied under the offline regime.
At least two previous works have positioned themselves in the intersection of sustainability and bin packing~\cite{Song2014AdaptiveRP,Kamali2015EfficientBP}. However, these works have been primarily framed as reducing the total number of servers (bins) used in data centers. These models do not account for the green cost based on bin utilization. Additionally, while these works are within the online setting, they both allow item migration, which is a significant departure from our model.

\section{Model and Preliminaries}
\label{sec:model-preliminaries}
\subsection{Problem Definition} 

We define the Green Bin Packing (\gbp) as an online optimization problem. 
We have an infinite number of bins, each having a maximum capacity of $1$. The capacity of each bin is divided into two parts: for a given $G\in [0,1]$, the capacity within $[0, G]$ is called \textit{green space}, and the capacity within $(G,1]$ is called \textit{\black space}. Opening a new bin incurs a cost of $1$, and using each unit of \black space costs $\beta > 0$. A sequence of $n$ items, $(a_1,\dots,a_n)$, arrive one at a time.  
Each item $i$ has a size $a_i \in (0,1]$, and upon its arrival, it must be assigned to an existing bin whose current fill level is no more than $1 - a_i$, or placed in a newly opened bin.  
After packing all $n$ items, let $N$ denote the number of opened bins and $L_j$ denote the final fill level of the $j$-th bin. 
The objective is to minimize the overall cost of opening bins and using \black space, i.e., $N + \sum_{j\in[N]}\beta (L_j - G)^+$, where $(L_j-G)^+ := \max\{0, L_j-G\}$. 

For given setup parameters $G$ and $\beta$, we use $\sigma :=\{a_1,\dots,a_n\}$ to denote an instance of \gbp. We define the offline optimal algorithm $\opt$ as that which results in the minimal cost for a given instance. The standard metric used to evaluate an online algorithm \Alg is its \textit{competitive ratio} (\CR), which is the worst-case ratio of the cost of \Alg to the cost of \Opt. Following the standard approach for classical bin packing \cite{Johnson74}, in this paper we consider the \textit{asymptotic} competitive ratio, in which the cost of \Opt is unbounded (i.e., the input sequence is arbitrarily large). Any future references to competitive ratio in this text refer to the asymptotic definition.

This definition is structured to capture applications where bins have two cost tiers: a free capacity and then a costly capacity. This relates to problems such as server allocation for sustainability (renewable energy is "green", energy from fossil fuels is "\black") or for overcommitment (actual capacity is "green" and overcommitted capacity is "black").

\subsection{Motivating applications}
\label{sec:app}
Incorporating a two-stage cost model into the classic bin packing problem is a natural and meaningful extension that introduces significant new challenges for algorithm design and analysis. At the same time, this formulation serves as a stylized yet practical model that captures key aspects of real-world bin packing scenarios. In the following, we highlight two representative applications that illustrate its relevance.

\paragraph{Greening geographical load balancing}
In a geographically distributed system, job requests can be served at multiple locations (data centers), each equipped with a limited capacity of on-site carbon-free energy (e.g., solar or wind). When demand exceeds this local green capacity, the system must rely on off-site power sources, often from the external grid or local generators, which may involve significantly higher carbon emissions. This setting is motivated by the growing demand for data center services and the associated concerns about their carbon footprint~\cite{hajiesmaili2025toward,irwin2025vision}. As future data centers increasingly integrate on-site renewable energy, they will still face the challenge of managing excess demand that must be met using carbon-intensive sources. This trade-off between clean and dirty energy motivates our model in \gbp, where serving additional load beyond green capacity incurs a carbon cost.

\paragraph{Cloud resource overcommitment}
The \gbp problem could also be considered as initial steps for the modeling of resource overcommitment cost. 
A central challenge in this domain---particularly in large-scale cloud platforms such as Amazon AWS, Microsoft Azure, and Google Cloud Engine---is the efficient allocation of physical resources to virtual machines (VMs). This challenge is often modeled using practical variants of the bin packing problem \cite{barbalho2023virtual, Grandl2014ClusterSched, Hadary2020Protean, mohan2022looking}. In practice, due to the unpredictability of VM resource demands and to reduce operational costs, operators commonly allocate only a fraction of the requested resources, e.g., 65\% of CPU requests, guided by customized ML-based prediction models~\cite{cortez2017resource}. This introduces a risk of overcommitment when actual resource usage exceeds the provisioned physical capacity, typically due to underestimation of ML-based models.
In such scenarios, the incurred cost is better represented by a two-tier cost structure---where only the portion exceeding guaranteed capacity is penalized---rather than a simple function of total usage. The \gbp model captures this behavior by using deterministic bin sizes and assigning a linear \black cost to the portion of resource usage that exceeds the green (safe) threshold. We note that Cohen et al.~\cite{Cohen2019Overcommitment} approached a different but related problem from the angle of extending bin packing by formulating it under the setting of stochastic bin sizes (i.e., where the actual size is unknown at the time of assignment) and aiming to bound the probability of overcommitment.
\subsection{A brief overview of existing algorithms for classic bin packing}
Classic bin packing is a special case of \gbp when $G = 1$ or $\beta = 0$, and thus its goal is just to minimize the total number of opened bins. Suppose the instance $\sigma$ is given in advance, the offline bin packing is known to be strongly NP-Hard \cite{lewis1983michael}, but there exists an APTAS for the offline optimal \cite{Vega81apt}. In the online setting, there exists a lower bound on the competitive ratio of all online bin packing algorithms of $\nicefrac{248}{161}\approx1.54037$ \cite{BALOGH2012lowerBound}, which no algorithm has yet achieved. Many well-known algorithms have been proposed for optimizing the upper bound. 
\begin{itemize}
    \item \NF: Each item is placed in the most recently opened bin if possible, else a new bin is opened.
    \item \WF: Each item is placed into the bin with the lowest fill level that the item can fit in. A new bin is opened if none exist.
    \item \BF: Each item is placed into the bin with the highest fill level that the item can fit in. A new bin is opened if none exist.
    \item \FF: Each item is placed into the earliest opened bin that it fits in. A new bin is opened if none exist.
    \item \HAR: Given a parameter, $K$, the range $(0,1]$ is partitioned into intervals: 
    
    ${(\nicefrac{1}{2}, 1],(\nicefrac{1}{3}, \nicefrac{1}{2}], \ldots, (\nicefrac{1}{K}, \nicefrac{1}{K-1}], (0,\nicefrac{1}{K}]}$. Each item may only be placed into a bin if all other items placed in that bin belong to the same interval, and that bin has capacity for the item. A new bin is opened if none exist.
\end{itemize}

To better classify bin packing strategies, past literature has defined \AF algorithms as those that never open a new bin unless the current item to be packed does not fit into any existing open bin in the packing. Additionally, \AAF algorithms are those which will not pack an item into an open bin with the lowest utilization unless there are multiple such bins or it is the only bin with enough capacity for the item \cite{Johnson74}. Note that of the algorithms we consider, \AF includes \WF, \FF, and \BF, while \AAF only includes \FF and \BF. 

\NF and \WF are both 2-competitive, giving them the worst competitive ratios of this list. Meanwhile, the \AAF algorithms, which includes \FF and \BF, have competitive ratios of 1.7 \cite{Johnson74}. Finally, the performance of \HAR improves as $K$ increases, beating $1.7$ by $K=7$, and approaching $1.69103$ in the limit \cite{Lee1985Har}. 
The best known upper bound is 1.5783 which is achieved by an Advanced \HAR algorithm~\cite{balogh2017new}.
Although the family of \HAR algorithms attain the best known competitive ratio, \HAR performs poorly in practice compared to \FF and \BF, which can be run more quickly and converge to be 1-competitive on uniformly random item sizes \cite{Bentley1984}\cite{Lee1985Har}.

\section{New design challenges and our approaches for \gbp}\label{subsec:approaches}

Compared to classical bin packing, online algorithms for \gbp\ must balance two competing costs: opening new bins and using black space. 
When all bins are filled up to the green threshold~$G$, the algorithm must decide whether to pay a fixed cost~$1$ to open a new bin or a per-unit cost~$\beta$ to continue packing into existing bins. 
This trade-off, determined by the parameters~$G$ and~$\beta$, makes direct adaptation of classical algorithms ineffective---indeed, algorithms that always fill bins to capacity can yield unbounded competitive ratios.

\begin{restatable}{proposition}{classicunbounded}\label{prop:classicUnbounded}
Given $G \in [0, 1)$, the competitive ratio of any $\beta, G$ oblivious algorithm for bin packing (including \NF, \WF, \AAF, \HAR) grows linearly with $\beta$.
\end{restatable}

To maintain bounded competitive performance, \gbp algorithms need to limit the usage of \black space when $\beta$ is relatively large.
In an extreme case where the sizes of all items are infinitesimal and the number of opened bins is large, the unit-cost of green space is $1/G$ (the cost of opening a bin over the green space capacity), and the unit-cost of \black space is $\beta$. 
Given this, we have the following conjecture: \gbp can be divided into two basic settings based on the value $\beta G$.
\begin{itemize}
    \item When $\beta G \le 1$ (i.e., $\beta \le 1/G$), the unit-cost of \black space is cheaper and thus competitive algorithms should attempt to fill each bin to its capacity, as seen in classic bin packing. 
    \item When $\beta G > 1$ (i.e., $\beta > 1/G$), the unit-cost of \black space is large, and thus competitive algorithms should attempt to fill each bin between $G$ and $1$, the larger $\beta$, the closer to $G$. 
\end{itemize}

In \autoref{sec:fundamentals}, we present results for the offline setting that motivate our online algorithms. \autoref{sec:GBleqOne} covers the setting where $\beta G \le 1$. We find that classic bin-packing algorithms work well and provide new competitive ratios for \gbp. In \autoref{sec:GBgtOne}, we propose and analyze threshold-based algorithms for the case where $\beta G > 1$. We present our experimental results in \autoref{sec:experiments}, where we analyze algorithms under both problem settings across a variety of simulated inputs.

\section{Offline Green Bin Packing}
\label{sec:fundamentals}
In this section, we present a set of \textit{offline} \gbp results: showing an offline lower bound which we consider optimal in our competitive analysis, proving that the problem is NP-hard, and establishing an asymptotic polynomial time approximation scheme (APTAS) for the problem.

\subsection{Lower Bound on Offline Optimum}
Let \textit{small} items be those with size less than $G$ and \textit{large} items with size at least $G$. 

\begin{restatable}{proposition}{smallItemBlackSpace}\label{lem:smallItemBlackSpace}
    When $\beta G > 1$, an optimal solution will place a volume of at most $\nicefrac{1}{\beta}$ small items in bins with large items. Similarly, no bin will contain more than one large item.
\end{restatable}

\begin{proof}
    Let $\Gamma$ be the set of all small items in bins with large items, with total volume $S_\Gamma$. Since large items have size at least $G$, the small items completely contribute to \black cost. Since this packing is part of an optimal solution, we know that placing the items of $\Gamma$ in bins with large items is cheaper than placing all of $S_\Gamma$ into a new bin: $\beta S_\Gamma \le 1+ \beta (S_\Gamma-G)^+$. Thus, $S_\Gamma\le \frac{1}{\beta}$. 
    
    This argument can be repeated by setting $\Gamma$ as a second large item in the bin. 
\end{proof}

To analyze upper bounds on the competitive ratios, we rely on a lower bound for the optimal cost of the offline \gbp. 
Let $S := S(\sigma)$ denote the total volume/size of items in the instance $\sigma$, and $\Opt(S)$ denote the minimum cost an optimal algorithm would take given instances that have a combined volume of $S$.

\begin{lemma}\label{OptLB1}
On any instance $\sigma$ with a total volume $S$, $\Opt(\sigma) \ge \min_{L \in (0,1]} \nicefrac{S}{L}(1+\beta(L-G)^+)$.
\end{lemma}
\begin{proof}
    An input of volume $S$ requires at least $S$ bins. Given a fixed number of bins in a packing, cost is minimized when each bin is filled to the same level, as this minimizes the usage of \black space. Therefore, the minimum cost of an optimal solution is equal to the minimization across fill levels for the sum of the number of bins and the total \black cost. 
\end{proof}

In the above statement, the equality only holds when the items of $\sigma$ allow for filling each bin to the same level $L$, for all $0  < L \leq 1$. This requires all items to be of sufficiently small size. Thus, we can say that for fixed $S$, $\Opt$ reaches minimum cost on inputs where each individual item is negligibly small.

\begin{theorem}[Lower Bound of Offline \gbp]
\label{thm:OptLB} Consider an input $\sigma$ with total volume $S$. When $\beta G \leq 1$, the lower bound is $ \Opt(\sigma) \geq S + S\beta(1-G)$, with equality when $\Opt$ fills each bin exactly to $1$. When $\beta G > 1$, the lower bound is $\Opt(\sigma) \geq \nicefrac{S}{G}$, with equality when $\Opt$ fills each bin exactly to $G$.
\end{theorem}
\begin{proof}
We will perform the minimization in \autoref{OptLB1}. We can first see that as $L \rightarrow G-$, the expression is $\nicefrac{S}{L}$ and is decreasing. Meanwhile, as $G < L \rightarrow 1-$, the expression becomes $S\beta+\nicefrac{S}{L}(1-\beta G)$. Thus, it will continue decreasing if $\beta G <1$ and will increase if $\beta G > 1$. Therefore, we find that expression is minimized at $S+S\beta(1-G)$ when $\beta G \le 1$, and at $\nicefrac{S}{G}$ when $\beta G > 1$. 
\end{proof}

\subsection{APTAS for offline \gbp}
We show that the offline \gbp is NP-hard for any value of $\beta >0$ and $G \in [0,1]$. Our hardness result also rules out any (absolute) polynomial-time approximation scheme (PTAS). We can show, however, that the problem admits an \textit{asymptotic} polynomial-time approximation scheme (APTAS). These results are in line with classic bin packing.

\begin{restatable}{theorem}{apxhard}\label{thm:np-hard}
The offline \gbp is NP-hard for any $\beta > 0$ and $G \in [0,1]$. Furthermore, unless $P = NP$, there exists no polynomial-time algorithm that approximates the optimal solution with an absolute approximation ratio smaller than $\frac{1.5 + \beta (1 - G)}{1 + \beta (1 - G)}$.
\end{restatable}

\begin{restatable}{theorem}{aptas}\label{thm:aptas}
    For any $\beta > 0$ and $G \in [0,1]$, and any input instance $\sigma$ of the offline \gbp, there exists an APTAS with cost at most $(1+\epsilon)\Opt(\sigma)$, for $\epsilon> 0$.
\end{restatable}

Both results align with the classical bin packing problem. However, to establish the APTAS in \autoref{thm:aptas}, the proof requires additional modifications. In particular, similar to the classical APTAS proof for bin packing, we can show that an APTAS exists when all items are larger than a threshold $\delta$ (for some $\delta \in (0,1)$).%
\shahinadd{This is achieved by first applying a rounding scheme similar to that of de~la~Vega and~Lueker~\cite{Vega81apt}, which groups item sizes into a constant number of classes and replaces each item with the upper bound of its class. 
After rounding, only a constant number of distinct item sizes remain, enabling an enumeration of all feasible bin configurations and the selection of the best packing in polynomial time.}
To handle items smaller than $\delta$, the classical bin packing approach uses a greedy strategy: assign each small item to the first bin it fits in, opening a new bin only when necessary. This retains the APTAS guarantee.
In contrast, the offline \gbp requires more careful handling of small item assignments due to an additional trade-off between opening a new bin and utilizing \black space. We address this by proposing different assignment strategies for the cases $\beta G \le 1$ and $\beta G > 1$, respectively, and show that the APTAS guarantee is preserved in both cases.
See \autoref{Appendix-Offline} for detailed proofs of \autoref{thm:np-hard} and \autoref{thm:aptas}.

\section{Green Bin Packing when $\beta G\leq 1$}
\label{sec:GBleqOne}

As discussed in \autoref{subsec:approaches}, when $\beta G \le 1$ the \black space is cheap, and thus, intuitively, an algorithm should avoid opening a new bin and place items into \black space if possible. This is also aligned with the lower bound of \Opt in \autoref{thm:OptLB}, where $\Opt(\sigma)$ is minimized when filling each bin to its capacity. \shahinadd{To formalize this intuition, we define a family of algorithms parameterized by a threshold $\tau \in [0, 1-G]$. 
Each behaves like its classical counterpart (e.g., \NF, \WF, \AAF) but treats bins as having an effective capacity of $G+\tau$: 
items are packed up to this level, and larger items are placed alone in a bin. 
These algorithms are the main focus of our study for $\beta G > 1$. 
When $\beta G \le 1$, however, increasing $\tau$ always improves performance, so it is best to take $\tau = 1-G$. }

{
\begin{restatable}{theorem}{fillToCapacity}\label{thm:fillToCapacity}
When $\beta G \le 1$, any threshold-based online algorithm with threshold $\tau$ has a competitive ratio no smaller than the same algorithm with threshold $\tau + \epsilon$. %
\end{restatable}
}
In light of the above theorem, we propose to use classical bin packing algorithms: attempt to fill out bins as much as possible before opening a new bin. 
In particular, we focus on the following algorithms: \NF, \WF, \FF, \BF, and \HAR. While these algorithms remain the same as the classic bin packing algorithms, their competitive analysis, however, is more intricate due to the new cost structure of \gbp. In addition to these specific algorithms, we provide a general lower bound on the competitive ratio of any deterministic online algorithm for \gbp when $\beta G \le 1$. We present our main results in \autoref{subsec:results1} and then provide the detailed proofs in \autoref{subsec:results2}.

\subsection{Summary of results for $\beta G \le 1$}
\label{subsec:results1}
We first establish tight bounds on the competitive ratio of \NF and \WF on \gbp. 

\begin{restatable}{theorem}{nextfit}\label{thm:nextfit}
    Given $\beta G \leq 1$, the competitive ratio of \NF for \gbp is exactly $\frac{2+ \beta (1- G)}{\beta + (1-\beta G)}$.
\end{restatable} 

\begin{restatable}{theorem}{worstfit}\label{thm:worstfit}
    Given $\beta G \leq 1$, the competitive ratio of \WF for \gbp is exactly $\frac{2 + \beta (1- 2G)^+}{\beta + (1-\beta G)}$.
\end{restatable}

These results show a clear separation between the competitive performance of \NF and \WF. As seen in \autoref{fig:GBltOne_optimal_cr}, the competitive ratio for \WF is consistently smaller than \NF. The worst-case inputs for both algorithms show that this separation is entirely due to the distribution of black space, not the number of bins opened. In both cases, the algorithm will open twice as many bins as \Opt. However, \NF will waste more green space (use more \black space) than \WF does, resulting in a higher cost. This is a clear departure from classic bin packing, where both algorithms have a competitive ratio of $2$. In \gbp, the bound of $2$ is reached when $G=1$; for all $G<1$, the bounds of \NF and \WF in \gbp are strictly below $2$. Additionally, as we see that for both algorithms, competitive ratio decreases as $\beta G$ approaches $1$, i.e., $\beta G \rightarrow 1^-$. 

\begin{figure}
    \centering
    \includegraphics[width=0.99\linewidth]{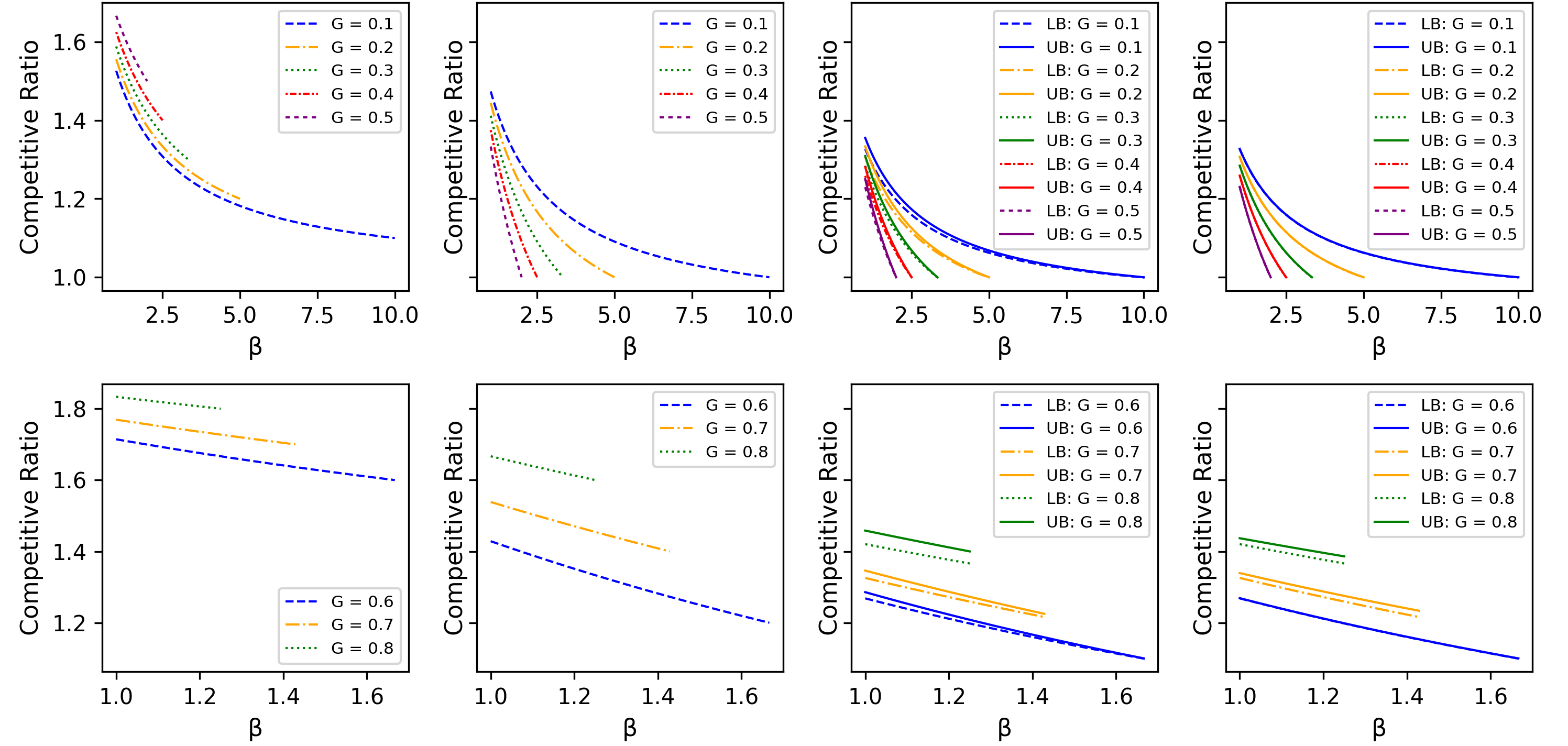}
    \makebox[0.99\textwidth][s]{ 
        \hspace{12mm}
        (a) \NF
        \hspace{14mm}
        (b) \WF
        \hspace{10mm}
        (c) \AAF
        \hspace{10mm}
        (d) \HAR
        \hspace{10mm}
    }
    \caption{Competitive ratio bounds for $\beta G \le 1$. Top row displays values of $G \le \nicefrac{1}{2}$, bottom row displays $G >\nicefrac{1}{2} $ }
    \Description{For $G\le 1/2$, competitive ratio decreases as $\beta \rightarrow 10$. Larger $G$ corresponds to lower competitive ratio for all algorithms except \NF. \WF performs notably better than \NF. \AAF somewhat improves \WF, and \HAR slightly improves \AAF. Similar trends hold for $G > 1/2$, however competitive ratios are generally larger, and do not reach 1 for $\beta G = 1$ as in $G \le 1/2$. \shahin{add a note to explain why some values of $x$ are missing. It will be good to include values on $y$ axis. }}
    \label{fig:GBltOne_optimal_cr}
\end{figure}

For the \AAF and \HAR algorithms, we derive lower and upper bounds on their competitive ratios, and show that the gap between these two bounds is quite small. 

\begin{restatable}{theorem}{anyfitLB}\label{thm:anyfitLB}
    Given $\beta G \leq 1$, the competitive ratios of \AAF (e.g., \BF and \FF) and \HAR are lower bounded by $$\CR \ge (\nicefrac{71}{42} + \beta [(\nicefrac{1}{43} - \nicefrac{G}{42})^+ + (\nicefrac{1}{7} - \nicefrac{G}{6})^+ +(\nicefrac{1}{3}- \nicefrac{G}{2})^+ +(\nicefrac{1}{2}- G)^+])/{(1 + \beta (1 - G))}$$

\end{restatable}

\begin{restatable}{theorem}{anyfitUB}\label{thm:anyfitUB}
Given $\beta G \leq 1$, the competitive ratio of \AAF is upper bounded by
\begin{align}
\label{eq:craaf}
    \CR \leq \begin{cases} 
\frac{7/4 +\beta (1-7G/4)}{1 + \beta(1- G)} &  \text{ if } G \leq \frac{1}{2} \\
\frac{7/4 +\beta(1/2-3G/4)}{1 + \beta(1- G)} & \text { if } \frac{1}{2} < G \leq \frac{2}{3}\\
\frac{7/4}{1 + \beta(1- G)} & \text{ if } G \geq \frac{2}{3}
\end{cases}
\end{align}
\end{restatable}

\begin{restatable}{theorem}{harUB}\label{thm:harUB}
Given $\beta G \leq 1$, the competitive ratio of \HAR is upper bounded by
\begin{align}
\label{eq:crhar}
\CR \leq \begin{cases}
        \frac{1.691 - 1.691 \beta G + \beta}{1+\beta (1-G)} & \text{\ if \ } G \leq 1/2 \\
        \frac{1.691 - 0.691\beta G + \beta /2}{1+\beta(1-G)} & \text{\ if \ } G \in (1/2,2/3] \\ 
        \max\{\frac{1.691 - \beta(1-G)/6}{1+\beta (1-G)}, \frac{1.636 - \beta(1-G)/2}{1+\beta (1-G)}\} & \text{\ if \ } G > 2/3
    \end{cases}
\end{align}
\end{restatable}

 \autoref{thm:anyfitLB}, \autoref{thm:anyfitUB}, and \autoref{thm:harUB} provide nearly tight bounds on \AAF and \HAR, and show that these algorithms have strictly better competitive performance than \NF and \WF (\autoref{fig:GBltOne_optimal_cr}). Again, we see that \CR decreases as $\beta G \rightarrow 1^-$, and is optimal for $\beta G = 1$ when $G\le \nicefrac{1}{2}$. To establish the lower bound, we use an input sequence which gives a lower bound of $\nicefrac{71}{42}$ for these algorithms in the classical setting. Adapting this sequence to the \gbp setting requires a case analysis to capture black cost incurred at different values of $G$. For the upper bound, we define a weighting scheme as described in \autoref{subsec:approaches}. We show that \Opt maximizes its \WPC ratio by filling bins to capacity, and give lower bounds on the \WPC that \AAF and \HAR can achieve based on the values of $\beta$ and $G$,

Finally, we show a general lower bound for the \gbp problem when $\beta G \le 1$, the proof of \autoref{thm:GenLB-gbLtOne} is deferred to \autoref{Appendix-GeneralLB}.

\begin{restatable}{theorem}{genLBgbLtOne} \label{thm:GenLB-gbLtOne}
    Given $\beta G \le 1$, the \CR of any deterministic online algorithm is at least $\frac{r+\beta(1-rG)^+}{1+\beta(1-G)}$, where $r=\nicefrac{248}{161}$ is the general lower bound for classic online bin packing.
\end{restatable}

\subsection{Proofs}
\label{subsec:results2}

In these proofs, we use $(a, b)^n$ to denote a sequence of $n$ item pairs, each consisting of one item of size $a$ and one item of size $b$. Recall that our focus is on analyzing the \textit{asymptotic} competitive ratio, in which the total size of items is arbitrarily large and thus we can ignore the cost from one single bin in both offline and online algorithms.

\shahinadd{First, we prove \autoref{thm:nextfit}, which states the competitive ratio of \NF is exactly $\frac{2+ \beta (1- G)}{\beta + (1-\beta G)}$.}
\subsubsection{Proof of \autoref{thm:nextfit}.}
We aim to show \NF has matching lower and upper bounds. 

\noindent\textbf{Lower bound:} Define an input $\hat{\sigma} := (1, \epsilon)^n$, where $\epsilon>1$ is infinitesimally small. \NF opens $2n$ bins and $n$ of $2n$ bins use the black space. Thus, its cost is $\NF(\hat{\sigma}) = n(2+\beta(1-G))$. Opt uses $n+1$ bins, where $n$ bins are used to pack items with size $1$ and the remaining one is used to pack the small items. In the analysis of asymptotic competitive ratio, the cost from the additional one bin can be ignored, and its cost is $\Opt(\hat{\sigma}) = n (1+\beta (1-G))$. Therefore, the competitive ratio of \NF is at least $\frac{\NF(\hat{\sigma})}{\Opt(\hat{\sigma})} \ge \frac{2+\beta(1-G)}{1+\beta (1-G)}$. 

\noindent\textbf{Upper bound:} Fix an input $\sigma$. Let $S$ be the total size of items in $\sigma$. The number of bins that \NF opens is at most $2S$ because the total size of items in each pair of consecutive bins is more than $1$. Let $S_b$ be the total size placed in its \black space. Then we have $S_b \leq (1-G)S$ because a portion of $G$ in each bin must be packed in green space before switching to \black space.
Thus, the cost of \NF is at most $\NF(\sigma) \le 2S + \beta S_b \le 2S + \beta S (1 - G)$. 
Also, based on \autoref{thm:OptLB} and $\beta G \leq 1$, we have $Opt(S) \geq S (1+\beta(1-G))$.
Therefore, we can have $\frac{\NF(\sigma)}{\Opt(\sigma)} \leq \frac{2 + \beta (1 - G)}{1+\beta(1-G)}$.

\shahinadd{Next, we prove \autoref{thm:worstfit}, establishing the competitive ratio of \WF for \gbp as $\frac{2 + \beta (1- 2G)^+}{\beta + (1-\beta G)}$.}

\subsubsection{Proof of \autoref{thm:worstfit}.}
We start with the case  $G\leq \nicefrac{1}{2}$: \\
\noindent \textbf{Lower bound:}    
Consider the input $\hat{\sigma} = (\nicefrac{1}{2},\epsilon)^n$. \WF opens $n$ bins, each filled up to level $\nicefrac{1}{2}+\epsilon$, and its cost is approximately $\WF(\hat{\sigma}) = n (1 + \beta (\nicefrac{1}{2}-G))$ for small $\epsilon$. \Opt uses $\nicefrac{n}{2}+1$ bins, and its cost is $\Opt(\hat{\sigma}) = (\nicefrac{n}{2}) (1+\beta(1-G))$. Therefore, the competitive ratio is at least $\frac{n (1 + \beta (1/2-G))}{(n/2)(1+\beta(1-G))} = \frac{2 + \beta (1-2G)}{1+\beta(1-G)} = \frac{\beta + 2(1-\beta G)}{\beta + (1-\beta G)}$.

\noindent\textbf{Upper bound:} Fix an input $\sigma$ and let $S$ be the total size of items in $\sigma$. 
\WF fills bins to the level of at least  \nicefrac{1}{2} (for all but at most one bin).
Let $k$ denote the number of bins opened by \WF.
Since $G\leq \nicefrac{1}{2}$, the green space of $k$ bins is fully used and the remaining size of $S - kG$ is placed in \black space. Therefore, we have $\WF(\sigma) = k + \beta (S - k G) = \beta S + k(1-\beta G)$.
Given that all bins of \WF have a level of at least \nicefrac{1}{2}, we have $k\leq 2S$ and $\WF(\sigma) \leq S (\beta + 2(1-\beta G)) = S(2 + \beta (1 - 2G))$.  
Also, since $\beta G \leq 1$, we have $Opt(\sigma) \geq S (1+\beta(1-G))$. Therefore, the competitive ratio of \WF is at most $\frac{S (2 + \beta(1- 2G))}{S (1+\beta(1-G))}  = \frac{2 + \beta(1- 2G)}{ \beta + (1- \beta G)}$.  \\

Next, we consider the case $G > \nicefrac{1}{2}$. 

\noindent\textbf{Lower bound:} 
For the same input $\hat{\sigma} = (\nicefrac{1}{2},\epsilon)^n$, Opt pays the cost of $(\nicefrac{n}{2})(1+\beta(1-G))$, but \WF pays $n$. The competitive ratio is at least $\frac{n}{(n/2)(1+\beta(1-G))} =\frac{2}{\beta+(1-\beta G)}$. 

\noindent\textbf{Upper bound:} Fix an input $\sigma$ and let $S$ be the total size of items in $\sigma$. Let $k_1$ and $k_2$ be the numbers of bins in $\WF$ whose fill level is $(G,1]$ and in $(\nicefrac{1}{2},G]$, respectively. Let $S_1$ and $S_2$ denote the total size of items placed in these groups of bins, respectively. 
We have $\WF(\sigma) = k_1 + k_2 + \beta(S_1 - k_1 G) = k_1(1-\beta G) + k_2 + \beta S_1 $. 
Given that all bins have a level at least \nicefrac{1}{2}, we have $k_1 \leq 2 S_1$ and $k_2 \leq 2 S_2$, and this gives $\WF(\sigma) \leq S_1(\beta +2(1-\beta G)) + 2S_2 $. 

Given that $Opt(\sigma) \geq (S_1+S_2) (1+\beta(1-G))$, the competitive ratio of \WF is at most 
\begin{align*}
    \frac{S_1(\beta +2(1-\beta G)) + 2S_2}{S_1 (1+\beta(1-G)) + S_2 (1+\beta(1-G))} \le 
\max \left\{\frac{\beta +2(1-\beta G)}{ 1+\beta(1-G)}, \frac{2}{1+\beta(1-G)} \right\}  = \frac{2}{1+\beta(1-G)},
\end{align*} 
where we use the fact that $\frac{a+b}{c+d}\leq \max\{\frac{a}{c},\frac{b}{d}\}$, and $\beta + 2(1-\beta G) = 2 - \beta (2G-1) < 2$, which holds when $G>\nicefrac{1}{2}$.

\shahinadd{Next, we prove \autoref{thm:anyfitLB},
establishing a lower bound for competitive ratios of \AAF and \HAR.}
\subsubsection{Proof of \autoref{thm:anyfitLB}.}

Consider the following input, which gives a lower bound of $71/42 \approx 1.6904$ for the competitive ratio of \AAF and \HAR for classic bin packing: \\
\scalebox{.8}{
\begin{minipage}{1.2\textwidth}
$$ \hat{\sigma} = (\underbrace{\nicefrac{1}{43}+\epsilon, \ldots, \nicefrac{1}{43}+\epsilon}_{n \text{ items}},\underbrace{\nicefrac{1}{7}+\epsilon, \ldots, \nicefrac{1}{7}+\epsilon}_{n \text{ items}},\underbrace{\nicefrac{1}{3}+\epsilon, \ldots, \nicefrac{1}{3}+\epsilon}_{n \text{ items}}, \underbrace{\nicefrac{1}{2}+\epsilon, \ldots, \nicefrac{1}{2}+\epsilon}_{n \text{ items}}).$$
\end{minipage}}

Under this input, \Opt of \gbp opens $n$ bins, each containing one item from each size, resulting in a cost of at most $n(1+\beta(1-G))$. On the other hand, \AAF place items of similar classes together. For example, there will be $n/42$ bins containing items of size $1/43+\epsilon$. Each such bin will have a level of approximately $42/43$, and its cost will be $1 + \beta (42/43 - G)^+$. 
With a similar argument, we can summarize the packing decisions and the associated cost by \AAF as follows:
\begin{itemize}
    \item $n/42$ bins, each of cost approximately $ 1 + \beta (42/43- G)^+$ for items of size $1/43+\epsilon$;
    \item $n/6$ bins, each of cost approximately $1+\beta (6/7 - G)^+$ for items of size $1/7+\epsilon$;
    \item $n/2$ bins, each of cost approximately $ 1+\beta (2/3- G)^+$ for items of size $1/3+\epsilon$;
    \item $n$ bins, each of cost approximately $1+\beta (1/2 - G)^+$ for items of size $1/2+\epsilon$. 
\end{itemize}

The total cost of the algorithm would be $n/42+n/6+n/2+n + \beta n((1/43 - G/42)^+ + (1/7 - G/6)^+ +(1/3- G/2)^+ +(1/2- G)^+) = 71n/42 + \beta n((1/43 - G/42)^+ + (1/7 - G/6)^+ +(1/3- G/2)^+ +(1/2- G)^+) $.
We conclude that the competitive ratio of \AAF is at least $\frac{\AAF(\hat{\sigma})}{\Opt(\hat{\sigma})} \ge \frac{71/42 + \beta [(1/43 - G/42)^+ + (1/7 - G/6)^+ +(1/3- G/2)^+ +(1/2- G)^+]}{1 + \beta (1 - G)}$.

\shahinadd{Next, we prove \autoref{thm:anyfitUB}, establishing the upper bound for competitive ratio of \AAF.}
\subsubsection{Proof of \autoref{thm:anyfitUB}.}
\label{sec:aaf-ub}
Here we give a sketch of the proof and defer the full proof to \autoref{Appendix-AAFUB}. 
We start by introducing the modified weighting argument. 
The idea is to define a weight $w(x)$ for any item of size $x$. Given an input $\sigma$, let $W$ denote the total weight of all items of $\sigma$. We aim to show for any $\sigma$, (i) the cost of an online algorithm $\Alg(\sigma) \leq W(\sigma)$ and (ii) $\Opt(\sigma) \geq W(\sigma)/c$; Then this ensures a competitive ratio $c$ for \Alg. 
To establish (ii), we show that the ratio between the total cost of \Opt and the total weight of items in \Opt's packing is at least $1/c$, i.e., $\Opt/W \geq 1/c$, or the \textit{weight per cost} (\WPC) is upper bounded by $W/\Opt \le c$.

\begin{restatable}{proposition}{propWPCboundI}
\label{prop:wpc bound}
  Given $\beta G \le 1$, an input $\sigma$, and a weight function $w(x)$ such that $w(x)/x \ge \beta$, the weight per cost of $\Opt(\sigma)$ is upper bounded by $w^*(\sigma)/(1+\beta(1-G))$, where $w^*(\sigma)$ is the maximum total weight of items in a bin opened by \Opt.
\end{restatable}

Based on \autoref{prop:wpc bound}, to build the upper bound of \WPC, we can switch to bound the maximum weight of a bin, which is in a setting similar to classic bin packing.

We let an item be \textit{large} if it is larger than 1/2 and \textit{small} otherwise. \AAF
has the property that each of its bins, except possibly one, either includes a large item or has a level of at least $2/3$. This is because if a bin is formed only by small items and level less than $2/3$, any subsequent bin formed by small items will include exactly two small items larger than $1/3$. Given this, we split our analysis into two cases: $G \le {2}/{3}$ and $G>{2}/{3}$. 

\begin{restatable}{proposition}{propWPCboundII}\label{prop:smallG}
    Given $G \le 2/3$, $\beta G \le 1$ and an input $\sigma$, if weight function is $w(x) = (3/2(1-\beta G)+\beta)x$ for small items and $w(x) = 1+\beta(x-G)^+$ for large items, then $\Alg(\sigma) \le w(\sigma)$ and $w^*(\sigma) \le \max\{1 + \beta(1/2 - G)^+ + 3(1-\beta G)/4 + \beta/2, 3(1-\beta G)/2 + \beta\}$.
\end{restatable}

\begin{restatable}{proposition}{propWPCboundIII}\label{prop:largeG}
Given $G > 2/3$, $\beta G \le 1$ and an input $\sigma$, if weight function is $w(x) = 3x/2$ for small items and $w(x) = 1+\beta(x-G)^+$ for large items, then $\Alg(\sigma) \le w(\sigma)$ and $w^*(\sigma) \le 7/4$. 
\end{restatable}

Combining \autoref{prop:wpc bound}, \autoref{prop:smallG}, and \autoref{prop:largeG} gives \autoref{thm:anyfitUB}.

\shahinadd{Finally, we prove \autoref{thm:harUB}, establishing an upper bound for competitive ratio of \HAR.}
\subsubsection{Proof of \autoref{thm:harUB}.}
\label{sec:har-ub} Here we give a sketch of the proof and defer the full text to \autoref{Appendix-AAFUB}. We follow the same weight per cost argument as described in \autoref{thm:anyfitUB} and \autoref{prop:wpc bound}. For evaluating \HAR, we use a different weighting scheme to align with the structure of the algorithm. Let the weight of an item of size $x$ from class $i$ (of size in $\in [1/(i+1),1/i]$ to be 

\begin{align}
    w(x) = \frac{1}{i} (1+\beta (ix - G)^+) \label{eq:har_weight}
\end{align}

Note that we have $\nicefrac{w(x)}{x}\ge \beta$ to satisfy the conditions of \autoref{prop:wpc bound}.

\begin{restatable}{proposition}{harUBpropI}\label{prop:har_ub_0}
    Suppose item weights are defined as in \autoref{eq:har_weight}. For all $\beta G \le 1$ and an input $\sigma$, then $\Alg(\sigma) \le w(\sigma)$.
\end{restatable}

\begin{restatable}{proposition}{harUBpropII}\label{prop:har_ub_1}
    Suppose item weights are defined as in \autoref{eq:har_weight}. Given $\beta G \le 1$ and an input $\sigma$, then 
    $$w^*(\sigma) \le \begin{cases}
        1.691-1.691\beta G+\beta & G\le 1/2 \\
        1.691-0.691\beta G+\beta/2 & 1/2<G\le 2/3 \\
        \max\{1.691+\beta(1-G)/6,1.636+\beta(1-G)/2\} & G > 2/3 \\
    \end{cases} $$
\end{restatable}

Combining \autoref{prop:wpc bound}, \autoref{prop:har_ub_0}, and \autoref{prop:har_ub_1} gives \autoref{thm:harUB}.

\section{Green Bin Packing when $\beta G > 1$}
\label{sec:GBgtOne}

When $\beta G > 1$, the unit cost of black space is relatively large, and filling up to the capacity of a bin can be more costly than simply opening up a new bin. Thus, we introduce a class of threshold algorithms for this case. The algorithms take a parameter $\tau \in [0, 1-G]$ and attempt to fill each bin up to a maximum level of $G+\tau \le 1$. Let \TNF, \TWF, \TAAF (including \TFF, \TBF), \GHAR denote the threshold algorithms that proceed identically to their classic counterparts but using the new modified capacity $G + \tau$. Any items with size greater than $G+\tau$ will be placed alone in a bin.

This section derives the upper and lower bounds on the competitive ratio for each of these threshold algorithms. We start by finding the competitive ratio for a given $\tau$, and then choosing the $\tau$ which minimizes \CR in \autoref{subsec:results5-1}. In a departure from the classical behavior, we find that \TNF performs strictly worse than \TWF, and that \TAAF and \GHAR have the same \CR bounds. In \autoref{subsec:results5-3} and \autoref{Appendix-TAAFUB}, we provide the detailed proofs for these results. 

\subsection{Summary of Results for  $\beta G > 1$}\label{subsec:results5-1}

\subsubsection{Tight Bounds of \TNF}

\begin{restatable}{theorem}{tnf}\label{thm:tnf}
Given $\beta G > 1$, the \CR of \TNF for \gbp is exactly $\max\left\{\frac{G(1+\tau \beta)}{G+\tau}, \frac{G(2+\tau \beta)}{G
    +2\tau}, \frac{2+\tau \beta }{1+\tau \beta } \right\}.$
\end{restatable}

The tight bounds of \TNF allow us to uniquely determine the optimal $\tau$ and the optimal \CR.

\begin{restatable}{corollary}{tnfOPT}\label{cor:tnfOPT}
Let $\tilde{\tau} =\frac{2-\beta G + \sqrt{5\beta^2 G^2 - 8\beta G + 4}}{2\beta(\beta G-1)}$. The optimal threshold $\tau^*$ is given by $\min\{\tilde{\tau},1-G\}$ when $1<\beta G \le 2$, by $\min\{\sqrt{G/\beta},1-G\}$ when $2 <\beta G < 4$, and by $0$ when $\beta G \ge 4$. The associated optimal \CR of \TNFop is given by
\begin{align}
\label{eq:nextfittau}
\CR =  
\begin{cases}
        \frac{3\beta G-2 +\sqrt{5\beta^2G^2-8\beta G +4}}{\beta G +\sqrt{5\beta^2G^2-8\beta G +4}} & 1 < \beta G \le 2 \text{ and } \tilde{\tau} < 1-G \\
        \frac{2+\beta(1-G)}{1+\beta(1-G)} & 1 < \beta G \le 2 \text{ and } \tilde{\tau} \ge 1-G \\
        \sqrt{\beta G} & 2 <\beta G  <4 \text{ and } \sqrt{G/\beta} < 1-G\\
        \frac{2G+\beta G(1-G)}{2-G} &
        2<\beta G < 4 \text{ and } \sqrt{G/\beta} \ge 1-G\\
        2 & \beta G \ge 4
    \end{cases}    
\end{align}
\end{restatable}

As is the case for $\beta G \le 1$, the \CR of \TNF can reach $2$; indeed it can reach this maximum \CR on non-extremal values of $\beta, G$. This contrasts to the remaining algorithms, which can only reach maximal \CR on extremal values of $\beta, G$. Next, we show the lower bound results for \TAAF, \TWF and \GHAR.

\subsubsection{Lower Bounds of \TAAF, \TWF and \GHAR}
\begin{restatable}{theorem}{tafLB}\label{thm:taf-lb}
 Let $\beta G > 1$. When $\tau <1 -G$, the \CR of \TAAF, \TWF and \GHAR ($K\ge 2$) is lower bounded by 
 $\max\{\frac{2}{1+\tau\beta}, 1+\frac{(G-\tau)(1+\tau\beta)}{2(G+\tau)},\frac{G(1+\tau\beta)}{G+\tau}\}$. When $\tau = 1 -G$ these algorithms are lower bounded by $\max\{\zeta(\beta,G), 1+\frac{(2G-1)(1+\beta(1-G))}{2},G(1+\beta(1-G))\}$, where $\zeta(\beta,G)$ is the lower bound presented in \autoref{thm:anyfitLB}. 
\end{restatable}
The above lower bounds are derived from three types of worst-case inputs, under which \Alg and \Opt fill bins using different strategies. This is in contrast to the lower bound construction in the classic bin packing and the \gbp when $\beta G \le 1$, which mainly rely on one type of worst-case input. The input strategy producing $\frac{2}{1+\tau\beta}$ does not occur when $\tau = 1-G$, but the lower bound input in \autoref{thm:anyfitLB} exploits a similar weakness. 
We now shift to the upper bounds for \TAF and \GHAR.

\subsubsection{Upper bounds of \TAAF, \GHAR and \TWF}

\begin{restatable}{theorem}{TWFUB}\label{thm:TWFUB}
    Given $\beta G >1$, the \CR of \TWF is upper bounded by         $\max\{\frac{2G}{G+\tau},\frac{G(1+\tau\beta)}{G+\tau}\}.$
\end{restatable}
By setting $\tau = \min\{\frac{1}{\beta}, 1-G)$, the \CR of \TWF is upper bounded by $\max\{\frac{2\beta G}{\beta G+1}, 2G\}$. As expected from classical bin packing, we can establish stronger bounds for \TAF and \GHAR than for \TWF.

\begin{restatable}{theorem}{TAAFUB}\label{thm:taf-UB} 
Given $\beta G >1$, the \CR of \TAAF and \GHAR ($K\ge 3$) is upper bounded by
    $ {\max}\{\frac{2}{1+\tau\beta},\frac{7G+\tau}{4(G+\tau)}, 1+\frac{(G-\tau)(1+\tau \beta)}{2(G+\tau)}, \frac{G(1+\tau\beta)}{G+\tau}\}$ when $\tau <1-G$. If $\tau = 1-G$, then the \CR is upper bounded by
    $ {\max}\{\frac{6G+1}{4},1+\frac{(2G-1)(1+\beta(1-G))}{2}, G(1+\beta(1-G))\}$. 
\end{restatable}

We distinguish between optimal $\tau$ (over tight bounds) and optimized $\tau$ (over loose bounds).

\begin{restatable}{corollary}{optimalTauOne}\label{Cor:optimalTauOne} The choice of $\tau$ will depend on the specific $\beta, G$ values according to three cases:

\begin{itemize}
    \item If $1 \le \beta(1-G)$, the optimized $\tau$ for \TAAF and \GHAR is $\tau = \frac{1}{\beta}$ when $\beta G \le \frac{7}{4}+\frac{\sqrt{57}}{4} (\approx 3.637)$ and $\tau = \frac{1}{2\beta}$ otherwise, giving \CR upper bounds of $\min\{\frac{2\beta G }{\beta G + 1},\frac{14\beta G + 1}{8\beta G + 4})$. Furthermore, $\tau = 1/\beta$ is optimal when $\beta G \le 1+\sqrt[3]{2+\sqrt{44/27}}+\sqrt[3]{2-\sqrt{44/27}} (\approx 3.383)$.
    \item If $1/2\beta \le 1-G < 1/\beta$, the optimized threshold for \TAAF and \GHAR is either $\tau = 1-G$ or $\tau = 2/\beta$, and the \CR is upper bounded by $\min \{ \frac{14\beta G+1}{8\beta G+4},1+\frac{(2G-1)(1+\beta(1-G))}{2} \}$. 
    \item If $1 > 2\beta(1-G)$, the optimized threshold for \TAAF and \GHAR is $\tau = 1-G$, and the \CR is upper bounded by $\frac{6G+1}{4}$.
\end{itemize}
 
\end{restatable}

\noindent The upper bounds of \TAAF and \GHAR are close to the lower bounds derived in \autoref{thm:taf-lb}. 
This also allows us to identify the selection of threshold $\tau$. In particular, we can observe that the best candidate $\tau$ that minimizes the upper bounds can be  $\frac{1}{\beta}$, $\frac{1}{2\beta}$, or $1-G$. As $\beta G \to \infty$, the \CR of \TAAF or \GHAR approaches $1.75$. In addition, when $G \le \nicefrac{1}{2}$ and $\beta G$ is sufficiently close to $1$, $\tau = \frac{1}{\beta}$ is the optimal design that can attains matching upper and lower bounds. Note that this guarantee does not extend to \TWF.

In summary, the upper bound results broadly follow the same trends as in classical bin packing and in the \gbp case of $\beta G \le 1$: we find that \TNF is worst, \TWF is in the middle, and \TAAF and \GHAR algorithms are best. Again, we find a clear separation between $\TNF$ and $\TWF$ which does not appear in the classical setting. In addition, \TAAF and \GHAR attain the same \CR bounds. It is also notable that the parameter $K$ in \GHAR is not relevant to the competitive ratio bounds (following our analysis)  so long as it is at least three. 

Finally, we show a general lower bound for the \gbp problem when $\beta G > 1$, the proof is given in \autoref{Appendix-GeneralLB}.

\begin{restatable}{theorem}{generalLB} \label{Thm:GenLB_BGgt1}
    Given $\beta G > 1$, the \CR of any deterministic online algorithm is at least \\ $\max\{\frac{r+1-G}{1+\beta(1-G)}, f(\beta G)\}$, where $r$ is the classical lower bound (\nicefrac{248}{161}) and $f(\beta G)$ is $\frac{3(\beta G+1)}{\beta G+5}$ for $ 1 < \beta G \le \nicefrac{3}{2}$ , $\frac{3(\beta G+3)}{\beta G+11}$ for $\nicefrac{3}{2} < \beta G \le 3$, $\nicefrac{9}{7}$ for $3 < \beta G \le 4$, $\nicefrac{4}{3}$ for $4 < \beta G \le 48$, and $\nicefrac{3}{2}$ when $\beta G > 48$.
\end{restatable}

\begin{figure}
    \centering
    \includegraphics[width=0.99\linewidth]{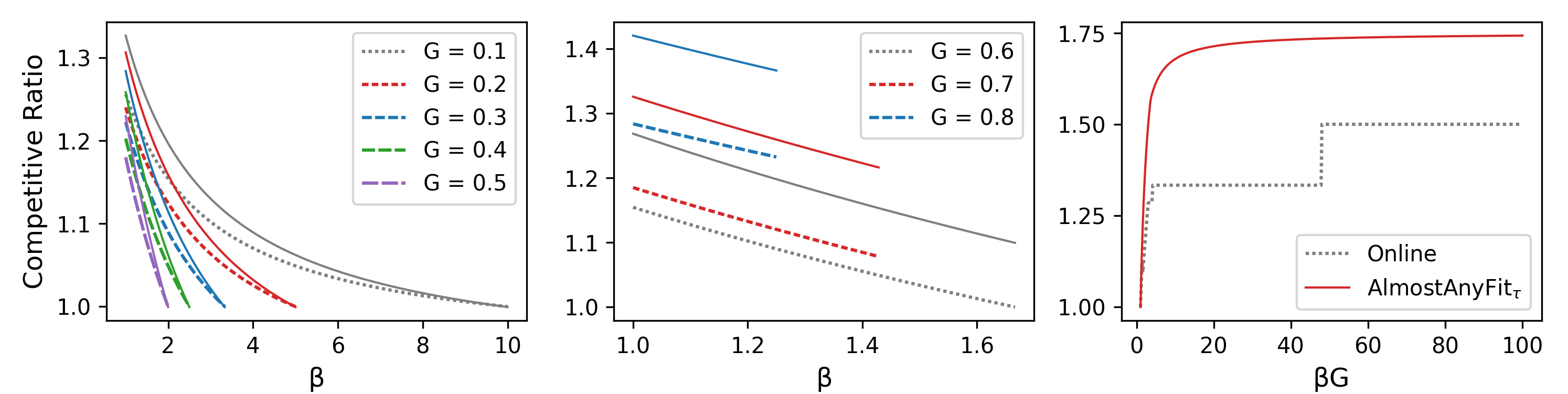}
    \makebox[0.99\textwidth][s]{ 
        \hspace{18mm}
        (a) $G\le\nicefrac{1}{2}$
        \hspace{29mm}
        (b) $G>\nicefrac{1}{2}$
        \hspace{25mm}
        (c) $1< \beta G \le 100$
        \hspace{10mm}
    }
    \caption{General Online Lower Bound vs \AAF/\HAR Lower Bound for $\beta G \le 1$ (a and b). The general online bound is shown with dotted lines, while \AAF/\HAR is shown with solid lines. And (c) gives a comparison of the General Online Lower Bound to the upper bound of \TAAF for $\beta G > 1$ and $G\le 1/2$.}
    \label{fig:general_onlineLB}
    \Description{For $G\le 1/2$, the general lower bound decreases  to 1 as $\beta G \rightarrow 1$. The lower bound for \AF begins slightly above the general lower bound, but they converge as $\beta G \rightarrow 1$. For $G>1/2$, the bounds are still decreasing in $\beta$, but do not reach 1. The gaps between the general lower bound and \AF is larger and does not decrease with increasing $\beta$. For $1< \beta G\le 100$, the general lower bound is a stepwise function which nearly matches \AAF for $\beta G < 5$, but has a much larger gap for increased $\beta G$}
\end{figure}

In \autoref{fig:general_onlineLB}, we illustrate the general lower bounds and compare them to the lower bounds derived in \autoref{thm:anyfitLB} and \autoref{Cor:optimalTauOne}. For $\beta G \le 1$ and $G \le \nicefrac{1}{2} $ the lower bounds of our algorithms are close to the online optimal, and this gap decreases as $\beta G$ approaches 1. When $\beta G \le 1$ and $G > \nicefrac{1}{2}$, there is a larger gap between the lower bound of our algorithms and the general lower bound which is constant across $\beta G$. Finally, when $\beta G>1$, \TAAF is nearly optimal for small $\beta G$ (at most 1.5), but a larger gap exists for $\beta G$ outside this range.

\subsection{Proofs}\label{subsec:results5-3}

In this section, we give proofs for the lower bounds of each algorithm, as well as the derivation of optimal $\tau$ from the bounds. For the upper bounds, we provide sketches of the arguments used for \TNF, \TAAF, and \GHAR, and defer the remaining proofs to \autoref{Appendix-TAAFUB}.

\subsubsection{Proof of \autoref{thm:tnf}} 
\shahinadd{First, we prove \autoref{thm:tnf}, establishing competitive ratio of \NF. }
We consider the following worst-case inputs for the lower bound: 

\noindent\textbf{Lower Bound:}  
    \begin{itemize}
        \item  Consider an input formed only by items of size $\epsilon$ and total size $S$ for some larger $S$.
        The cost of \TNF is simply $\frac{S}{G+\tau} (1+\beta \tau)$ while that of \Opt is $S/G$, giving a lower bound of $\frac{(S/(G+\tau))(1+\beta \tau)}{S/G} = \frac{G(1+\beta \tau)}{G+\tau}$ for the 
        competitive ratio. 

        \item Consider the input $(\epsilon, (G+\tau))^n$. The cost of the algorithm is $n(2+\tau \beta)$ while that of \Opt is $n(1+ \tau \beta)+1$. For asymptotically large $n$, the lower bound would be $\frac{2+\tau \beta}{1+\tau \beta }$. 

        \item Consider the input $(\tau, \epsilon, G, \epsilon^m)^n$ such that the $m$ items of such $\epsilon$ have total weight $\tau$. The cost of the algorithm will be $n (2+\alpha \beta G)$, while that of \Opt will be $n + n (2\tau )/ G = n(1+2\tau/G)$ (\Opt fills every bin with an item of size $G$ or an item of size $\tau$ along with items of size $\epsilon$ with total size $(G-\tau)$. Therefore, the competitive ratio is lower-bounded by $\frac{2+\tau \beta}{2\tau/G + 1} = \frac{G(2+\tau \beta)}{G+2\tau}$. 
    \end{itemize}

\noindent\textbf{Upper Bound (Sketch).}
The detailed proof is provided in Appendix~\ref{Appendix-TAAFUB}; here we summarize the main argument. 
We begin by partitioning the final packing of \TNF\ into four types of bin pairs, depending on whether consecutive bins are filled up to~$G$ (green) or exceed~$G$ (black). 
Specifically, we distinguish between \texttt{GG} pairs (two consecutive green bins), \texttt{GBR} pairs (a green bin followed by a black bin opened by an item of size $\leq G$), \texttt{GBS} pairs (a green bin followed by a black bin opened by a ``special'' item, i.e., an item of size $>G$), and single \texttt{B} bins (isolated black bins). 
Each type contributes differently to the total cost, depending on the amount of green and black space it occupies.

For each type~$X \in \{\texttt{GG}, \texttt{GBR}, \texttt{GBS}, \texttt{B}\}$, let $S_X$ denote the total size of items packed in bins of that type, $A_X$ the total unused green space (waste), and $B_X$ the total black space used. 
Let $m$ denote the number of special items in \texttt{GBS} bins, and suppose their total size is $mG + \Gamma$ for  $\Gamma > 0$.

The cost of \TNF\ can then be expressed as a weighted sum of these quantities:
\[
\text{ALG} = 
   \frac{S_{\texttt{GG}} + A_{\texttt{GG}}}{G}
 + \frac{S_{\texttt{GBR}} + A_{\texttt{GBR}} - B_{\texttt{GBR}}}{G} + \beta B_{\texttt{GBR}}
 + \frac{S_{\texttt{B}} - B_{\texttt{B}}}{G} + \beta B_{\texttt{B}}
 + 2m + \beta(\Gamma + B_{\texttt{GBS}}),
\]
The optimal solution, in contrast, must use at least $\tfrac{S_{\texttt{GG}} + S_{\texttt{GBR}} + S_{\texttt{GBS}} + S_{\texttt{B}}}{G}$ bins plus $m$ additional bins for the special items, with a total black-space cost of~$\beta\Gamma$.

Taking the ratio of these two quantities yields the expression in~\autoref{exp:nftUB}, which can be upper-bounded by the maximum of four simpler terms, each corresponding to one bin type. 
Bounding each term in terms of~$\tau$, $\beta$, and~$G$ gives the inequalities summarized in \autoref{prop:tnf_ub_2}, from which the upper bound in the theorem follows.

\shahinadd{Next, we prove \autoref{thm:taf-lb}, establishing a lower bound for competitive ratios of \WF, \AAF, and \HAR.}
\subsubsection{Proof of \autoref{thm:taf-lb}}

\textbf{Case 1: $\tau <1-G$}. The lower bounds for these algorithms are derived using instances which capture three scenarios. First, when $\beta\tau$ is relatively small compared to the cost of a new bin, the worst case is to fill $\Alg$ to $\frac{G+\tau}{2}$ and $\Opt$ to $G+\tau$. Next, if $\beta\tau$ is comparable to the cost of a new bin: $\Opt$ will fill to $G$, and $\Alg$ will fill some bins to $G+\tau$ and some to $\frac{G+\tau}{2}$. Finally, when $\beta\tau \ge 1$, the worst case is to fill $\Alg$ to $G+\tau$ and $\Opt$ to $G$. These scenarios give three separative functions for the competitive ratio based on $\tau$. The lower bound is the maximum of these three functions over $\tau\in [0,1-G]$. 

We define $\hat{\tau}$ as the smallest positive root of $\beta^2\tau^3+\tau^2(\beta^2G)+\tau(4\beta G-3)-G$. Note that this is the intersection point between our first two lower-bound functions. 

\textit{Case 1.1:} $\tau \le \hat{\tau}$. For asymptotically small $\epsilon > 0$, consider the input: $(\frac{G+\tau+\epsilon}{2})^n$. In this case, the \Opt will go above the $G+\tau$ threshold by placing two items per bin, for a cost $\frac{n}{2} + \frac{n}{2}\beta(\tau+\epsilon)$. Meanwhile $\Alg$ will only fit one item per bin, for cost $n$, $\Alg$ will not incur any \black cost since $\frac{G+\tau+\epsilon}{2} < G$. Since $\epsilon$ is asymptotically small, we can then lower bound the competitive ratio as
$\CR \ge \frac{2}{1+\tau\beta}$.

\textit{Case 1.2:} $\hat{\tau} < \tau < \frac{1}{\beta}$.  Consider the sequence $((\epsilon)^E,( \frac{G+\tau}{2}+\epsilon)^F)$. We can pick $E = (\frac{G-\tau}{2\epsilon})F$ so that each bin of \Opt will have one item of size $\frac{G+\tau}{2}+\epsilon$ and many of size $\epsilon$, and be filled to $G$. Thus the cost of the \Opt will just be the number of bins, $F$.
Meanwhile, $\Alg$ will consist of bins filled with $\epsilon$ up to $G+\tau$, and then bins with a single item of size $\frac{G+\tau}{2}+\epsilon$. This is because \TAAF algorithms will not open new bins while there is available space, and because \GHAR will place the $\epsilon$'s in type $K> 1$ bins while the $\frac{G+\tau}{2}+\epsilon$ are placed in type $1$ bins. The cost of \Alg will be $F + \frac{E\epsilon}{G+\tau}+\beta\tau(\frac{E\epsilon}{G+\tau})$. We can substitute to find that cost is $F+F\frac{G-\tau}{2(G+\tau)}(1+\beta\tau)$. This gives $\CR \ge 1+\frac{G-\tau}{2(G+\tau)}(1+\beta\tau)$.

\textit{Case 1.3:} $\tau \ge \frac{1}{\beta}$. Consider the sequence of items $(\epsilon)^n$ with total volume $S$. In this case, $\Opt$ will fill each bin exactly to $G$, for a cost $S/G$. $\Alg$ will fill each bin exactly to $G+\tau$, for a cost $\frac{S}{G+\tau}+\frac{S}{G+\tau}\tau\beta$. Therefore, $\CR \ge \frac{G(1+\tau\beta)}{G+\tau}$. 

\textbf{Case 2: $\tau = 1-G$ }. In this scenario, the argument from Case 1.1 is not possible, as \Opt cannot fit two items of size $\frac{G+\tau+\epsilon}{2}$ into a single bin. We instead reuse the bound presented in \autoref{thm:anyfitLB}, which still gives a lower bound for certain $\beta, G$ in this setting, even when it is no longer optimal to fill to capacity 1. Note this bound is never maximal when $\tau < 1-G$. Meanwhile, the input sequences seen in Cases 1.2 and 1.3 still hold, and we can evaluate those functions at $\tau=1-G$.

\shahinadd{Next, we establish \autoref{thm:taf-UB}, which gives an upper bound for the competitive ratios of \AAF and \HAR.}
\subsubsection{Proof of \autoref{thm:taf-UB}}
We provide a proof sketch for the upper bound and the full proof is deferred to \autoref{app:taaf-ub}. 
We still rely on the weighting argument; however, unlike the upper bound proof when $\beta G \le 1$ (as shown in \autoref{sec:aaf-ub}), we do not simply upper bound the  cost of \Alg by its weight. Instead, we derive both upper bounds on the \WPC of \Opt (i.e., $W(\sigma)/\Opt(\sigma)$), and lower bounds on the \WPC of \Alg (i.e., $W(\sigma)/\Alg(\sigma)$) for any input $\sigma$. Since the total weight $W(\sigma)$ is the same for both \Opt and \Alg, the ratio of the two bounds gives the competitive ratio. 

In this case, we define the weight of an item of size $x$ as
\begin{align}
\label{eq:weight}
w(x) = 
\begin{cases}
    x & x\le \nicefrac{G+\tau}{2}\\
    x+R & x > \nicefrac{G+\tau}{2}
    \end{cases}, \quad \text{for a given parameter}\ R \in\left[0,\nicefrac{G+\tau}{6}\right].   
\end{align}

We can build lower bounds and upper bounds of \WPC for the online algorithms and the offline optimal, respectively, in the following two lemmas.
\begin{restatable}{lemma}{taafUBlemI}\label{lem:ALG-WPC}
Given $\beta G > 1$ and the weighting function $w(x)$ given in Equation~\eqref{eq:weight}, define $\mu_1:=\mu_1(R) = \frac{(1-G\beta-2R\beta)+\sqrt{(G\beta+2R\beta-1)^2+4\beta(G-2R)}}{2\beta}$. The \WPC of \TAAF or \GHAR is at least $\frac{G+\tau + 2R}{2}$ when $\tau \le \mu_1(R)$ and at least $\frac{G+\tau}{1+\tau \beta}$ when $\tau \ge \mu_1(R)$.
\end{restatable}

\begin{restatable}{lemma}{taafUBlemII} \label{lem:opt-WPC}
Given $\beta G > 1$ and the weighting function $w(x)$ given in Equation~\eqref{eq:weight}, define $\mu_2 := \mu_2(R) = \frac{R}{\beta(G+R)-1}$. If $\tau < 1-G$, \WPC of \Opt is at most $\frac{G+\tau+2R}{1+\tau \beta}$ if $\tau \le \mu_2$, and at most $G+R$ if $\tau > \mu_2$. Meanwhile, if $\tau = 1-G$, the \WPC of \Opt is at most $G+R$.

\end{restatable}

Furthermore, we can show that $\mu_1(R) > \mu_2(R)$. 
\begin{restatable}{lemma}{weightingGroup}\label{lem:weighting-group-1}
If $\beta G > 1$, then $\mu_1(R) > \mu_2(R)$ for all $R \in [0, \frac{G+\tau}{6}]$.
\end{restatable}

Using this weighting function and above lemmas,
we construct upper bounds ($\tau <1-G$ on left, $\tau=1-G$ on right) which hold for any given value $R\in [0, \frac{G+\tau}{6}]$. 

$$CR \le \begin{cases}
    \frac{\frac{G+\tau+2R}{1+\tau\beta}}{\frac{G+\tau+2R}{2}} = \frac{2}{1+\tau \beta} & 0 \le \tau \le \mu_2 \\
    \frac{G+R}{\frac{G+\tau+2R}{2}} = \frac{2(G+R)}{G+\tau+2R} & \mu_2 \le \tau \le \mu_1 \\
    \frac{G+R}{\frac{G+\tau}{1+\tau\beta}} = \frac{(G+R)(1+\tau\beta)}{G+\tau} & \mu_1 \le \tau \le 1-G \\
\end{cases}, \text{ and } CR \le \begin{cases}
    \frac{G+R}{\frac{1+2R}{2}} = \frac{2(G+R)}{1+2R} & 0\le  1-G \le \mu_1 \\
    \frac{G+R}{\frac{1}{1+\tau\beta}} = \frac{(G+R)(1+\beta(1-G))}{1} & \mu_1 \le 1-G \\
\end{cases}$$

When $R = 0$, we have $\mu_1 = 1/\beta$, and thus $\CR \le \frac{G(1+\tau \beta)}{G + \tau}$ for $\tau \in [1/\beta, 1-G]$. This matches the third term in the lower bound in \autoref{thm:taf-lb}.

As $R$ increases from $0$ to $\frac{G+\tau}{6}$, $\mu_1(R)$ decreases from $\frac{1}{\beta}$ to $\frac{1}{2\beta}$, and $\CR \le 1 + \frac{(G-\tau)(1+\tau\beta)}{G+ \tau}$ for $\tau = \mu_1(R) \in [\frac{1}{2\beta},\frac{1}{\beta}]$. This also matches the second term in the lower bound.

When $R = \frac{G+\tau}{6}$, we have $\mu_1 = \frac{1}{2\beta}$, and $\frac{2(G+R)}{G+\tau + 2 R} =  \frac{7G + \tau}{4(G+\tau)}$. Thus, we have $\CR \le \max\{\frac{2}{1+\tau\beta}, \frac{7G + \tau}{4(G+\tau)})$, which is close to the first term in the lower bound. Note that $\frac{2}{1+\tau\beta}$ is only possible if $\tau <1-G$.

Summarizing all above cases gives the upper bounds in \autoref{thm:taf-UB}.

\section{Experiments}
\label{sec:experiments}
In this section, we evaluate the numerical performance of each of the algorithms discussed in both settings of $\beta G \le 1$ and $\beta G > 1$.

\subsection{Experimental Setup}
We evaluate the algorithms when items are drawn from a Weibull distribution~\cite{angelopoulos2024BPpredictions}, which was applied to bin packing by \cite{Castiñeiras2012} in the context of data center job scheduling. Following the previous work by \cite{angelopoulos2024BPpredictions}, we generate items using a shape parameter of $3.0$.
In \autoref{Appendix-Experiments}, we evaluate the algorithms under two additional input distributions: continuous uniform $[0,1]$, the GI benchmark in the BPPLIB library~\cite{Gschwind2016,Dolorme2018}. All experiments show similar trends.
Due to the complexity of computing the offline optimal bin packing solution, we follow the convention of evaluating the algorithms' performance by \textit{empirical competitive ratio}, which is the ratio of the empirical cost from an online algorithm and lower bound of offline optimal algorithm, for a given input sequence. 
For each experiment, we test the specific $\beta$ and $G$ on an input with $3000$ items, and present the average empirical competitive ratio.

\subsection{Experimental Results}

\noindent\textbf{Results for $\beta G \le 1$.}
In Figure~\ref{fig:unif_weibullleqOne}, we present the empirical competitive ratios of \NF, \WF, \AAF (including \FF and \BF), and \HAR (with $K = 10$).
The results of these simulations are in line with our theoretical results and with empirical patterns in classical bin packing. First, we see that when $\beta G = 1$ and $G \le \nicefrac{1}{2}$, as in the $\beta=4$ case, \BF and \FF \shahinadd{achieve performance closest to the optimum.}%
This mirrors our competitive ratio bounds. Additionally, we see a separation of the algorithms into three groups by performance. The worst performing group contains \NF and \HAR. \WF itself makes up the middle group; while \FF and \BF are near optimal in practice. This is aligned with empirical performance in classic bin packing, where \HAR does poorly despite having the best theoretical competitive ratio \cite{Bentley1984}\cite{Lee1985Har}.

\begin{figure}[t]
    \centering
    \includegraphics[width=0.9\linewidth]{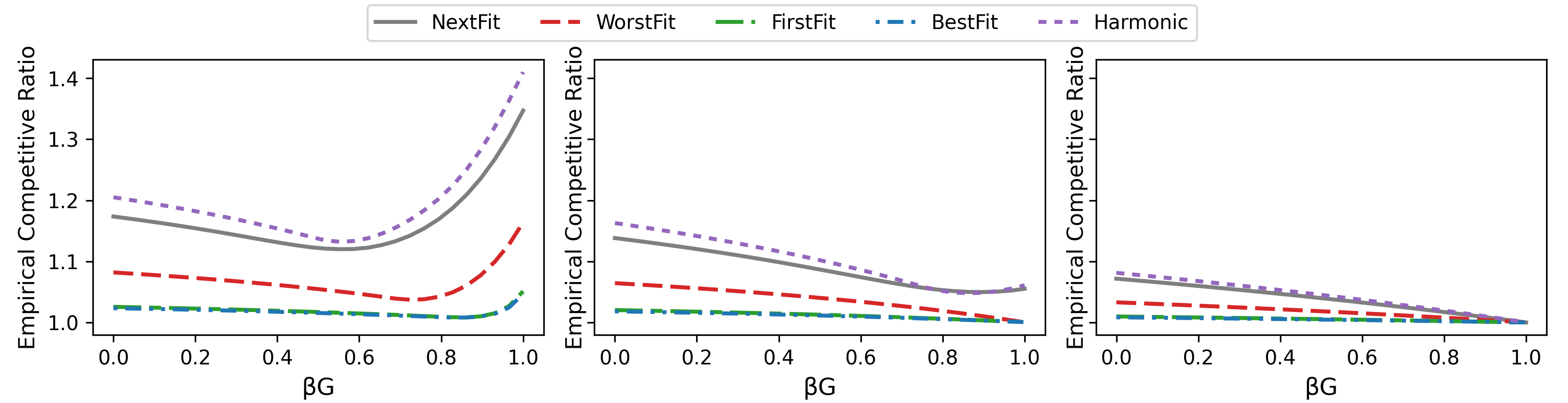}
    \makebox[0.99\textwidth][s]{ 
        \hspace{16mm}
        (a) $\beta = 1$
        \hspace{27mm}
        (b) $\beta = 3/2$
        \hspace{27mm}
        (c) $\beta = 4$
        \hspace{10mm}
    }
    \caption{Experiments on Weibull distribution with fixed $\beta$ for $\beta G \le 1$}
    \Description{\NF and \HAR perform worse than \WF, which is worse than \FF and \BF. As $\beta$ increases, the competitive ratio for all algorithms decreases at any given $\beta G$. For $\beta = 1$, the competitive ratio increases as $\G \rightarrow 1$, most dramatically for the worst performing algorithms.}
    \label{fig:unif_weibullleqOne}
\end{figure}

The most interesting observation is the uptick in competitive ratio for large $\beta G$, as seen in the $\beta=1$ and $\beta = 3/2$ plots. 
This is due to wasted green space as $G$ becomes large. 
For a given bin, this uptick compared to $\Opt$ will occur when the fill level $L < \beta G$. Since \NF and \HAR achieve lower fill levels, they waste more green space relative to $\Opt$. Notably, these upticks disappear when $\beta \ge 2$ because each algorithm will be able to fill every bin to at least $\nicefrac{1}{2}$, and when $\beta \ge 2$, $G \le \nicefrac{1}{2}$. Thus, no green space will be wasted by any of the algorithms.

\noindent\textbf{Results for $\beta G > 1$.}
In this case, we have the added complexity of selecting the value of $\tau$ to use for each algorithm. 
In \autoref{fig:compare_tau_weibull}, we evaluate the algorithms with both theoretical $\tau$ as derived in \autoref{sec:GBgtOne} and empirically-inspired $\tau$ (See details in \autoref{Appendix-Experiments}). 
We find that there is a trade-off between worst-case bounds and practical performance. Each algorithm except for \GHAR appears to perform better with a smaller $\tau$ than that suggested by the theory.
We can observe that the performances of the algorithms using the theoretical $\tau$ remain divided into three groups, with \TWF and \GHAR performing the best, \TNF performing the worst, and \TFF and \TBF performing as bad as \TNF for small $\beta G$ but as well as \TWF for large $\beta G$. The performance gain of \TWF is mainly due to the load-balancing nature of this algorithm, which prevents the use of costly black space for bins with effective capacity $G + \tau$. However, this advantage can be mitigated by a proper selection of $\tau$, as shown in \autoref{fig:compare_tau_weibull}(b). This motivates the design of data-driven algorithms for threshold-based \gbp, which represents a promising future direction.
Another interesting observation is that the performance of algorithms using either the theoretical or empirical $\tau$ first increases and then decreases as $\beta G$ grows. Bounding the performance when $\beta G$ is not at the extremes ($\beta G \to 1$ or $\beta G \to \infty$) appears to be a particularly critical and interesting future direction.

\begin{figure}[t]
    \centering
    \includegraphics[width=0.9\linewidth]{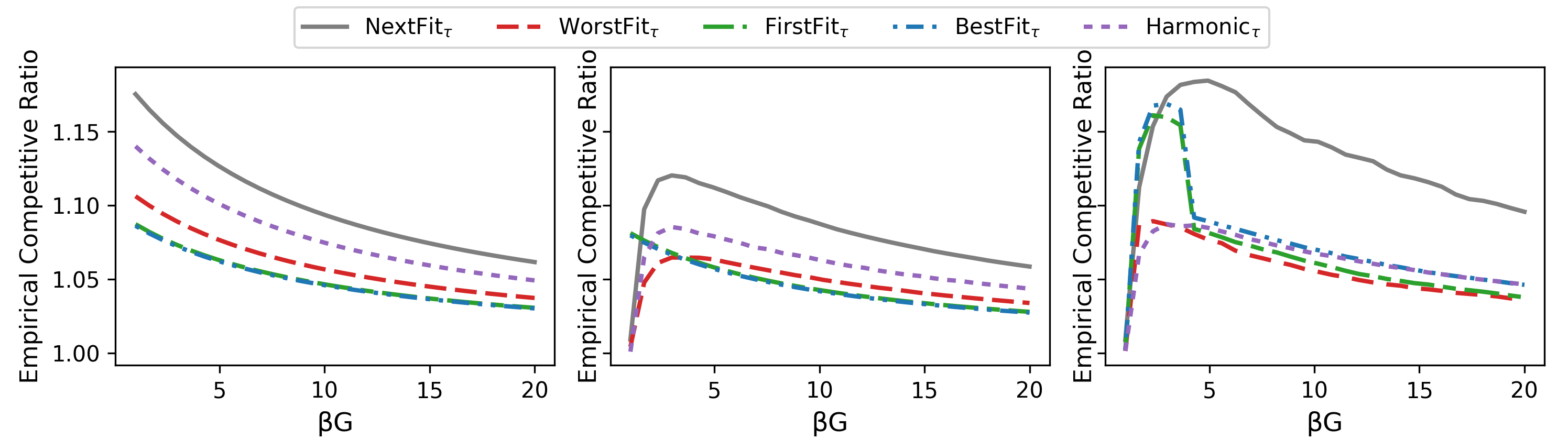}
    \makebox[0.99\textwidth][s]{ 
        \hspace{20mm}
        (a) $\tau = 0$
        \hspace{28mm}
        (b) Empirical $\tau$
        \hspace{23mm}
        (c) Theoretical $\tau$
        \hspace{10mm}
    }
    \caption{Experiments on Weibull distribution with $G=\nicefrac{1}{2}$. (a): All $\tau=0$ (b): All $\tau$ set to empirically best $\tau$ (c): All $\tau$ set to theoretical best $\tau$ for given $\beta G$}
    \Description{For $\tau=0$ all algorithms are decreasing in $\beta G$. For the empirical and theoretical $\tau$, performance is increasing for small $\beta G$, but decreases once $\beta G > 5$.}
    \label{fig:compare_tau_weibull}
\end{figure}

To further demonstrate the potential improvement from empirical thresholds, we evaluate all our threshold algorithms in \autoref{fig:WeibullgeOne} by fixing $G \in \{1/2, 3/4, 19/20\}$ and selecting $30$ values of $\beta$ from $1.005/G$ to $20/G$.
Here, we see that the value of $G$ impacts both the absolute performance and the degree of separation between algorithms. For $G = \nicefrac{1}{2}$, none of the algorithms waste green space, leading to less separation among the algorithms overall. As $G$ increases, the grouping of \NF and \HAR, \WF, and \FF and \BF becomes more pronounced, and the absolute empirical competitive ratio generally increases.
We can also observe that as the green space $G$ increases, the load-balancing advantage of \TWF in using green space becomes weaker, and \TBF and \TFF perform much better than \TWF. Notably, \GHAR can still perform as poorly as \TNF.

\begin{figure}[t]
    \centering
    \includegraphics[width=0.899\linewidth]{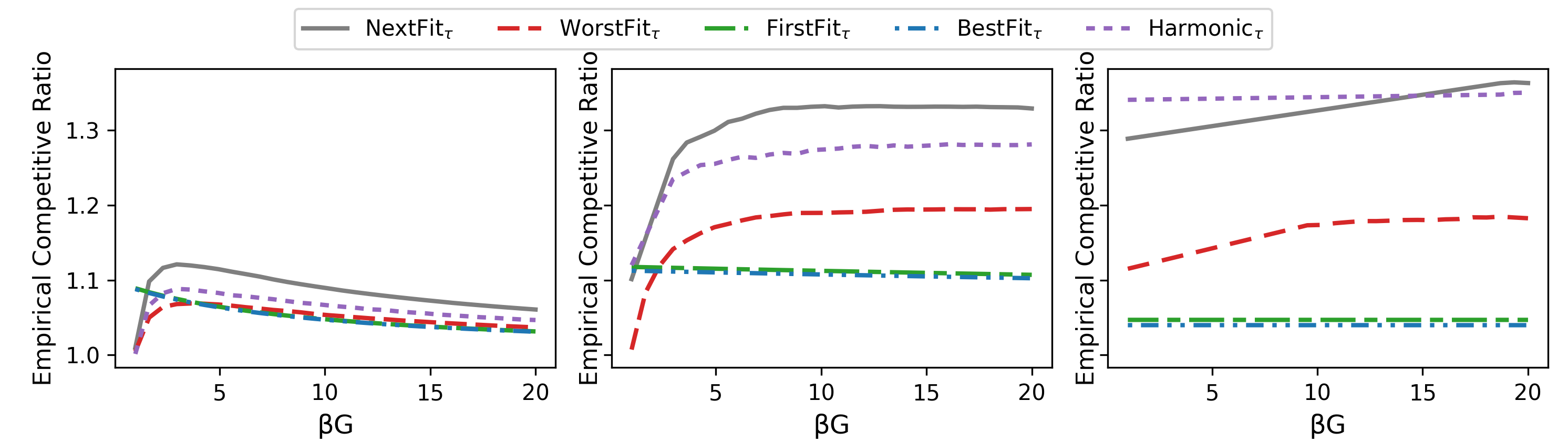}
    \makebox[0.99\textwidth][s]{ 
        \hspace{17mm}
        (a) $G = \nicefrac{1}{2}$
        \hspace{28mm}
        (b) $G= \nicefrac{3}{4}$
        \hspace{25mm}
        (c) $G= \nicefrac{19}{20}$
        \hspace{10mm}
    }
    \caption{Experiments on Weibull distribution with fixed $G$ using empirically determined $\tau$.}
    \Description{Algorithm performance is very similar for $G=1/2$. For larger $G$ the algorithms separate into three groups: \NF and \HAR, \WF, \FF and \BF.}
    \label{fig:WeibullgeOne}
\end{figure}

\section{Concluding Remarks and Future Directions}
\label{sec:conclusion}
We introduced and analyzed the \gbp problem, where filling bins up to a free-capacity threshold $G < 1$ incurs no cost, while usage beyond this point incurs a unit cost of $\beta$. Our results show that classical bin packing algorithms remain effective when $\beta G \le 1$, although their competitive analyses require new analytical techniques and reveal new insights. In contrast, when $\beta G > 1$, the problem exhibits different behavior, calling for novel algorithmic strategies. For this regime, we proposed a family of threshold-based algorithms that strategically control the utilization of costly bin capacity before opening new bins. We characterized their competitive ratios and revealed several ways in which the resulting trade-offs fundamentally differ from those in classical bin packing.

As summarized in Tables~\ref{tab:less1} and~\ref{tab:greater}, our theoretical bounds cover most algorithms and parameter ranges, yet certain gaps remain between the upper bounds on the competitive ratios achieved by our algorithms and the general lower bounds. 
Such gaps are not unique to \gbp: similar discrepancies persist even in the classical bin packing problem and have proven notoriously difficult to close. 
Tightening these gaps in the \gbp\ setting will likely require new adversarial constructions or more refined analytical techniques capable of capturing the joint dependence on~$\beta$ and~$G$. 
We regard this as a natural direction for further research.
Beyond these analytical gaps, \gbp\ opens several promising research directions.
One direction is to design adaptive or randomized thresholding strategies for the $\beta G > 1$ regime. Another is to leverage ML predictions ---for example, forecasts of future item sizes---drawing inspiration from recent advances in learning-augmented algorithms. Finally, extending the model to incorporate nonlinear (e.g., convex) cost functions would capture a broader range of practical systems where marginal energy or overcommitment costs increase superlinearly with usage.

\begin{acks}

This research is supported by National Science Foundation grants 2045641, 2325956, 2512128, and 2533814. In addition, we acknowledge the support of the Natural Sciences and Engineering Research Council of Canada (NSERC), [funding reference number RGPIN-2025-07295]. 

\end{acks}

\bibliographystyle{ACM-Reference-Format}
\bibliography{references}

\appendix

\section{Preliminaries}
\label{Appendix-Prelim}
\classicunbounded*

\begin{proof}
    Suppose an input $\sigma$ total volume $S = \sum_{a\in \sigma} a$, where each item is negligibly small. On this input, any classic algorithm $\Alg$ will fill each bin to capacity and will incur a cost $S(1+\beta(1-G))$. Now consider an alternate algorithm (\autoref{thm:OptLB}), which fills each bin to $G$, for cost $S/G$. Thus, the competitive ratio for $\Alg$ least $G(1+\beta(1-G))$, which is unbounded as $\beta \rightarrow \infty$.
\end{proof}

\section{Proofs of Offline Results}
\label{Appendix-Offline}

    \subsection{NP-Hardness and APTAS}
 
    \apxhard*

\begin{proof}
We reduce from the classic \emph{Partition} problem, which is known to be NP-hard~\cite{GareyJ79}. The input is a multiset $S$ of positive integers, and the task is to decide whether $S$ can be partitioned into two multisubsets with equal sum. We assume the total sum of elements in $S$ is $2s$ for some integer $s$; if the sum is odd, the answer is trivially no.

Let $\epsilon > 0$ be a sufficiently small constant such that $\epsilon < \min\{1/2,1/2\beta, 1 - G\}$. From the instance $S$ of Partition, we construct an instance $I(S)$ of GBP as follows. First, add two \emph{large items}, each of size $1 - \epsilon$ to $I(S)$. Second, for each item $a \in S$, create a corresponding item in $I(S)$ of size $a \cdot \frac{\epsilon}{s}$. The total size of these scaled items is $\frac{\epsilon}{s} \cdot 2s = 2\epsilon$. Thus, the total size of all items in $I(S)$ is $2(1 - \epsilon) + 2\epsilon = 2$. Note that the two large items cannot fit in the same bin, so at least two bins must be opened, each containing one large item. The cost for each such bin is $1 + \beta(1 - \epsilon - G)$, yielding a \emph{base cost} of $2\left(1 + \beta(1 - \epsilon - G)\right).$

Now, the scaled items derived from $S$ can fit into the two existing bins (without exceeding bin capacity) if and only if the original multiset $S$ can be partitioned into two subsets with equal sum $s$. In this case, the two bins will each receive a subset of scaled items summing to $\epsilon$, and the contribution from these items to the cost will be $2\beta \epsilon$. So the total cost becomes
$2(1 + \beta(1 - \epsilon - G)) + 2\beta \epsilon = 2 + 2\beta(1 - G).$

On the other hand, if $S$ cannot be partitioned into two equal-sum subsets, then at least one additional bin must be opened to pack the remaining items. In this case, the total cost is at least $2(1 + \beta(1 - \epsilon - G)) + 1 = 2(1.5 + \beta(1 - \epsilon - G))$. 

Therefore, solving the GBP instance $I(S)$ to optimality allows us to decide whether the original Partition instance has a solution. We can conclude that it is NP-hard to decide whether the optimal packing of $I(S)$ has a cost $2(1+\beta(1-G))$ or $2(1.5+\beta (1-\epsilon-G))$.
Moreover, since $\epsilon$ can be made arbitrarily small, the cost gap between the "yes" and "no" instances of GBP can be made arbitrarily close to the ratio
$\frac{2(1.5 + \beta(1 - G))}{2(1 + \beta(1 - G))} = \frac{1.5 + \beta(1 - G)}{1 + \beta(1 - G)}.$

Thus, unless $P = NP$, no polynomial-time algorithm can approximate GBP within a factor strictly smaller than this ratio.

\end{proof}

    \aptas*
    \begin{proof}
         There are three steps in our approach. The first two steps have similarities to the APTAS to classical bin packing described by Fernandez de la Vega and Lueker \cite{Vega81apt}.  
	\begin{itemize}
		\item Step I: We show that instances of the green bin packing problem with $n$ items where item sizes are all larger than $0<\delta <1$ and come from a set $S$ of finite size can be solved optimally in $O(n^{p+1})$ where $p \in O(|S|^{1/\delta})$. In particular, if both $|S|$ and $\delta$ are constants, this restricted version of the problem can be solved optimally in polynomial time.
		\item Step II: We show that there is an APTAS for the instances where all items are larger than $\delta$.
		\item Step III: Finally, we relax the last assumption and show that there is an APTAS for any instance of the green bin packing problem.
	\end{itemize}
	
	For Step I, we note that for these restricted instances, each bin contains at most $1/\delta$ items, and hence there are $p = O(|S|^{1/\delta})$ possible \emph{bin types}. Each bin type defines a set of items that can be packed inside a bin; there are at most $1/\delta$ ``spots" in such a bin type, and each can be accompanied by one of the $|S|$ possible items. 
	An instance of the problem (restricted as in Step I) with $n$ items requires at most $n$ bins, and any of its potential packings can be described with the number of bins from each bin type in the packing. Each bin type is repeated at most $n$ times in the packing, and thus there are $O(n^p)$ possible packings. If we try all of them exhaustively, we can find the cost of each in $O(n)$, and hence, an exhaustive search of all packings takes $O(n^{p+1})$. This completes Step I.

	For Step II, given an instance $I$ of the green bin packing, where all items are of a size larger than $\delta$, we do the following. 
	First, sort $I$ in a non-decreasing order of item sizes. Let $s$ be a large constant whose value will be decided later as a function of the PTAS parameter $\epsilon$. Now partition the input into $s$ partitions, each containing $\frac{n}{s}$ items (in the sorted order). Now, create an instance $I'$ of the problem where each item is scaled up to have a size equal to the largest element in its partition. For example, if $n=12$, $s=4$, and the input is \ \ $I = 0.3,0.3,0.31,0.34,0.41,0.43,0.48,0.51,0.51,0.53,.55,0.6$, then $I' = 0.31,0.31,0.31,0.43,0.43,0.43,0.51,0.51,0.51, 0.6,0.6,0.6$. Note that all items in $I'$ are larger than $\delta$, and there are $s$ item sizes in $I'$. Thus, we can find an optimal packing of $I'$ in polynomial time as in Step I. Also, note that this packing is valid for $I$. We call this packing $Alg(I)$, and we have $Alg(I) \leq Opt(I')$. To complete Step II, we show that $Opt(I') \leq (1+\epsilon) Opt(I)$. For that, form an instance $I''$ from $I$ in which every item's size (except those in the first partition, which is ignored) is decreased to the largest item in its previous partition. In the above example, $I'' = 0.31,0.31,0.31,0.43,0.43,0.43,0.51,0.51,0.51$. Note that $I''$ and $I'$ are the same sets except that $\frac{n}{s}$ largest items in $I'$ are missing in $I''$. The extra cost due to these $s$ items is at most $\frac{n}{s}(1+\beta(1-G))$. Therefore, we can write $Opt(I') \leq Opt(I'') + \frac{n}{s} (1+\beta(1-G))$. 
	Given that $I''$ is a reduced instance compared to $I$ (some items are missing and some have become smaller), we can write $Opt(I'') \leq Opt(I)$. Also, recall that $Alg(I) \leq Opt(I')$. Thus we can write $Alg(I) \leq Opt(I) +\frac{n}{s}(1+\beta(1-G))$. Since all items are larger than $\delta$, we have $Opt(I) \geq \delta n$  and thus $Alg(I) \leq Opt(I) + \frac{Opt(I)}{\delta s}(1+\beta(1-G)) = Opt(I)(1+\frac{1}{\delta s}(1+\beta(1-G)))$. For any given $\epsilon$, if we choose $s \geq \frac{1}{\delta \epsilon}(1+\beta(1-G))$, which is still constant, we can ensure $Alg(I) \leq (1+\epsilon) Opt(I)$. This completes Step II.
	
	It remains to do Step III. That is, we want to show there is an APTAS with a given parameter $\epsilon$ for any instance of the green bin packing problem. Let $\delta$ be a small value which will be decided later as a function of $\epsilon$. In our algorithm, we partition items to the \emph{small} ones that are of size at most $\delta$ and large ones that are larger than $\delta$. There are two cases depending on the value of $\beta G$ being at most or larger than 1.
	
	\textbf{Case of Small $\beta$ when $\beta \leq 1/G$:}
	In this case, we let $\delta = \min \{\epsilon/2, \epsilon (1 - G\beta)/(2\beta) \}$.
	We first pack the input formed by large items, denoted by $I_L$, using the APTAS of Step II with parameter $\epsilon' = \epsilon/3$. 
	Now, we explain the packing of small items. We process small items in the non-decreasing order of their sizes to add them to the packing of the large items given by the APTAS. To place an item $x$, we first check whether it fits in an existing bin such that the level of the bin becomes at most $G$ after placing $x$ there; if so, we place $x$ in such a bin. Consider otherwise; then, we check if $x$ fits in any existing bin $B$. If so, we place $x$ in $B$; otherwise, we open a new bin for $x$.
	
	We now analyze the algorithm. First, assume no new bin is opened for small items, and all of them have been placed in the green space of existing bins of the APTAS. In that case, the cost of the algorithm remains the same as the cost of the APTAS for $I_L$, and we can write $Alg(I) \le (1+\epsilon')Opt(I_L) \leq (1+\epsilon')Opt(I) < (1+\epsilon)Opt(I)$, which completes the proof; note that $Opt(I_L) \leq Opt(I)$ holds trivially as $I_L$ is a subset of $I$. 
	
	Next, consider a case where a new bin is opened for small items. In this case, the level of all bins in the final packing of $Alg(I)$ is at least $1-\delta$ (except for the very last bin). If $S_I$ denotes the total size of items in $I$, then we have $Opt(I) \geq S_I(1+\beta(1-G))$.
    For the cost of the algorithm, we can write $Alg(I) \leq \frac{S_I}{1-\delta}(1+\beta(1-G))$; this is because the algorithm has at most $\frac{S_I}{1-\delta}$ bins (each bin except the last has a level of at least $1 - \delta$; otherwise a new bin was not opened). Moreover, each bin has a cost of at most $1+\beta(1-G)$. We conclude that the competitive ratio is at most $1/(1-\delta) = 1+\delta/(1-\delta) \leq 1+\epsilon$. 
	
	Finally, suppose no new bin is opened for small items, but some are partly or fully placed in \black space. The packing of the algorithm has a property that all bins are filled up to a level of at least $G-\delta$ (otherwise, a small item is partly placed into the \black space of some bin while it could be placed fully in a green space of some other bin). Let $m$ denote the number of bins in $Opt(I_L)$ and $S_I$ is the total size of items. Given that $m$ bins are necessary for packing $I$, we conclude that $Opt(I) \geq m + \beta (S_I-mG)$; this is because in an ideal scenario that $Opt$ uses $m$ bins, at most a total size of $mG$ is placed in the green spaces; the remaining bulk of size $S_I-mG$ is in the \black space. 
    
      On the other hand, let $m'$ denote the number of bins used by our algorithm, which is the same as $Opt(I')$, where $I'$ is the instance with only large items and items being scaled up in their partition.

     Then all bins are filled up to a level of at least $G-\delta$; thus, a total size of at least $m' (G-\delta)$ is placed in green space; the cost of the algorithm is thus at most $m' + \beta (S_I-m'(G-\delta))$. The ratio between the two costs is at most 
     \begin{align*}
      \frac{m'+\beta (S_I- m'(G-\delta)))}{m +\beta (S_I - mG)} = \frac{m' (1-G \beta+\delta \beta) + \beta S_I}{m (1- G\beta) + \beta S_I} &\le \frac{m'}{m}\cdot (1 + \frac{\delta \beta}{1 - G \beta})\\
      &\le (1 + \epsilon') \cdot (1 + \frac{\delta \beta}{1 - G \beta}) \le 1 + \epsilon.
     \end{align*}
     The second inequality holds since $m' \le (1 + \epsilon') m$.
     To see this, recall that $m'$ is the number of bins for obtaining $Opt(I')$ and Let $m''$ denote the number of bins for obtaining $Opt(I'')$, where $I''$ is the same as $I'$ except that it excludes the largest $n/s$ items. Then we have $m' \le m'' + \frac{n}{s} \le m + \frac{n}{s} \le m + \frac{m}{s \delta} \le (1 + \epsilon') m$, where we use the fact that $m'' \le m$, $m \ge \delta n$, and $s \ge \frac{1}{\delta \epsilon'}(1+\beta(1-G)) \ge \frac{1}{\delta \epsilon'}$.   
	
	\textbf{Case of large $\beta$ when $\beta > 1/G$:} In this case, we let $\delta = \epsilon/(4\beta^2)$. We start with \emph{bundling} small items as follows. Process small items in an arbitrary order and group them so that each group has a total size in $[\delta,2\delta)$. We call each group a bundle (ignore the last group in the asymptotic analysis). Let $I'$ be an instance of the green bin packing in which small items are replaced with their bundles. We apply the APTAS of step II to pack $I'$ with parameter $\epsilon' = \epsilon/3$; the result will be a valid packing of $I$ (in which small items forming a bundle are placed together), which we call Alg(I). Note that $Alg(I) \leq (1+\epsilon') Opt(I')$.
	
	For the analysis, we claim that $Opt(I') \leq (1+2\delta \beta^2) Opt(I)$. Recall that the APTAS tells us that $Alg(I) \leq (1+\epsilon') Opt(I')$. We can conclude $Alg(I) \leq (1+\epsilon')(1+2\delta \beta^2) Opt(I) = 
	1+(2\delta \beta^2+ \epsilon'(1+2\delta \beta^2))Opt(I) \leq (1+\epsilon) Opt(I)$, which completes the proof.

	It remains to prove the claim. We process each bin in the optimal packing of $I$ as follows. First, remove all small items from the packing. Now, process bins in the optimal packing in the non-decreasing order of their original level. For each bin $B$, let the total size of (removed) small items be $\ell_B$ and if $\ell_B > \epsilon$, keep adding bundles to $B$ until the total size of added bundles, denoted by $\ell'_B$, becomes more than $\ell_B$; we will have $\ell_B < \ell'_B \leq \ell_B + 2\delta$. Note that the cost of $B$ is increased by at most by $2\delta \beta$ in the packing, and the cost of $B$ in the resulting packing is increased by a factor of at most $(1+2\delta \beta^2)$. Given that each bin receives more small items than it originally had, we run out of bundles before processing all bins. The result is a valid packing of $I'$, which uses the same number of bins as the optimal packing of $I$, and where the cost of each bin is at most $(1+2\delta \beta^2)$ times the cost of the corresponding bin in the optimal packing. We conclude that $Opt(I') \leq (1+2\delta \beta^2) Opt(I)$.
    \end{proof}

\section{\AAF and \HAR Upper Bounds}
\label{Appendix-AAFUB}
In this section, we give additional proofs for the case that $\beta G \le 1$. We begin by proving that any online algorithm should attempt to fill to capacity 1 in this setting, and then give full proofs for   \autoref{thm:anyfitUB} and \autoref{thm:harUB} which establish upper bounds for \AAF and \HAR algorithms when $\beta G \le 1$. 

\subsection{Proof of \autoref{thm:fillToCapacity}}

\fillToCapacity*

\begin{proof}
Let $A$ and $A^+$ denote the same algorithm with thresholds $\tau$ and $\tau + \epsilon$, respectively.  
Define $\gamma = \frac{\tau}{\tau + \epsilon}$, noting that $\gamma \in [0, 1]$.

Consider the worst-case input $\sigma$ for $A^+$.  
Let $s(\sigma)$ be the total size of items in $\sigma$.  
Construct a new input sequence $\sigma'$ that begins with many ``sand'' items (arbitrarily small items) of total size $X = (1-\gamma)s(\sigma)$ %
followed by a copy of $\sigma$ in which all item sizes are scaled down by the factor~$\gamma$.  
Note that $X$ is asymptotically large and we can assume it is an integer multiple of $G + \tau$.  
Since the total size of $\sigma'$ equals that of $\sigma$ (i.e., $s(\sigma') = s(\sigma)$), 
any offline packing for $\sigma$ can be converted into a valid packing for $\sigma'$ 
(by replacing each item with its scaled-down copy plus the corresponding sand).  
Thus, $\OPT(\sigma') \le \OPT(\sigma)$.

When $\sigma'$ is given to algorithm~$A$, it first places the sand items into $\frac{X}{G+\tau}$ bins and then proceeds to pack the remaining items exactly as $A^+$ packs $\sigma$—the scaling ensures that bin utilization patterns are identical.  
Hence, both algorithms use the same number~$m$ of bins for the main part of the sequence.

The cost of $A$ for the sand items is $\tfrac{(1 + \beta \tau) X}{G + \tau}$.  
For the rest of the sequence, the packing of~$A$ corresponds to that of~$A^+(\sigma)$, except that the scaled-down items may save at most $\beta X$ in black-space cost.  
Therefore,
\[
A(\sigma') 
  \ge \frac{(1 + \beta \tau) X}{G + \tau} + A^+(\sigma) - \beta X
  = A^+(\sigma) + \frac{(1 - \beta G) X}{G + \tau}
  \ge A^+(\sigma),
\]
where the last inequality follows from $\beta G \le 1$.

Finally, since $\OPT(\sigma') \le \OPT(\sigma)$, we have
\[
\CR(A) 
   \geq \frac{A(\sigma')}{\OPT(\sigma')}
   \ge \frac{A(\sigma')}{\OPT(\sigma)}
   \ge \frac{A^+(\sigma)}{\OPT(\sigma)}
   = \CR(A^+).
\]
This proves that decreasing
the threshold cannot improve the competitive ratio when $\beta G \le 1$, and thus the optimal choice is $\tau = 1 - G$.
\end{proof}

\subsection{\AAF and \HAR Upper Bounds}

To establish these upper bounds, we will use a weighting argument.

 The idea is to define a weight $w(x)$ for any item of size $x$. We let $W$ show the total weight of all items. Then we show I) $Alg \leq W$ and II) $Opt \geq W/c$. This ensures a competitive ratio $c$ for Alg. To establish (I), we show the cost of Alg for any bin is at most equal to the total weight of items in the bin. To establish (II), we show that the minimum possible ratio between the cost and the total weight of items in any bin is at least $1/c$. This means that the ratio between the total cost of \Opt and the total weight of items in \Opt's packing is at least $1/c$. That is, $W/Opt \leq c$. To find this $c$, we rely on the following proposition.

\propWPCboundI*
\begin{proof}
     Consider a bin $B$ of \Opt. We aim to identify a packing of $B$ that maximizes the ratio between the weight of the bin and its cost. %
First, suppose that items in $B$ form a set $X$ and the total size of items in $X$ is $S(X) < G$; then $cost(B) = 1$, and the ratio between weight and cost is $W(X)/cost(X) = W(X)$. If we add a collection of items of size $\epsilon$ to $X$ to get a set $X'$, we will have $cost(X') = cost(X) = 1$, and $W(X') > W(X)$ (the weights are all positive). Thus $\frac{w(X')}{cost(X')} > \frac{W(X)}{cost(X)}$. In other words, to maximize the weight/cost ratio, it is best to fill to at least level $G$. 

Next, we show that when $\beta G \leq 1$, adding more items of size $\epsilon$ to $X'$ can increase the ratio even further.
This requires assuming the density of all items, defined as their weight/size ratio, is at least $\beta$ (which we need to pay attention to when defining weights).
Let $X''$ be a set formed by $X'$ via adding one item of size $\epsilon$. We will have $cost(X'') = cost(X') + \beta \epsilon$ and $W(X'') \geq W(X')  +\beta \epsilon$ due to the assumption we made about the density. %
So, $\frac{W(X'')}{cost(X'')} \geq  \frac{W(X') + \beta \epsilon}{cost(X') + \beta \epsilon}  \geq \frac{W(X')}{cost(X')} $. Repeating the above argument, we may conclude that to minimize the cost/weight ratio, it is best to fill up $B$ completely to a level of 1.
Thus, the cost of \Opt for $B$ will be exactly $1+\beta(1-G)$. To maximize the weight/cost ratio then, we need to maximize the weight of items inside a bin. Now, we are in a setting similar to classic bin packing.

If the maximum weight of a bin is $w^*$, then the maximum weight/cost ratio (and upper bound for the \CR of \Alg) would be $w^*/(1+\beta(1-G))/w^*$.
\end{proof}

 With this background, we now approach the theorems.

\subsubsection{\AAF Upper Bound}

\anyfitUB*

Let an item be `large' if it is larger than 1/2 and 'small' otherwise. In this  proof, we rely on the property of \AAF algorithms that each of its bins, except possibly one, either includes a large item or has a level of at least 2/3. This is because if a bin is formed only by small items and level less than 2/3, any subsequent bin formed by small items will include exactly two small items larger than 1/3.

Based on this property, we split the cases for $G\leq 2/3$ and $G>2/3$ into the following two propositions.

\propWPCboundII*
\begin{proof}
 Note that by our chosen weights and $\beta G \leq 1$, we have $w(x)\geq \beta x$. Additionally, since $G\le 2/3$, any \AAF bins containing only small items will fully use their green space.

Next, we show that the total weight of items in a bin of \AAF is at least equal to that bin's cost. Consider a bin $B$ of \AAF. We show that the total weight of items in $B$ is at least the cost \AAF pays for $B$. Suppose the total size of small and large items in $B$ is $L_s$ and $L_\ell$, respectively. Now, if $B$ does not contain a large item, we will have $L_s \geq 2/3$ and $L_\ell = 0$; in this case, the total weight of items in $B$ would be 
$(1.5(1-\beta G) + \beta) L_s =  
1.5(1-\beta G)L_s + \beta L_s  
\geq (1-\beta G) + \beta L_s = 1 + \beta (L_s-G) = cost(B)
$. If $B$ does contain an item $x>1/2$, the weight of $B$ would be $(1+\beta (x - G)^+) + (1.5(1-\beta G)+\beta)L_s \geq (1+\beta (x - G)^+) + \beta L_s \geq 1 + \beta (x + L_s - G) = cost(B) $.

To complete the upper bounds analysis, we establish an upper bound $w^*$ for the total weight of items in any bin of \Opt. As mentioned earlier, that would establish an upper bound of $w^*/(1+\beta (1-G))$ for \AAF.

Consider a bin $B^*$ of \Opt. If $B^*$ does not contain any large item, its total weight is at most $w_1 = 1.5(1-\beta G)+\beta$. We will refer back to $w_1$ later.

If $B^*$ contains a large item of size $x$ (we have $x>1/2)$, its weight would be $1+\beta(x - G)^+$ for $x$ plus at most $(1.5(1-\beta G)+\beta)(1-x)$ for small items. Let $y = 1-x$, and we have $w(B^*) \leq 1+\beta((1-y)- G)^+ + (1.5(1-\beta G)+\beta)y$, which is an increasing function of $y$; this is because the coefficient of $y$ is at least $1.5(1-\beta G)$ %
which is positive. So, the weight of $B^*$ is maximized when $y$ takes its maximum value of $0.5-\epsilon$ (when $x$ takes its minimum value of $1/2+\epsilon$). 
In this case, the weight would be $1+\beta (0.5 - G)^+ + 0.75(1-\beta G) + 0.5\beta$. If $G\leq 0.5$, this would be $1.75 + \beta - 1.75 \beta G = 1.75 (1-\beta G)+\beta $; note that this weight is always larger than $w_1$. We conclude the competitive ratio would be at most $\frac{1.75(1-\beta G)+\beta}{\beta + (1-\beta G)} = \frac{1.75 +\beta (1-1.75 G)}{\beta + (1-\beta G)}$. 
If $G>0.5$, the weight of $B^*$ would be at most $1 + 0.75(1-\beta G) + 0.5\beta = 1.75 +\beta(0.5-0.75G)$, and the competitive ratio of \AAF will be bounded by $\frac{\max\{1.75 +\beta(0.5-0.75G),w_1\}}{\beta + (1-\beta G)} $. For all $\beta G \le 1$ and $0.5 < G \le 2/3$, the bound will become $\frac{1.75 +\beta(0.5-0.75G)}{\beta + (1-\beta G)}$ 
\end{proof}

\propWPCboundIII*
\begin{proof}
Again note that weights are chosen so that $w(x) \ge \beta x$.

We first show the weight of any bin $B$ of \AAF is at least equal to its cost. Suppose $B$ contains no large item, and its level is $L$. If $L\leq G$, the cost of \AAF for $B$ is 1, while the total weight of items is $1.5 L \geq 1.5 (2/3) = 1$. Next, suppose $L = G + q$ for some positive $q$; then its total weight is $1.5 (G+q) = 1.5G + 1.5 q \geq 1 + 1.5q \geq 1+\beta q = cost(\AAF)$; the first inequality holds because $g\geq 2/3$ and the second inequality is due to $\beta \leq 3/2$ (again because $G\geq 2/3$ and $\beta G\leq 1$). If $B$ does contain an item $x>1/2$ and $L_s$ denote the total size of its small items, the weight of $B$ would be $(1+\beta (x - G)^+) + 1.5L_s \geq (1+\beta (x - G)^+) + \beta L_s \geq 1 + \beta ((x + L_s) -G) = cost(B) $.

Consider a bin $B^*$ of \Opt. We show the weight of items in $B^*$ cannot exceed 1.75. If $B^*$ contains no large item, its total weight is at most $1.5$. If $B^*$ contains a large item of size $x$ (we have $x>1/2)$, its weight would be $1+\beta(x - G)^+$ for $x$ plus at most $1.5(1-x)$ for small items. 
This would give $w(B^*) \leq 1+\beta(x - G)^+ + 1.5(1-x)$; this is a decreasing function of $x$ because the coefficient of $x$ is $-(1.5 - \beta (x - G)^+)$, which is non-positive. So, the weight of $B^*$ is maximized when $x$ takes its minimum value of $0.5+\epsilon$. 
In this case, the weight would be
$1+\beta(0.5 - G)^+ + 0.75 = 1.75$, this gives an upper bound of $\frac{7/4}{1+\beta(1-G)}$ on the competitive ratio of \AAF for large $G$.   

\end{proof}

\subsubsection{\HAR Upper Bound}

We begin by establishing the following lemma to support the proof.

\begin{lemma}\label{lemma:HAUpper}
    Let $S \subseteq (0,1]$ be a set of item sizes, and suppose we want to select an arbitrary subset of items from $S$ whose total size is at most $c \in (0,1]$. Define the \emph{fixed weight} of an item of size $x \in (1/(i+1), 1/i]$ as follows ($K$ is a sufficiently large value, e.g., $K\geq 50$):
    \[
        \text{fixed weight}(x) = 
        \begin{cases}
            1/i & \text{if } i \leq K-1, \\
            \frac{K}{K-1} \cdot x & \text{otherwise.}
        \end{cases}
    \]
    Then, the maximum total fixed weight achievable by selecting items with total size at most $c$ is bounded as follows:
    \begin{itemize}
        \item If $c = 1$ and $S = (0,1]$, the total fixed weight is at most $1.691$.
        \item If $c = 1/2$ and $S = (0,1)$, the total fixed weight is at most $0.691$
        \item If $c = 1$ and $S = (0,1/2] \cup [2/3,1]$, the total fixed weight is at most $1.42$
        \item If $c = 1/3$ and $S = (0,1)$, the total fixed weight is at most $0.4231$
        \item If $c = 1/6$ and $S = (0,1)$, the total fixed weight is at most $0.191$
        \item if $c=1/2$ and $S=(0,1/3]$, the total fixed weight is at most 0.636.
    \end{itemize}
\end{lemma}

\begin{proof}
To maximize the weight, if we choose an item of class $i \leq K-1$, then it is best to select the smallest item of that class (making items smaller does not decrease the weight but makes room for other items).
\begin{itemize}
    \item     If $c=1$ and $S=(0,1]$, the analysis is similar to that of classic \HAR. It is best to select items of sizes $1/2+\epsilon, 1/3+\epsilon, 1/7+\epsilon, \ldots$ which have a total weight of $1+1/2+1/6+1/42 + \ldots \approx 1.691$ \cite{Lee1985Har}.
    
    \item Suppose $c = 1/2$ and $S = (0,1)$. (I) Assume $B^*$ includes an item of size $1/3+\epsilon$ and an item of size $1/7+\epsilon$. The remaining space would be $1/42 - 3\epsilon$ and the density of items in this space is at most $\frac{1/42}{1/43+\epsilon} = 43/42$. The total weight would be at most $1/2+1/6+(1/42)(43/42) \approx 0.691$. (II) If there is no item of size larger than 1/3 in $B^*$, the density of remaining items would be at most $\frac{1/3}{1/4+\epsilon} < 4/3$, and the total weight would be at most $(4/3)\cdot (1/2) = 2/3 \approx 0.66$. (III) If there is an item of size $1/3+\epsilon$ but no item of size $1/7+\epsilon$, the density of other items would be at most $\frac{1/7}{1/8+\epsilon} < 8/7$, and the total weight would be at most $1/2+ (1/6)\cdot(8/7) = 0.6904$.

    \item Suppose $c=1$ and $S = (0,1/2] \cup [2/3,1]$. We use a case analysis. (I) Suppose $B^*$ includes an item of size $2/3+\epsilon$ or two items of $1/3+\epsilon$; regardless, the total remaining space is $1/3-\epsilon$ and the weight of selected items (so far) would be 1. Now, if the remaining space contains an item of size $1/4+\epsilon$ (case (Ia)), the remaining items would have sizes less than 1/12 and density (weight/size ratio) of at most $\frac{1/12}{1/13 + \epsilon} < 13/12$; the total weight would be $1+1/3+ (1/12)\cdot(13/12) \approx 1.4231$. 
    Otherwise (case (Ib)), if there is no item of size $1/4+\epsilon$, the density (weight/size ratio) of other items would be $\frac{1/4}{1/5+\epsilon} < 5/4$ and thus the total weight would be $1+(5/4)\times (1/3) = 1.416$. (II) If $B^*$ includes one item of size $1/3+\epsilon$, the density of remaining items would be at most $\frac{1/3}{1/4+\epsilon} < 4/3$ and the total weight becomes $1/2+(2/3)\times (4/3) \approx 1.388$. (III) If $B^*$ does not include an item of size larger than 1/3, the total weight becomes at most $1\times 4/3 \approx 1.333$.
 
    \item Suppose $c = 1/3$ and $S = (0,1)$. (I) If $B^*$ includes an item of size $1/4+\epsilon$ and an item of size $1/13+\epsilon$, the remaining space would have size at most $1/156$ and, since the density of items in this space is at most $\frac{1/156}{1/157+\epsilon} = 157/156$. Therefore, the total weight of items in $B^*$ would be $1/3+1/12+(1/156)\cdot(157/156) = 0.423$. (II) If there is no item of size larger than 1/4 in $B^*$, the density of remaining items would be at most $\frac{1/4}{1/5+\epsilon} < 5/4$, and the total weight would be at most $(5/4)\cdot (1/3) \approx 0.416$. (III) If there is an item of size $1/4+\epsilon$ but no item of size $1/13+\epsilon$, the density of other items would be at most $\frac{1/13}{1/14+\epsilon} < 14/13$, and the total weight would be at most $1/3+ (1/12)\cdot(14/13) = 0.423$.

    \item Suppose $c = 1/6$ and $S = (0,1)$. (I) If $B^*$ includes an item of size $1/7+\epsilon$, the remaining space is $1/42-\epsilon$ and all items in this space have density at most $\frac{1/42}{1/43+\epsilon} < 43/42$ and total weight would be at most $1/6+(1/42)\cdot (43/42) \approx 0.191$. (II) If there is no item of size larger than 1/7 in $B^*$, the density of remaining items would be at most $\frac{1/7}{1/8+\epsilon} < 8/7$, and the total weight would be at most $(8/7)\cdot (1/6) = 0.1904$.

\item If $c=1/2$ and $S=(0,1/3]$.
    (I)    If there is no item of size larger than 1/4, the density of remaining items would be at most $\frac{1/5}{1/4+\epsilon} < 5/4$, and the total weight would be at most $5/8$. (II) If there is an item of size $1/4+\epsilon$ but no other item of size larger than $1/5$, the density of other items would be at most $\frac{1/5}{1/6+\epsilon} < 6/5$, and the total weight would be at most $1/3+ (1/4)\cdot(6/5) \approx 0.6333$.
(III) If there is an item of size $1/4+\epsilon$, an item of size $1/5+\epsilon$, the density of the other item is at most $\frac{1/20}{1/21+\epsilon} = 21/20$ and the total weight would be at most $1/3+1/4+(1/20)\cdot(21/20) < 0.636$.

\end{itemize}

\end{proof}

We are now ready to prove the upper bound for the competitive ratio of Harmonic.

\harUB*

We use a weighting argument, as described for \autoref{thm:anyfitUB}. Let the weight of an item of size $x$ from class $i$ (of size in $\in [1/(i+1),1/i]$ to be 

\begin{align*}w(x) = \frac{1}{i} (1+\beta (ix - G)^+) \tag{\ref{eq:har_weight}}\end{align*}

We first show that the cost of \HAR for each bin is at most equal to the weight of that bin (except for possibly a constant number of bins). 

\harUBpropI* 
\begin{proof}
Let $B$ be a bin of \HAR of type $i$. There are $i$ items of class $i$ in $B$ with sizes $x_1,x_2,\ldots, x_i$ totalling a weight of $$1 + \frac{\beta}{i} \sum_{j\in[1,i]} (i\cdot x_j-G)^+ \geq 1 + \frac{\beta}{i} \sum_{j\in[1,i]} (i\cdot x_j-G) = 1+\beta ( (\sum_{j\in [1,i]}x_j) - G).$$ This latter term is the cost of \HAR for the bin (by definition). Thus, for each bin B, the cost of \HAR is at most equal to the weight of all items in $B$.
\end{proof}

Next, we need to find an upper bound for the weight of any bin $B^*$ of \Opt. 
\harUBpropII* 
\begin{proof}
We will use a case analysis based on the value of $G$.

\begin{itemize}
    \item Suppose $G \leq 1/2$.  
Suppose $B^*$ contains items $y_1,y_2,\ldots, y_m$ for some integer $m$, such that $y_j$ has class $c_j$. Given that $G\leq 1/2$, for any item of size $y$ of class $c$, we can write $cy \geq G$. Thus, for the total weight of items in $B^*$, we can write:

\begin{align*}
\sum_{j\in[1,m]}(\frac{1}{c_j}(1+\beta(c_j y_j - G)^+)   =  
(1-\beta G)\sum_{j\in[1,m]} \frac{1}{c_j} + \beta \sum_{j\in[1,m]} y_j  
\end{align*}
Recall that \Opt should fill out the bins fully to maximize its weight/cost ratio. Therefore, we have $\sum_{j\in[1,m]} y_j =1$, and the total weight is $$ (1-\beta G)\sum_{j\in[1,m]} \frac{1}{c_j} + \beta $$
To maximize the weight, it is necessary to select items in a way to maximize $\sum_{j\in[1,m]} \frac{1}{c_j}$. By \autoref{lemma:HAUpper}, the maximum value of $\sum_{j\in[1,m]} \frac{1}{c_j}$ is 1.691. Therefore, we can conclude that the maximum weight of items in $B^*$ is $\beta + 1.691(1-\beta G)$ and the competitive ratio would be $\frac{\beta + 1.691(1-\beta G)}{1+\beta(1- G)}$.

\item Suppose $G\in (1/2,2/3]$. We consider two cases. \\
First, suppose $B^*$ contains an item $y_0 \in (1/2,2/3]$. Note that the weight of $y_0$ would be 1, and to maximize the total weight, it is best for $y_0$ to have size $1/2+\epsilon$. Suppose other items are $y_1,y_2,\ldots, y_m$ for some integer $m$, such that $y_j$ has class $c_j$. These items are all smaller than 1/2, so we have $cy > G$. The total weight of items in $B^*$ is then

\begin{align*}
1 + \sum_{j\in[1,m]}(\frac{1}{c_j}(1+\beta(c_j y_j - G)^+)   &= 1 +  (1-\beta G) 
\sum_{j\in[1,m]} \frac{1}{c_j} + \beta \sum_{j\in[1,m]} y_j \\  
& = 1 +  (1-\beta G) 
\sum_{j\in[1,m]} \frac{1}{c_j} + \beta(1/2-\epsilon).
\end{align*}
The last equality holds because $B^*$ is fully filled and $y_0$ has size $1/2+\epsilon$.

So, to maximize the total weight, it is necessary to maximize $\sum_{j\in[1,m]} \frac{1}{c_j}$, when the total available capacity is 1/2. By \autoref{lemma:HAUpper}, this value is at most 0.691. In this case, we get a total weight of $ 1+ 0.691 (1-\beta G) + \beta/2 = 1.691 - 0.691\beta G + \beta /2$.

Next, assume there is no item of size in $(1/2,2/3]$ in $B^*$. Suppose $B^*$ contains items $y_1,y_2,\ldots, y_m$ for some integer $m$, such that $y_j$ has class $c_j$. Since each $y_j \le 1/2$ we have $cy > G$.

\begin{align*}
\sum_{j\in[1,m]}(\frac{1}{c_j}(1+\beta(c_j y_j - G)^+)   =  
\sum_{j\in[1,m]} \frac{(1-\beta G)}{c_j} + \beta \sum_{j\in[1,m]} y_j  
\end{align*}
Again, since \Opt should fill out the bins fully to maximize its weight/cost ratio, we have $\sum_{j\in[1,m]} y_j =1$, and the total weight is $$\beta + (1-\beta G)\sum_{j\in[1,m]} \frac{1}{c_j} $$
As before, to maximize the weight, it is need to select 
items in a way to maximize $\sum_{j\in[1,m]} \frac{1}{c_j}$ subject to all item sizes being in $(0,1/2] \cup (2/3,1]$. By \autoref{lemma:HAUpper}, the maximum value of $\sum_{j\in[1,m]} \frac{1}{c_j}$ is 1.42. Therefore, we can conclude that the maximum weight of items in $B^*$ is $\beta + 1.42(1-\beta G) = 1.42 - 1.42\beta G + \beta$.
To conclude, the maximum weight in this case would be at most 
\begin{align*}
\max\{1.691 - 0.691\beta G + \beta /2, 1.42 - 1.42 \beta G + \beta \} = 1.691 - 0.691\beta G + \beta /2.    
\end{align*}
The maximum is realized by the first term in the specific range for $G$.
Therefore, the competitive ratio of \HAR in this case would be at most $\frac{1.691 - 0.691\beta G + \beta /2}{1+\beta(1-G)}$.

\item Suppose $G>2/3$. To analyze this case, we look at the density of items of class $i$. An item of size $x$ and class $i$ has the following density (note $x\in (1/(i+1),1/i]$):
\begin{align*}
den(x)  = \frac{w(x)}{x}  &= \frac{1}{x i} (1+\beta (ix - G)^+)    \\ & = \frac{1}{x i} + \beta (1 - \frac{G}{x i})^+ \\ & \leq \frac{1}{xi} + \beta(1-G) \\ & \leq \frac{i+1}{i} + \beta(1-G).
\end{align*}

We consider the following cases.\\
Case I: First, suppose $B^*$ contains an item $y_0 \in (1/2,2/3]$. The weight of $y_0$ would be 1, and it is best to have size $1/2+\epsilon$.
\begin{itemize}
    \item Case Ia: Suppose $B^*$ contains also an item $y_1 \in (1/3,2/3]$. The weight of $y_1$ would be 1/2, and it is best to have size $1/3+\epsilon$. 
    Let other items be $y_2,\ldots, y_m$ for some integer $m$ such that $y_j$ has class $c_j$. %
    Given that the remaining space in $B$ is $1/6-2\epsilon$, %
    the total weight would be %

    \begin{align*}
   1+ 1/2 + \sum_{j\in[2,m]} y_j \cdot den(j) & = 1.5+ \sum_{j\in[2,m]}y_j (\cdot\frac{1}{y_j c_j} + \beta(1-G)) \\ & = 1.5 + \sum_{j\in[2,m]} \frac{1}{c_j} + \beta(1-G)  \sum_{j\in[1,m]} y_j \\ & = 1.5 + \sum_{j\in[2,m]} \frac{1}{c_j} + \beta(1-G)/6   
\end{align*}

 So, to maximize the total weight, it is necessary to maximize $\sum_{j\in[2,m]} \frac{1}{c_j}$, when the total available capacity is 1/6. 
By \autoref{lemma:HAUpper}, this value is at most 0.191. In this case, we get a total weight of at most $ 1.5 + 0.191 + \beta (1-G)/6 = 1.691 + \beta (1-G)/6$.

\item Case Ib: Suppose $B^*$ does not contain an item of size in $(1/3,2/3]$. 
Other items belong to classes $3$ or larger.
Let other items be $y_1,\ldots, y_m$ for some integer $m$ such that $y_j$ has class $c_j$. Given that the remaining space in $B^*$ is $1/2-\epsilon$, the total weight of items in $B^*$ would be at most:
\begin{align*}
   1+ \sum_{j\in[1,m]} y_j \cdot den(j) & = 1+ \sum_{j\in[2,m]}y_j (\cdot\frac{1}{y_j c_j} + \beta(1-G)) \\ & = 1 + \sum_{j\in[1,m]} \frac{1}{c_j} + \beta(1-G)  \sum_{j\in[1,m]} y_j \\ & = 1 + \sum_{j\in[1,m]} \frac{1}{c_j} + \beta(1-G)/2   
\end{align*}

So, to maximize the total weight, it is necessary to maximize $\sum_{j\in[1,m]} \frac{1}{c_j}$, when the total available capacity is 1/2 and items are of size at most 1/3. By \autoref{lemma:HAUpper}, this value is at most 0.636. In this case, we get a total weight of 
$1+0.636 + \beta(1-G)/2 = 1.636 + \beta (1-G)/2$. %
Depending on the specific $\beta, G$ either $1.691 + \beta (1-G)/6$ or $1.636 + \beta (1-G)/2$ may be larger.

\end{itemize}

Case II: Next, suppose $B^*$ does not contain an item in $(1/2,2/3]$ but contains one or two items in $(1/3,1/2]$.
\begin{itemize}
    \item Case IIa: Suppose $B^*$ contains two items $y_0, y_1$ of size in $(1/3,1/2]$. They both have weight 1/2, and to maximize weight of $B^*$, they should be of size $1/3+\epsilon$. 

Other items belong to class 3 or more and have density $\frac{4}{3} + \beta (1-G)$. Given that the remaining size $1/3-2\epsilon$, the total weight is thus at most $1/2+1/2+ \frac{1}{3} (\frac{4}{3} + \beta (1-G))  = 1.444 + \beta (1-G)/3$, which is less than $1.691 + \beta(1-G)/6$. %

    \item Case IIb: Suppose $B^*$ contains only one item $y_0$ of size in $(1/3,1/2]$. This item has weight 1/2, and to maximize the weight of $B^*$, it should be of size $1/3+\epsilon$. 
    As in the previous case, other items belong to class 3 or more and have density $\frac{4}{3} + \beta (1-G)$. Given that the remaining size $2/3-\epsilon$, the total weight is thus at most $1/2+ \frac{2}{3} (\frac{4}{3} + \beta (1-G))  = 1.388 + 2\beta (1-G)/3$, which is less than $1.691 + \beta(1-G)/6$.

\end{itemize}

Case III: Next, suppose $B^*$ does not contain any item with sizes in $(1/3,2/3]$.  
    So, all items belong to class 3 or more and have density $\frac{4}{3} + \beta (1-G)$. Given that the available size is $1$, the total weight is at most $(\frac{4}{3} + \beta (1-G))  = 1.333 + \beta (1-G)$, which is less than $1.691 + \beta(1-G)/2$.

To conclude, when $G>2/3$, the total weight of $B^*$ is no more than $1.691+ \beta(1-G)/6$, and thus the upper bound for the competitive ratio of HA is at most  $\max\{\frac{1.691 + \beta(1-G)/6}{1+\beta (1-G)}, \frac{1.636+\beta(1-G)/2}{1+\beta(1-G)}\}$.

\end{itemize}
\end{proof}

Combining \autoref{prop:wpc bound}, \autoref{prop:har_ub_0}, and \autoref{prop:har_ub_1} gives \autoref{thm:harUB}.

\section{Proofs of Upper Bounds for $\beta G>1$}
\label{Appendix-TAAFUB}

In this section, we provide full proofs for the upper bounds of \TNF, \TWF, and \TAAF \& \GHAR, as introduced in \autoref{sec:GBgtOne}.

\subsection{Upper Bound of \TNF}

\tnf*

Here we show the upper bound portion of this proof. We will begin by establishing some useful notation.

We refer to a bin of \TNF as green (respectively \black) if the total size of items in the bin is at most $G$ (respectively more than $G$). 
    We partition the bins in the packing of \TNF into GG bins, formed by a pair of consecutive green bins; GB bins, formed by a green bin followed by a \black bin; and B bins, which is a single \black bin. Note that a simple linear scan can help us partition the bins as such. A pair of GB bins is said to \emph{special} if the \black bin is opened by an item of size larger than $G$, such item is called an \emph{special} item. Non-special GB bins are called \emph{regular} GB bins. We use GBR and GBS to refer to regular and special GB bins, respectively.

    We use $S_{GG}$, $S_{GBR}$, and $S_B$ to respectively denote the total size of items in GG, GBR, and B-bins. We use $m$ to denote the number of pairs of $GBS$ pairs (there are $2m$ such bins). As before, we let $S_{GBS}$ be the total size of non-special items in the GBS bins.
    Assume the total size of special items in all GBS bins is $mG + \Gamma$; that is, if we place every special item in a separate bin, the total \black space used will be $\Gamma$.

    Given the final packing of \TNF, for any $X\in \{GG, GBR, GBS\}$, we let $A_X$ denote the space wasted in all green bins of type $X$. That is, if a green bin of type $X$ has a level $G-a$, the contribution of that bin to $A_X$ will be $a$. Similarly, for any $X \in \{GBR, GBS, B\}$, let the $B_X$ denote the total \black space used by bins of type $X$. 

Now, we are ready to investigate the cost of \TNF as follows:

\begin{itemize}
    \item Consider the GG bins. The cost of the algorithm for these bins is $\frac{S_{GG}+A_{GG}}{G}$. Note that $S_{GG}+A_{GG}$ is simply the total green capacity of all GG bins, and dividing it by $G$ gives the total number of such bins, which defines the cost of the algorithm for these bins. 
    \item Consider the B bins. The number of B bins is simply $\frac{S_B - B_B}{G}$, and thus their total cost is $\frac{S_B - B_B}{G} + \beta B_B$. 
    \item Consider GBR bins. The number of such bins is $\frac{S_{GBR} + A_{GBR} - B_{GBR} }{G}$; this is because $S+A-B$ is the total green space across all bins. The cost of the algorithm for these bins is thus $\frac{S_{GBR} + A_{GBR} - B_{GBR} }{G} + \beta B_{GBR}$. 
    \item Consider GBS bins. The algorithm opens exactly $2m$ such bins and uses a total \black space of $\Gamma + B_{GBS}$. Therefore, the cost of the algorithm is $2m+ \beta (\Gamma + B_{GBS})$.
\end{itemize}
From the above discussion, we conclude that the total cost of \TNF is $\frac{S_{GG}+A_{GG}}{G} + \frac{S_B - B_B}{G} + \beta B_B + \frac{S_{GBR} + A_{GBR} - B_{GBR} }{G} + \beta B_{GBR} + 2m+ \beta (\Gamma + B_{GBS})$. On the other hand, Opt places each special item in a separate bin for a total cost of $m+ \beta \Gamma$ (there are $m$ such items and a \black space of total size $\Gamma$ must be used); other items are placed in the green space of at least $\frac{S_{GG}+S_B+S_{GBR}+S_{GBS}}{G}$. Thus, the cost of Opt is at least $\frac{S_{GG}+S_B+S_{GBR}+S_{GBS}}{G} + m+ \beta \Gamma $. We conclude the competitive ratio of \TNF is at most: 
\\
\newcommand\numberthis{\addtocounter{equation}{1}\tag{\theequation}}
\scalebox{.9}{
\begin{minipage}{\textwidth}
\begin{align*}
 \ &   \frac{\frac{S_{GG}+A_{GG}}{G} + \frac{S_B - B_B}{G} + \beta B_B + \frac{S_{GBR} + A_{GBR} - B_{GBR} }{G} + \beta B_{GBR} + 2m+ \beta (\Gamma + B_{GBS})}{\frac{S_{GG}+S_B+S_{GBR}+S_{GBS}}{G} + m+ \beta \Gamma } \\ 
 = &   \frac{(S_{GG}+A_{GG}) + (S_B - B_B + \beta G B_B) + (S_{GBR} + A_{GBR} - B_{GBR}  + \beta G B_{GBR}) + (2m G+ \beta G (\Gamma + B_{GBS}))}{S_{GG}+S_B+S_{GBR}+(S_{GBS} + mG + \beta G \Gamma )} 
 \\
 \leq & \max \{  \frac{S_{GG}+A_{GG}}{S_{GG}}, \frac{S_B - B_B + \beta G B_B}{S_{B}}, \frac{S_{GBR} + A_{GBR} - B_{GBR}  + \beta G B_{GBR}}{S_{GBR}}, \frac{2m G+ \beta G (\Gamma + B_{GBS})}{S_{GBS} + mG + \beta G \Gamma} \} \numberthis \label{exp:nftUB}.
\\ 
\end{align*}
\end{minipage}}

Note that our discussion so far has been independent of $\tau$. We review the above terms in the $\max$ function and bound them above by a function of $\tau$. 
\begin{restatable}{proposition}{tnfUBII}\label{prop:tnf_ub_2}
    Suppose $\beta G > 1$ and $\tau \in [0,G]$, then the following relations hold:
    \begin{itemize}
        \item $\frac{S_{GG}+A_{GG}}{S_{GG}} \le \frac{2G}{G+\tau}$
        \item $\frac{S_B - B_B + \beta G B_B}{S_{B}} \le \frac{G(1+\tau \beta)}{G+\tau}$
        \item $\frac{S_{GBR} + A_{GBR} - B_{GBR}  + \beta G B_{GBR}}{S_{GBR}} \le \frac{G(2+\tau \beta )}{G+2\tau}$ 
        \item $\frac{2m G+ \beta G (\Gamma + B_{GBS})}{S_{GBS} + mG + \beta G \Gamma} \le \frac{G(2+\tau \beta )}{G+2\tau} \;$ when \; $1<\beta G \le 2$
        \item $\frac{2m G+ \beta G (\Gamma + B_{GBS})}{S_{GBS} + mG + \beta G \Gamma} \le \frac{2+\tau\beta}{1+\tau\beta }$ \; when \; $\beta G \ge 2$
    \end{itemize}
\end{restatable}

\begin{proof}
For each of the terms in the max function, we will bound them above by a function of $\tau, \beta, G$. Suppose $\tau = \alpha G$ for some $\alpha > 0$.

\begin{itemize}
    \item Consider the $i$'th pair of GG bins. Let $s_i$ denote the total size of items inside the two bins and $a_i$ denote the wasted green space. 
    We know that the second bin is opened because the total size of the items inside the two bins is more than $G+\tau = (1+\alpha) \tau$, i.e., $s_i \geq (1+\alpha)\tau$. The total green space inside the two bins is $2G$; that is, the wasted space within the two bins is $a_i \leq 2G - (G+\tau) = G - \tau = (1-\alpha) \tau$.   Thus, we have ${a_i} \leq \frac{1-\alpha}{1+\alpha} s_i$. Summing over all pairs $i$, we get ${A_{GG}}\leq \frac{1-\alpha}{1+\alpha} S_{GG}$. We can conclude that $\frac{S_{GG}+A_{GG}}{S_{GG}} \leq 1 + \frac{1-\alpha}{1+\alpha} = \frac{2}{1+\alpha}$.
    \item Consider the $i$'th \black bin and let $s_i$ be the total size of items and $b_i$ be the \black space used in the bin.     
    Given that the \black space is used after the green space, we can write $\frac{b_i}{s_i} \leq \frac{\alpha G}{G + \alpha G} = \frac{\alpha}{1+\alpha}$. Summing over all values of $i$, we can write $B_{B} \leq \frac{\alpha}{1+\alpha} S_B$. Therefore, $\frac{S_B - B_B + \beta G B_B}{S_{B}} \le 1 +\frac{\alpha}{1+\alpha}(\beta G-1) = \frac{1+\alpha \beta G}{1+\alpha}$.

    \item Consider the $i$'th pair of BGR bins. As before, let $s_i$ denote the total size of items inside the two bins, $a_i$ denote the wasted green space in the green bin, and $b_i$ denote the used \black space in the \black bin. Given that the \black bin is opened by an item smaller than $G$ (otherwise, the pair was special), we can write that $(G-a_i) + G > G+ \tau $ or $a_i < G-\tau = (1-\alpha)G$. 
    Moreover, we have $b_i \leq \alpha G$ (threshold definition).
    Therefore, for the ratio $r =\frac{a_i + b_i (\beta G-1)}{s_i}$, we can write $r = \frac{a_i + b_i (\beta G-1)}{(G-a_i) + b_i + G}$, which is increasing in $a_i$ and thus $r \leq \frac{(1-\alpha)G + b_i (\beta G-1)}{(\alpha+1)G + b_i} $. This is also increasing in $b_i$ (since $\beta G >1$) and thus $r\leq \frac{(1-\alpha)G + \alpha G (\beta G-1)}{(\alpha+1)G + \alpha G} = \frac{(1-\alpha) + \alpha (\beta G-1)}{(\alpha+1) + \alpha} = \frac{1-2\alpha + \alpha \beta G}{2\alpha + 1}$. So far, we have shown that $a_i + b_i (\beta G-1) \leq  s_i (\frac{1-2\alpha + \alpha \beta G}{2\alpha + 1})$. Summing over all $i$'s, we get $A_{GBR} + B_{GBR}(\beta G-1) \leq S_{GBR} (\frac{1-2\alpha + \alpha \beta G}{2\alpha + 1}) $. To conclude, we have $\frac{S_{GBR} + A_{GBR} - B_{GBR}  + \beta G B_{GBR}}{S_{GBR}} \leq  1 + \frac{1-2\alpha + \alpha \beta G}{2\alpha + 1} = \frac{2+\alpha \beta G}{2\alpha + 1}$.

\item Consider the $i$'th pair of BGS bins. Let $s_i$ denote the total size of non-special items inside the two bins, and $b_i$ be the total \black space used in the \black bin. Suppose the special item that caused the opening of the second bin has size $G+\gamma_i$. So, we have $(s_i-b_i) + G + \gamma_i > (1+\alpha)G$; this is because $s_i-b_i$ is the space used in the green bin. Therefore, $s_i - b_i + \gamma_i > \alpha G $ or $s_i > \alpha G + b_i - \gamma_i$. Summing over all $i$'s, we get $S_{GBS} \geq \alpha G m + B_{GBS} - \Gamma$. Therefore, for the ratio  
$r = \frac{2m G+ \beta G (\Gamma + B_{GBS})}{S_{GBS} + mG + \beta G \Gamma}$, we can write $r\leq \frac{2m G+ \beta G (\Gamma + B_{GBS})}{\alpha G m + B_{GBS} - \Gamma + mG + \beta G \Gamma} = \frac{2m G+ \beta G  B_{GBS} + \beta G \Gamma}{ (1+\alpha) mG  + B_{GBS} + (\beta G-1) \Gamma} $. 
Moreover, back to the $i$'th pair, we have $b_i + \gamma_i \leq \alpha G$ (they both use the \black space that is at most $\alpha G$). Summing over all $i$'s, we get $B_{GBS} + \Gamma \leq \alpha m G$ or $B_{GBS} \leq \alpha m G - \Gamma$. Given that the above upper bound is increasing in $B_{GBS}$ (since $\beta G >1$), we can write $r\leq \frac{2m G+ \beta G  (\alpha m G - \Gamma) + \beta G \Gamma}{ (1+\alpha) mG  + (\alpha m G - \Gamma) + (\beta G-1) \Gamma} = \frac{2m+ \beta  \cdot \alpha m G }{ (1+2\alpha) m  + (\beta-2/G) \Gamma}$. This function is monotone in $\Gamma$ and takes its maximum either at $\Gamma = 0$ or the largest value of $\Gamma$, which is $\alpha m G$. In the former case, the upper bound becomes $\frac{2 + \alpha \beta G}{1+2\alpha}$ (similar to the previous case), and in the former case, it becomes $\frac{2 + \alpha \beta G}{(1+2\alpha) + (\beta -2/G)\alpha G} = \frac{2 + \alpha \beta G}{(1+2\alpha) +  (\alpha \beta G -2\alpha)} = \frac{2+\alpha \beta G}{1+\alpha \beta G}$. 
\end{itemize}
\end{proof}

\autoref{prop:wpc bound} and \autoref{prop:tnf_ub_2} imply that the competitive ratio is at most \\  $\max\{ \frac{2}{1+\alpha}, \frac{1+\alpha \beta G}{1+\alpha}, \frac{2+\alpha \beta G}{2\alpha + 1}, \frac{2+\alpha \beta G}{1+\alpha \beta G} \}$. The maximum, however, is never realized by $\frac{2}{1+\alpha}$. We can then substitute out $\alpha = \tau/G$ to complete the proof of \autoref{thm:tnf}.

\subsection{Upper bound of \TWF}

\TWFUB*

We use weighting argument to prove the upper bound. Different from \TAAF, we define the weight of each item equal its size, i.e., $w(x) = x$. Then we derive an upper bound by taking the ratio of the maximum \WPC of \Opt to the minimum \WPC of \TWF.

First, consider \WPC of \Opt. Fix a bin $B^*$ of Opt and suppose $B^*$ has a level $pG$ for some $p>0$. If $p < 1$, one can add items of small size to increase the weight of $B^*$ without changing its cost. If $p > 1$, the weight per unit cost would be $\frac{pG}{1+\beta(p-1)G}$, which is decreasing in $p$ and is maximized at $p=1$, in which case the which case it will be $G$.

Now consider the \WPC of \TWF. Fix a bin $B$ of \TWF and suppose it has a level $p(G+\tau)$ for $p \in [\nicefrac{1}{2},1]$. If $p(G+\tau) \le G$, then the cost of the bin is $1$, and \WPC is minimized when $p=\nicefrac{1}{2}$. If $p(G+\tau) > G$, then the cost is $1+\beta(p(G+\tau)-G))$, for a \WPC of $\frac{p(G+\tau)}{1+\beta(p(G+\tau)-G)}$. This function is decreasing in $p$, and reaches its minimum at $p=1$.

Finally, we can take the ratio of the maximal \WPC for \Opt and the two minimal \WPC of \TWF to reveal the upper bounds $\frac{2G}{G+\tau}$ and $\frac{G(1+\tau\beta)}{G+\tau}$.

\subsection{Upper Bound of \TAAF and \GHAR}
\label{app:taaf-ub}

\TAAFUB*

In order to provide upper bounds for \TAAF and \GHAR, we will use the following weighting argument. For an item of size $x$, let its weight be $x$ if $x\le \frac{G+\tau}{2}$, and $x+R$ if $x>\frac{G+\tau}{2}$, where $R \in[0,\frac{G+\tau}{6}]$. To derive \CR bounds, we will provide lower bounds on the weight per cost (\WPC) of $\Alg$, and upper bounds on the \WPC for $\Opt$. Since both $\Alg$ and $\Opt$ will take the same total weight, we can upper bound the \CR by dividing the $\Opt$ \WPC by the $\Alg$ \WPC. To prove \autoref{thm:taf-UB} we will first establish a number of intermediary results.

\taafUBlemI*

\begin{proof}
    Let large items be those with size $>\frac{G+\tau}{2}$. We will show that in both cases, all bins (except a constant number) without large items will fill to a level of at least $\frac{2(G+\tau)}{3}$. For \TAAF algorithms, if some bin fills to less than $\frac{2(G+\tau)}{3}$, then we know that for each bin after, items are all of size greater than $\frac{G+\tau}{3}$, and each bin (except perhaps the last) will have two items with fill above $\frac{2(G+\tau)}{3}$. Thus, at most two bins without large items are filled below $\frac{2(G+\tau)}{3}$. For \GHAR, $i$-type bins for $2\le i \le K$ contain only small items. Each closed $i$-type bin fills to a level of at least $\frac{i(G+\tau)}{i+1}$ for $2 \le i \le K-1$, or $\frac{(K-1)(G+\tau)}{K}$ for $K$-type bins, since $K \ge 3$ this value is at least $\frac{2(G+\tau)}{3}$. Recall that there are at most $K$ open bins which may be filled lower than this level. Since we are considering asymptotic competitive ratio, the constant number of bins filled below $\frac{2(G+\tau)}{3}$ will not affect the analysis. Now considering bins with large items: for both algorithms, the minimum fill level for a bin with large items is $\frac{G+\tau}{2}$. Thus, for both \TAAF and \GHAR, each bin (except at most a constant number) will be filled to a \textit{weight} of at least $\frac{G+\tau}{2}+R$.

We can now make two observations. First, if the fill level of an $\Alg$ bin is below $G$, then its cost is $1$, and its minimum \WPC is $\frac{G+\tau+2R}{2}$. Second, if the fill level is a value $L > G$, then its cost is $1+\beta(L - G)$, and \WPC is minimized when $L = G+\tau$, for a value of $\frac{G+\tau}{1+\beta\tau}$. Depending on the value of $R$, one of these expressions will be lower than the other at a given $\tau$. If we denote the value of $\tau$ where they intersect as $\mu_1(R) = \frac{(1-G\beta-2R\beta)+\sqrt{(G\beta+2R\beta-1)^2+4\beta(G-2R)}}{2\beta}$, then we can see that $\frac{G+\tau+2R}{2}$ is smaller for $\tau <\mu_1$ and $\frac{G+\tau}{1+\tau\beta}$ is smaller for $\tau > \mu_1$. 
\end{proof}

\taafUBlemII*

\begin{proof}
For each possible fill level $L$, we will consider the maximal \WPC. First, if $L \le \frac{G+\tau}{2}$, then \WPC at most $\frac{G+\tau}{2}$. Next, if $\frac{G+\tau}{2} < L \le G$, \WPC maximized when $L = G$ and \WPC is $G+R$. Thirdly, if $G+\tau < L \le 1$, \WPC maximized when $L = G+\tau$ with $\frac{G+\tau+2R}{1+\beta \tau}$, but this is only possible if $\tau < 1-G$. Thus, if $\tau = 1-G$, the maximum \WPC for $\Opt$ is $G+R$. However, for $\tau <1-G$, the maximum \WPC for $\Opt$ will either be $G+R$ or $\frac{G+\tau+2R}{1+\tau \beta}$. These expressions are equal at $\tau = \mu_2(R) = \frac{R}{\beta(G+R)-1}$. The $G+R$ case occurs for $\tau > \mu_2$ and the $\frac{G+\tau+2R}{1+\tau \beta}$ for $\tau <\mu_2$.    
\end{proof}

With these results, we can find expressions for the upper bound with respect to a given choice of $R$. The specific expressions depend on the relation between $\mu_1(R)$ and $\mu_2(R)$, but we show that we only need to consider the case that $\mu_1(R) > \mu_2(R)$.

\weightingGroup*

\begin{proof}
Recall $\mu_1(R) = \frac{(1-\beta G-2RB)+\sqrt{(G\beta+2R\beta-1 )^2+4\beta(G-2R)}}{2\beta}$ and $\mu_2(R) = \frac{R}{\beta(G+R)-1}$.
    We will first show that when $R=0$, $\mu_1(0) > \mu_2(0)$. Then we will show that $\mu_1(R)$ is decreasing in $R$, while $\mu_2(R)$ is increasing in $R$. With these established, we can say that if $\mu_1(\frac{G+\tau}{6}) > \mu_2(\frac{G+\tau}{6})$, then $\mu_1(R) > \mu_2(R)$ for all $R \in [0, \frac{G+\tau}{6}]$. Since $\tau$ varies, we will take the largest value of interest: $\nicefrac{1}{\beta}$, as larger $\tau$ are never optimal. 

    Consider $R=0$, a value $\mu_2(0) = 0$, while $\mu_1(0) = \nicefrac{1}{\beta} > 0$. Now consider the derivatives of each. First, it is clear that $\frac{d}{dR}\mu_2(R) = \frac{\beta G-1}{(\beta(G+R)-1)^2}$, which will be positive since $\beta G >1$ and $R \ge 0$. Now consider $\frac{d}{dR}\mu_1(R) = -1 +\frac{\beta G + 2RB - 3}{\sqrt{(G\beta+2R\beta-1 )^2+4\beta(G-2R)}}$. This expression is positive only when $\beta G < 1$ and $\beta G + 2RB > 3$. By the restriction on $\beta G$ in the statement of the Lemma, the first condition does not hold, establishing that $\mu_1(R)$ is decreasing in $R$.

    Finally, we will show $\mu_1(\frac{G+1/\beta}{6}) > \mu_2(\frac{G+1/\beta}{6})$. First, we have 
    $\mu_1(\frac{G+1/\beta}{6})=\frac{1-\beta G+\sqrt{4\beta^2 G^2+2\beta G - 2}}{3\beta}$
    Meanwhile, 
    $\mu_2(\frac{G+1/\beta}{6}) =\frac{G+1/\beta}{7\beta G}$. We can now write:

    \begin{align*}
        \frac{1-\beta G+\sqrt{4\beta^2 G^2+2\beta G - 2}}{3\beta} >& \frac{G+1/\beta}{7\beta G} \\
         7\beta G\left(1 - 2\beta G + \sqrt{4(\beta G)^2 + 2\beta G - 2}\right) >& 3\beta (G + 1/\beta ) \\
         7\beta G  \sqrt{4(\beta G)^2 + 2\beta G - 2} >& 14(\beta G)^2-4\beta G+3 \\
         \text{Now square both sides and simplify:}\\
         196(\beta G)^2 + 98(\beta G)^3-98(\beta G)^2 >&
         196(\beta G)^4-112(\beta G)^3+100(\beta G)^2-24(\beta G)+9 \\
         210(\beta G)^3-198(\beta G)^2+24\beta G-9 >& 0
    \end{align*}
    Note that we are able to square to remove the radical because $\beta G > 1$, thus $4(\beta G)^2+2\beta G-2>0$ and $14(\beta G)^2-4\beta G+3>0$. We can then use numerical methods to find the cutoff for $\beta G$ as approximately $0.868$. Since $\beta G >1$, we can then conclude that  $\mu_1(\frac{G+1/\beta}{6}) > \mu_2(\frac{G+1/\beta}{6})$. Thus, we have proved the statement of the Lemma.

\end{proof}

With that result established, we can now present the upper bounds with respect to a given value of $R$. We consider $\tau < 1-G$ (left) and $\tau = 1-G$ (right).

$$CR \le \begin{cases}
    \frac{\frac{G+\tau+2R}{1+\tau\beta}}{\frac{G+\tau+2R}{2}} = \frac{2}{1+\tau \beta} & 0 \le \tau \le \mu_2 \\
    \frac{G+R}{\frac{G+\tau+2R}{2}} = \frac{2(G+R)}{G+\tau+2R} & \mu_2 \le \tau \le \mu_1 \\
    \frac{G+R}{\frac{G+\tau}{1+\tau\beta}} = \frac{(G+R)(1+\tau\beta)}{G+\tau} & \mu_1 \le \tau \le 1-G \\
\end{cases}, \text{ and } CR \le \begin{cases}
    \frac{G+R}{\frac{1+2R}{2}} = \frac{2(G+R)}{1+2R} & 0\le  1-G \le \mu_1 \\
    \frac{G+R}{\frac{1}{1+\tau\beta}} = \frac{(G+R)(1+\beta(1-G))}{1} & \mu_1 \le 1-G \\
\end{cases}$$

We will show that as we vary $R$, our we achieve different upper bounds which result in the bounds stated in \autoref{thm:taf-UB}.

First consider $R = 0$, we see that $\mu_1 = \frac{1}{\beta}$, and so we have the upper bound $\frac{G(1+\tau\beta)}{G+\tau}$ when $\tau \ge 1/\beta$, which matches our expectation. This means that we have tight bounds for $\tau \ge 1/\beta$. When $\tau = 1-G$, this function becomes $G(1+\beta(1-G))$.

 Next consider $0 < R \le \frac{G+\tau}{6}$. We will see that the intersection of the second ($\frac{2(G+R)}{G+\tau+2R}$) and third ($\frac{(G+R)(1+\tau\beta)}{G+\tau}$) bounds will \textit{trace} along the $1+\frac{(G-\tau)(1+\tau \beta)}{G+\tau}$ bound. By construction, we know that the second and third case bounds intersect at $\tau = \mu_1$. We will show that both bounds also intersect with $1+\frac{(G-\tau)(1+\tau \beta)}{G+\tau}$ at $\tau = \mu_1$.

\begin{align*}
    \frac{(G+R)(1+\tau\beta)}{G+\tau} &= 1+\frac{(G-\tau)(1+\tau\beta)}{2(G+\tau)}\\
    2(G+R)(1+\tau\beta) &= 2(G+\tau)+(G-\tau)(1+\tau\beta) \\
    2(G+R+G\tau\beta+R\tau\beta) &= 2G+2\tau+G-\tau+G\tau\beta-\tau^2\beta \\
    \tau^2(\beta)+\tau(G\beta+2R\beta-1)+(2R-G) &=0 \\
    \tau &= \mu_1
\end{align*}

As $R$ increases, $\mu_1$ decreases, reaching a minimum value when $R = \frac{G+\tau}{6}$. At this value, the second bound will become $\frac{7G+\tau}{4(G+\tau)}$, which intersects with $1+\frac{(G-\tau)(1+\tau \beta)}{G+\tau}$ at $\tau = 1/2\beta$. Therefore, we can say that $1+\frac{(G-\tau)(1+\tau \beta)}{G+\tau}$ is tight from $1/2\beta \le \tau \le 1/\beta$. When $\tau = 1-G$, this function becomes $1+\frac{(2G-1)(1+\beta(1-G))}{2}$.

Since we can no longer increase $R$, we are left with a loose upper bound $\frac{7G+\tau}{4(G+\tau)}$ between either its intersection with $\frac{2}{1+\tau \beta}$  ($\tau <1-G$) or 0 ($\tau =  1-G$) and $1/2\beta$. For $\tau <1-G$, this intersection occurs at $\frac{7-7\beta G +\sqrt{(7\beta G-7)^2+4\beta G}}{2\beta}$. Therefore, if $\tau <1-G$ we conclude that $\frac{2}{1+\tau \beta}$ is a tight bound $0 \le \tau \le \frac{7-7\beta G +\sqrt{(7\beta G-7)^2+4\beta G}}{2\beta}$ and $\frac{7G+\tau}{4(G+\tau)}$ is an upper bound from $\frac{7-7\beta G +\sqrt{(7\beta G-7)^2+4\beta G}}{2\beta} \le \tau \le 1/2\beta$. Meanwhile, if $\tau = 1-G$, we have the bound $\frac{6G+1}{4}$ for $0 \le 1-g \le 1/2\beta$. Thus, we establish the statement of \autoref{thm:taf-UB}.

\subsection{Proofs of Corollaries in \autoref{sec:GBgtOne}}

\tnfOPT*

\begin{proof} 

Its easy to check that the $\frac{G(2+\tau \beta)}{G+2\tau}$ and $ \frac{2+\tau \beta}{1+\tau\beta}$ bounds are equivalent when $\beta G =2$. 

We find that that when $\beta G < 2$, $\frac{G(2+\tau \beta)}{G+2\tau} < \frac{2+\tau \beta}{1+\tau\beta}$ for $\tau \in [0, 1-G]$, so the minimum competitive ratio occurs at the intersection of $\frac{2+\tau \beta}{1+\tau\beta}$ and $\frac{G(1+\beta \tau)}{G+\tau}$. This intersection occurs at $\tau = \frac{2-\beta G + \sqrt{5\beta^2 G^2 - 8\beta G + 4}}{2\beta(\beta G-1)}$, which we can substitute back in to find a competitive ratio of $\frac{3\beta G-2 +\sqrt{5B^2G^2-8\beta G +4}}{\beta G +\sqrt{5B^2G^2-8\beta G +4}}$.

Meanwhile, when $\beta G > 2$, $\frac{G(2+\tau \beta)}{G+2\tau} > \frac{2+\tau \beta}{1+\tau\beta}$ for $\tau \in [0, 1-G]$. The intersection of $\frac{G(2+\tau \beta)}{G+2\tau}$ and $\frac{G(1+\beta \tau)}{G+\tau}$ occurs when $\tau = \sqrt{G/\beta}$, which we can substitute to find a competitive ratio of $\frac{\beta G (1+ \sqrt{\beta G})}{\beta G +\sqrt{\beta G}}$.
\end{proof}

\optimalTauOne*

\begin{proof}

\textbf{Case 1:} $1 <\beta(1-G)$.

To establish this result, we begin by identifying critical points on the upper and lower bounds which may minimize the competitive ratio based on $\beta G$. These critical points are the boundary conditions ($\tau \in \{0, 1-G\}$) and the intersection points between bounds: these are $\tau \in \{ \frac{7-7\beta G+\sqrt{(7\beta G-7)^2+4\beta G}}{2\beta},\nicefrac{1}{2\beta}, \nicefrac{1}{\beta}\}$ for upper bounds and $\tau \in \{\hat{\tau}, \nicefrac{1}{\beta}\}$ for the lower bounds. The restriction that $1<\beta(1-G)$ ensures that each of these critical points are valid values of $\tau$. We can then consider the slope of each bound to filter out trivially non-competitive points. This removes the boundary points in both cases, and $\tau = \frac{7-7\beta G+\sqrt{(7\beta G-7)^2+4\beta G}}{2\beta}$ for the upper bound.

This process leaves us with two critical points which can minimize the upper bound of the competitive ratio: $\tau = \nicefrac{1}{2\beta}$ and $\tau = \nicefrac{1}{\beta}$. Similarly, there are two critical points which can minimize the lower bound of the competitive ratio: $\tau = \hat{\tau}$ and $\tau = \nicefrac{1}{\beta}$. Note that we have tight bounds for $\tau \ge \nicefrac{1}{2\beta}$, while the upper bound of $\tau=\hat{\tau}$ is loose. We will consider three ranges of $\beta G$ and discuss the appropriate $\tau$ in each range.

\textbf{Case 1.1}. Let $1< \beta G \le 1+\sqrt[3]{2+\sqrt{44/27}}+\sqrt[3]{2-\sqrt{44/27}}$. We will show that in this range, $\tau = \nicefrac{1}{\beta}$ \underline{is optimal}. 

We begin by finding that $\tau= \nicefrac{1}{\beta}$ gives a tight competitive ratio of $\frac{2\beta G}{\beta G + 1}$. To say that $\tau= \nicefrac{1}{\beta}$ is optimal in this range, we must show that the lower bound on competitive ratio is higher at every other $\tau$. The other critical point to examine is $\hat{\tau}$, which is the intersection of the lower bounds lower bounds $\frac{2}{1+\tau\beta}$ and $1+\frac{(G-\tau)(1+\tau\beta)}{2(G+\tau)}$. 

Since directly evaluating $\hat{\tau}$ is complex, we use an indirect argument. We can consider when the bounds $\frac{2}{1+\tau\beta}$ and $1+\frac{(G-\tau)(1+\tau\beta)}{2(G+\tau)}$ intersect with the line $\frac{2 \beta G}{\beta G +1 }$. If $\frac{2}{1+\tau\beta}$ intersects at a larger $\tau$ than $1+\frac{(G-\tau)(1+\tau\beta)}{2(G+\tau)}$ does, then this means that the competitive ratio lower bound at $\hat{\tau}$ is greater than $\frac{2 \beta G}{\beta G +1 }$.

The $\frac{2}{1+\tau\beta}$ lower bound achieves the value $\frac{2 \beta G}{\beta G +1 }$ at $\tau = \frac{1}{\beta^2G}$. Meanwhile the $1+\frac{(G-\tau)(1+\tau\beta)}{2(G+\tau)}$ bound achieves it when
$\tau = \frac{G(\beta G-3)}{\beta G+1}$. Therefore, we can conclude that $\tau = 1/\beta$ is optimal when $\frac{1}{\beta^2G}\ge\frac{G(\beta G-3)}{\beta G+1}$. The inequality holds when $\beta G \le 1+\sqrt[3]{2+\sqrt{44/27}}+\sqrt[3]{2-\sqrt{44/27}}$, which is approximately $\beta G \le 3.383$.

\textbf{Case 1.2}. Let $1+\sqrt[3]{2+\sqrt{44/27}}+\sqrt[3]{2-\sqrt{44/27}}< \beta G \le \frac{7}{4}+\frac{\sqrt{57}}{4}$. We will show that in this range, $\tau = \nicefrac{1}{\beta}$ \underline{minimizes the upper bound}.

To do so, we must show that in this range, $\tau = \nicefrac{1}{\beta}$ achieves a lower competitive ratio than the other critical point, $\tau = \nicefrac{1}{2\beta}$. We can see that the competitive ratio at $\tau = \nicefrac{1}{2\beta}$ is exactly $\frac{14\beta G+1}{8\beta G+4}$. Simple algebra shows that $\frac{2\beta G}{\beta G + 1}$ is smaller than $\frac{14\beta G+1}{8\beta G+4}$ when $\beta G \le \frac{7}{4}+\frac{\sqrt{57}}{4}$

\textbf{Case 1.3}. Let $ \beta G \ge \frac{7}{4}+\frac{\sqrt{57}}{4}$. We will show that in this range, $\tau = \nicefrac{1}{2\beta}$ \underline{minimizes the upper bound}.

We can use the same argument as Case 1.2, except now $\frac{14\beta G+1}{8\beta G+4}$ is smaller in this range. 

\textbf{Case 2} $ 1/2\beta\le1-G < 1/\beta$.
    Since $1/2\beta \le 1-G <1/\beta$,  $\frac{1}{\beta}$ is not a possible value of $\tau$. As seen previously, if $\tau < 1/2\beta$, the upper bound is $\frac{7G+\tau}{4(G+\tau)}$, which is decreasing in $\tau$. Thus, the optimized value of $\tau\in[1/2\beta, 1-G]$ and must lie along the bound $1+\frac{(G-\tau)(1+\tau\beta)}{2(G+\tau)}$. Since this bound is concave in $\tau$, the minimal value will occur at either $\tau = 1/2\beta$  or $\tau=1-G$, with \CR $\frac{14\beta G+1}{8\beta G+4}$ or $1+\frac{(2G-1)(1+\beta(1-G))}{2}$, respectively.

\textbf{Case 3} $ 1>2\beta(1-G)$. In this range neither $\frac{1}{\beta}$ nor $\frac{1}{2\beta}$ are possible values of $\tau$. Thus, the only possible upper bounds is $\frac{7G+\tau}{4(G+\tau)}$ which is decreasing in $\tau$. We can minimize the bound by selecting $\tau = 1-G$.
\end{proof}

\section{Visualization Across $\beta, G$}
\label{Appendix-Visual}
\begin{figure}[h!]
    \centering
    \includegraphics[width=1\linewidth]{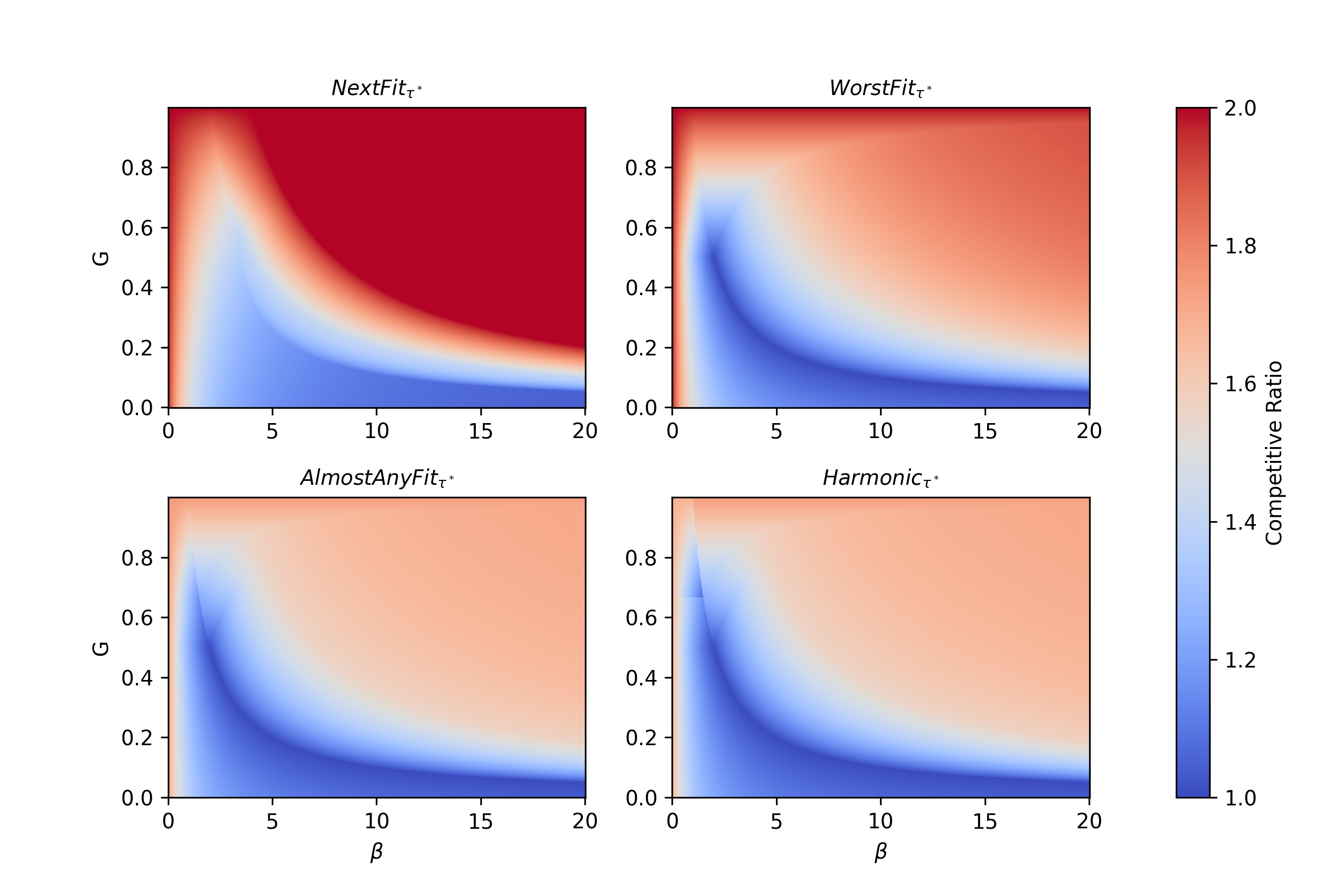}
    \caption{Competitive Ratio upper bounds for all algorithms using the optimized $\tau$, across $G\in[0,1]$ and $\beta \in [0,20]$}. 
    \Description{Four heatmap-style plots corresponding to \TNF, \TWF, \TAAF,  and \GHAR where the \CR upper bound is plotted as a color value for each $\beta,G$. All plots show low values near $\beta G = 1$ for $G\le 2/3$, with \CR increasing relatively smoothly outwards. \TNF quickly reaches 2-competitive, while the other algorithms only achieve maximal values for extremal values of $\beta,G$.}
    \label{fig:allGB}
\end{figure}

\autoref{fig:allGB} shows the optimized upper bound for each of the algorithms considered in this paper. The results show that while $\beta G \le 1$ and $\beta G >1$ cases require separate analysis, the resulting bounds generally transition across this boundary quite smoothly. This figure also highlights the separation between \TNF and \TWF. \TAAF and \GHAR are identical for $\beta G > 1$; \GHAR is slightly better when $\beta G \le 1$, but the difference is difficult to see in this visualization. Finally, when $G>1/2$ and $\beta G >1$, each of the algorithms except for \TNF demonstrates slight regions of lower \CR along $1 > \beta(1-G)$ and $\beta = \frac{G-1/2}{3G-2G^2-1}$, corresponding to the boundaries between worst-case expressions.

\section{Proof of General Lower Bound}
\label{Appendix-GeneralLB}
In this section, we present the proofs for our general lower bound on competitive ratio for any deterministic online algorithm.

\subsection{$\beta G\le 1$}

\genLBgbLtOne*
\begin{proof}
Consider the input used in the classic bin packing setting~\cite{BALOGH2012lowerBound, LIANG1980, VANVLIET1992}: $(\sigma_1,\dots,\sigma_i,\dots,\sigma_k)$, where $\sigma_i = (\frac{1}{m_{k-j+1}+\epsilon})^n$ and $m_i$ is the $i$-th number in the Sylvester sequence \cite{Sylvester1880}.

Since $\beta G \le 1$, it is optimal to fill each bin maximally. Thus, under the input, \Opt uses $n$ bins while \Alg uses at least $r n$ bins, where $r = \nicefrac{248}{161}$.
Thus, cost of \Opt on $n$ bins is at most $n(1+\beta(1-G))$. Meanwhile, the minimal cost for \Alg occurs when the volume is evenly spread across the $rn$ bins, for a cost $rn+rn\beta(\frac{n}{rn}-G)^+ = rn+n\beta(1-rG)^+$. 
    
    This gives a lower bound of $\frac{r+\beta(1-rG)^+}{1+\beta(1-G)}$ on the competitive ratio of any online algorithm for \gbp when $\beta G \le 1$.

\end{proof}

\subsection{$\beta G >1$}

\generalLB*

The first function in this bound
is $\frac{r+1-G}{1+\beta(1-G)}$, which is derived in the same way as \autoref{thm:GenLB-gbLtOne}. Since $\beta G >1$, it may no longer be optimal to fill each bin to capacity, so we will denote the cost $\tilde{\Opt}$. So long as $\Alg>\tilde{\Opt}$, we still have a valid lower bound as $\Alg>\tilde{\Opt}\ge\Opt$. If $\tilde{\Opt}$ uses $n$ bins, its cost is at most $n(1+\beta(1-G))$. Meanwhile, we lower bound the cost of $\Alg$ as $nr+nr(1/r-G) = nr+n(1-rG)$, since $\Alg$ has at least $n(1-G)$ volume which cannot be in the green space of the first $nr$ bins. Note that this differs from our proof of \autoref{thm:GenLB-gbLtOne}, as it may now be cheaper to place this volume in new bins rather than take the black cost. Thus, we realize the bound $\frac{r+1-G}{1+\beta(1-G)}$, which is only feasible if $r >G+\beta(1-G)$. The remainder of this section discusses the function $f(\beta G)$.

To clarify the proof we use the following terms to classify suboptimal packings. First, the term \textit{trivially suboptimal} describes a packing where any bin takes more than $1/\beta$ \black space from items with size $< G$, as the packing could be trivially improved by moving a subset of the items of that bin into a newly created bin. Comparatively, a \textit{weakly suboptimal} packing takes more cost than the offline optimal, but would require redistributing items from multiple bins, potentially into one or more new bins.

\subsubsection{$1 < \beta G \le 1.5$}

    Consider the inputs $\sigma_1 = \{\frac{G}{3},\ldots\}$, $\sigma_2=\{\frac{2G}{3},\ldots\}$, where each sequence contains $n$ items and $n$ is a multiple of three. We will consider the performance of $\Alg$ on $\sigma_1$ and $\sigma_1\sigma_2$.

    First consider the packing after $\sigma_1$. $\Alg$ can place at most five items of size $G/3$ in a given bin due to $\beta G \le 1.5$. However, placing more than three items is weakly suboptimal: three items per bin is optimal on $\sigma_1$, and placing four or 5 in a given bin does not help performance on $\sigma_1\sigma_2$. We can then assume that $\Alg$ will place at most 3 items per bin after $\sigma_1$. Let $s, d, t$ be the number of single, double, and triple item bins respectively. Note that $n= s+2d+3t$. The cost for $\Alg$ is simply $s+d+t$, while the optimal cost is $n/3$. Thus, we have $CR \ge \frac{3(s+d+t)}{n} = \frac{3(s+d+t)}{s+2d+3t}$.

        \begin{table}[h!]
    \centering
        \begin{tabular}{|c c c c}
        \hline
    & (1,0) & (2,0) & (3,0) \\
    (0,1) & (1,1) & (2,1) & (3,1)  \\ 
    (0,2) & (1,2)
\end{tabular}
    \caption{Case $1 < \beta G \le 1.5$. Possible bin types after $\sigma_1\sigma_2$. $(i,j)$ indicates $i$ items of size $G/3$ and $j$ items of size $2G/3$.}
    \label{table:genLB-items-2}
\end{table}

    Now consider the packing after $\sigma_2$. The minimum cost packing of $\Alg$ is to first add a single  $\frac{2G}{3}$ item to the $s$ single bins, then to the $d$ double bins, then to the $t$ triple bins. We will then have $n-s-t-d = 2t+d$ items remaining. We will then add a second item to $\min\{2t+d, s\}$ single bins, create $\lfloor\frac{(2t+d-s)^+}{2}\rfloor$ bins with only 2 items $2G/3$, and have at most one item left over to make a bin with only $2G/3$. The cost of this packing is: 
    $$s-\min\{2t+d,s\}+(d+\frac{(2t+d-s)^+}{2}) (1+\beta G/3)+(t+\min\{2t+d,s\})(1+2\beta G/3)$$

    Meanwhile, the optimal cost is $n = s+2d+3t$.
    
    We can think of these two bounds as functions of $s,d,t$. 
    $$f(s,d,t) = \frac{3(s+d+t)}{s+2d+3t}$$
    $$g(s,d,t) = \frac{s-\min\{2t+d,s\}+(d+\frac{(2t+d-s)^+}{2}) (1+\beta G/3)+(t+\min\{2t+d,s\})(1+2\beta G/3)}{s+2d+3t}$$
    
    The competitive ratio of $\Alg$ is thus lower bounded by:
    $$\lim_{n\rightarrow \infty} \min_{(s,d,t)}\max\{f(s,d,t),g(s,d,t)\}$$

    To analyze this problem, we will first normalize each of the competitive ratio functions. This will allow us to identify locations where $f=g$, which will minimize the maximum. 

    Let $\alpha = s/n$, $\delta = d/n$, $\gamma = t/n$, then we create new functions:
    $$F(\alpha,\delta,\gamma) = 3(\alpha+\delta+\gamma)$$
    $$G(\alpha,\delta,\gamma) = \alpha-\min\{2\gamma+\delta,\alpha\}+(\delta+\frac{(2\gamma+\delta-\alpha)^+}{2}) (1+\beta G/3)+(\gamma+\min\{2\gamma+\delta,\alpha\})(1+2\beta G/3)$$
    Note that we have the constraints that $\alpha, \delta,\gamma \ge 0$ and $\alpha+2\delta+3\gamma = 1$.

    To simplify $G$, let us consider two cases: $G_1$ where $2\gamma+\delta > s$ and $G_2$ where $2\gamma+\delta \le s$

    $$G_1 = \alpha/2(1+\beta G) + \delta/2(3+\beta G)+\gamma(2+\beta G)$$
    $$G_2 = \alpha+\delta(1+\beta G)+\gamma(1+2\beta G)$$

    We can now eliminate the variable $\alpha$ using the substitution $1-2\delta-3\gamma$. 
    $$F(\delta,\gamma) = 3(1-\delta-2\gamma)$$
    $$G_1(\delta,\gamma) = \frac{\beta G+1}{2}+(\frac{1-\beta G}{2})(\delta+\gamma)$$
    $$G_2(\delta,\gamma) = 1+\delta(\beta G - 1)+2\gamma(\beta G  -1) $$

    \textit{Case 1: $F= G_1, \; 2\gamma+\delta >\alpha$}
    \begin{align*}
        3(1-\delta-2\gamma) &= \frac{\beta G+1}{2}+(\frac{1-\beta G}{2})(\delta+\gamma) \\
        5 - \beta G&= \delta(7-\beta G) +\gamma(13-\beta G) 
    \end{align*}

    Therefore, we claim that the minimum is achieved when $$\delta = \frac{5-\beta G+\gamma(\beta G-13)}{7-\beta G}$$
    Since $\delta \ge 0$ and by restrictions on $\beta G$ this implies $\gamma \le \frac{\beta G-5}{\beta G-13}$. 

    We can then substitute into either $F$ or $G_1$ to find the competitive ratio as a function of $\gamma$.
    \begin{align*}
        F(\gamma) &= 3(1-\frac{5-\beta G+\gamma(\beta G-13)}{7-\beta G}-2\gamma) \\
        &=3(\frac{7-\beta G}{7-\beta G}-\frac{5-\beta G}{7-\beta G}-\gamma(\frac{\beta G-13}{7-\beta G}+2)) \\
        &= 3(\frac{2}{7-\beta G}-\gamma(\frac{1-\beta G}{7-\beta G}))\\
    \end{align*}

    Note that since $1<\beta G\le 1.5$, $F(\gamma)$ is increasing in $\gamma$. Thus, to minimize $F(\gamma)$ we use the smallest $\gamma$ which satisfies the constraints: $\gamma \ge 0$, $\delta = \frac{5-\beta G+\gamma(\beta G-13)}{7-\beta G}$, $2\delta+3\gamma \le 1$ and $2\gamma+\delta > \alpha$. This optimal value is $\gamma^* = \frac{3-\beta G}{\beta G+5}$, giving $\delta^* = \frac{2(\beta G-1)}{\beta G+5}$ and $\alpha^* = 0$. We can then plug these optimal values into either $F$ or $G_1$ to find a competitive ratio of  $\frac{3(\beta G+1)}{\beta G+5}$.

    \textit{Case 2: $F= G_2, \; 2\gamma+\delta \le \alpha$}
    \begin{align*}
        3(1-\delta-2\gamma) &= 1 + \delta(\beta G - 1)+2\gamma(\beta G-1)\\
        2 &= \delta(2+\beta G) +2\gamma(\beta G+2)
    \end{align*}

    Therefore, we claim that the minimum is achieved when 
    $$\delta = \frac{2-2\gamma(\beta G+2)}{\beta G+2}$$

    Since $\delta \ge 0$ this implies $\gamma \le \frac{1}{\beta G+2}$.

    We can then substitute into either $F$ or $G_2$ to find the competitive ratio as a function of $\gamma$.

    \begin{align*}
        F(\gamma) &= 3(1-\frac{2-2\gamma(\beta G+2)}{\beta G+2}-2\gamma) \\
        &=3(\frac{\beta G+2}{\beta G+2}-\frac{2}{\beta G+2}+\frac{2\gamma(\beta G+2)}{\beta G+2}-\frac{2\gamma(\beta G+2)}{\beta G+2}) \\
        &= 3(\frac{\beta G}{\beta G+2})
    \end{align*}

    In this case, the minimum possible competitive ratio is $\frac{3\beta G}{\beta G+2}$. Note that $\frac{3\beta G}{\beta G+2} \ge \frac{3(\beta G+1)}{\beta G+5}$ for $1 <\beta G \le 1.5$, meaning that Case 1 establishes our online lower bound for this range of $\beta G$.

    \subsubsection{$1.5 < \beta G \le 3$}

    This case is quite similar to the one above. Consider the inputs $\sigma_1 = \{\frac{G}{3},\ldots\}$, $\sigma_2=\{\frac{2G}{3},\ldots\}$, where each sequence contains $n$ items and $n$ is a multiple of three. We will consider the performance of $\Alg$ on $\sigma_1$ and $\sigma_1\sigma_2$.

    First consider the packing after $\sigma_1$. $\Alg$ can place at most four items of size $G/3$ in a given bin due to $1.5 < \beta G \le 3$. However, placing more than three items is weakly suboptimal: three items per bin is optimal on $\sigma_1$, and placing four in a given bin does not help performance on $\sigma_1\sigma_2$. We can then assume that $\Alg$ will place at most 3 items per bin after $\sigma_1$. Let $s, d, t$ be the number of single, double, and triple item bins respectively. Note that $n= s+2d+3t$. The cost for $\Alg$ is simply $s+d+t$, while the optimal cost is $n/3$. Thus, we have $CR \ge \frac{3(s+d+t)}{n} = \frac{3(s+d+t)}{s+2d+3t}$.

        \begin{table}[h!]
    \centering
        \begin{tabular}{|c c c c}
        \hline
    & (1,0) & (2,0) & (3,0) \\
    (0,1) & (1,1) & (2,1) &    \\ 
    (0,2) & & &
\end{tabular}
    \caption{Case $3 < \beta G \le 3$. Possible bin types after $\sigma_1\sigma_2$. $(i,j)$ indicates $i$ items of size $G/3$ and $j$ items of size $2G/3$.}
    \label{table:genLB-items-3}
\end{table}

    Now consider the packing after $\sigma_2$. The minimum cost packing of $\Alg$ is to first add a single $\frac{2G}{3}$ item to the $s$ single bins, and then to the $d$ double bins. We will then have $n-s-d = 3t+d$ items remaining. We will then create $\frac{d+3t}{2}$ bins with only 2 items $2G/3$. The cost of this packing is: 
    $$s+t+\frac{3(d+t)}{2}(1+\frac{\beta G}{3})$$

    Meanwhile, the optimal cost is $n = s+2d+3t$.
    
    We can think of these two bounds as functions of $s,d,t$. 
    $$f(s,d,t) = \frac{3(s+d+t)}{s+2d+3t}$$
    $$g(s,d,t) = \frac{s+t+\frac{3(d+t)}{2}(1+\frac{\beta G}{3})}{s+2d+3t}$$
    
    The competitive ratio of $\Alg$ is thus lower bounded by:
    $$\lim_{n\rightarrow \infty} \min_{(s,d,t)}\max\{f(s,d,t),g(s,d,t)\}$$

    To analyze this problem, we will first normalize each of the competitive ratio functions. This will allow us to identify locations where $f=g$, which will minimize the maximum. 

    Let $\alpha = s/n$, $\delta = d/n$, $\gamma = t/n$, then we create new functions:
    $$F(\alpha,\delta,\gamma) = 3(\alpha+\delta+\gamma)$$
    $$G(\alpha,\delta,\gamma) = \alpha+\gamma+\frac{3}{2}(\delta+\gamma)(1+\frac{\beta G }{3})$$

    We can use the constraint $\alpha+2\delta+3\gamma = 1$ to remove a variable.
    $$F(\delta,\gamma) = 3(1-\delta-2\gamma)$$
    $$G(\delta,\gamma) = 1 +(\delta+\gamma)(\frac{\beta G-1}{2})$$

    To identify the minimum, we will examine $F=G$.

    \begin{align}
        3(1-\delta-2\gamma) &= 1 - 2(\delta+\gamma)+\frac{3}{2}(\delta+\gamma)(1+\frac{\beta G }{3}) \notag \\
        2 &= \delta(\frac{5+\beta G}{2})+\gamma(\frac{11+\beta G}{2})\label{eq:gen_lb_3}
    \end{align}

    Thus, we can say that the minimum is achieved when
    $$\delta = \frac{2-b\gamma}{a}, \text{ for }a = \frac{5+\beta G}{2}, b = \frac{11+\beta G}{2}$$

    We can then substitute into either $F$ or $G$ to find the competitive ratio as a function of $\gamma$.
    \begin{align*}
        F(\gamma) &= 3(1-\frac{2-b\gamma}{a}-2\gamma) \\
        &=3(1-\frac{2}{a}+\frac{b\gamma}{a}-2\gamma) \\
        &= 3(\frac{1+\beta G}{\beta G+5}+\gamma(\frac{1-\beta G}{\beta G+5}))\\
        &= \frac{3(1+\beta G+\gamma-\gamma\beta G)}{\beta G+5}
    \end{align*}

    We can now find the value of $\gamma$ which minimizes $F(\gamma)$, subject to the constraints that $\alpha,\delta,\gamma \ge 0$. 

    By \autoref{eq:gen_lb_3} and the fact that $\delta\ge 0$, we can conclude that $\gamma \le \frac{4}{11+\beta G}$. Note that since $\beta G > 1$, $F(\gamma)$ is decreasing in $\gamma$. Thus, we find the minimum at $\gamma = \frac{4}{11+\beta G}$. 

    \begin{align*}
        F(\frac{4}{11+\beta G}) &= \frac{3(1+\beta G+\frac{4}{11+\beta G}-\beta G\frac{4}{11+\beta G})}{\beta G+5}\\
        &=3(\frac{\frac{(1+\beta G)(11+\beta G)}{11+\beta G}+\frac{4-4\beta G}{11+\beta G}}{5+\beta G}) \\
        &= 3(\frac{\beta^2 G^2+8\beta G+15}{(11+\beta G)(5+\beta G)}) \\
        &= 3\frac{\beta G+3}{\beta G+11}
    \end{align*}

    Thus, we can conclude that for $1.5 < \beta G \le 3$, any online algorithm will have a competitive ratio of at least $\frac{3(\beta G+3)}{\beta G+11}$.

    \subsubsection{$3 < \beta G \le 4$}

    Consider the inputs $\sigma_1 = \{\frac{G}{3},\ldots\}$, $\sigma_2=\{\frac{2G}{3},\ldots\}$, where each sequence contains $n$ items, where $n$ is a multiple of three. We will consider the performance of $\Alg$ on $\sigma_1$ and $\sigma_1\sigma_2$. 

    First consider the packing after $\sigma_1$. $\Alg$ can place between one to three items of size $G/3$ in a bin. Because $\beta G > 3$, it is never worth adding more than 3 items into a bin. Suppose $\Alg$ produces the packing $k$ bins $G$, $\ell$ bins $2G/3$, and $n-3k-2\ell$ bins $G/3$. The cost of $\Alg$ will be $n-2k-\ell$, while $\Opt$ is $n/3$. This results in a competitive ratio of $\frac{3(n-2k-\ell)}{n}$

    Now consider the packing after $\sigma_1\sigma_2$. The minimum cost packing possible will be: $n-3k-2\ell$ bins with one item $G/3$ and one item $2G/3$, $k$ bins with 3 items $G/3$, $\ell$ bins with 2 items $G/3$ and one item $2G/3$, and $3k+2\ell$ bins with one item $2G/3$. This will cost $n+k+\ell$, while $\Opt$ is $n$. This results in a competitive ratio of $\frac{n+k+\ell}{n}$.

    Now, we can say that the competitive ratio of $\Alg$ is at least $\max\{\frac{3(n-2k-\ell)}{n}, \frac{n+k+\ell}{n}\}$. The minimum value occurs at the intersection: $n=\frac{7k+4\ell}{2}$, for a value $CR \ge \frac{9k+6\ell}{7k+4\ell}$. In order for $n$ to be nonzero, we need to allow either $k, \ell$ to be nonzero. The expression is minimized when $\ell = 0$ and $k$ is nonzero, giving a lower bound of $9/7$ on the competitive ratio for $\Alg$.

    \subsubsection{$4 < \beta G \le 48$}

    Consider the inputs $\sigma_1 = \{G(\frac{2\beta G + 1}{6\beta G}), \ldots\}$ and $\sigma_2 = \{G(\frac{4\beta G - 1}{6\beta G}), \ldots\}$, where each sequence contains $n$ items with $n$ a multiple of six.

    Consider the performance of $\Alg$ on $\sigma_1$. $\Alg$ can place up to 3 items of size $G(\frac{2\beta G + 1}{6\beta G})$ in the same bin: any more items would create a trivially suboptimal packing. Note that if $\Alg$ packs 3 items, then it will incur a \black cost of $\beta\frac{3G}{\beta G} = 1/2$ for that bin. Thus, packing $6x$ items entirely in bins of 3 items is equivalent to packing those $6x$ items into bins of 2 items. Because neither of these bin types will be able to accept a $G(\frac{4\beta G - 1}{6\beta G})$, we can assume without loss of generality that $\Alg$ will only use bins with one or three items. We can now state that if there are $\ell$ bins with 3 items, then there are $n-3\ell$ bins with one item, for a cost of $n-1.5\ell$. The optimal cost on $\sigma_1$ is $n/2$, meaning that the competitive ratio is at least $\frac{2(n-1.5\ell)}{n}$.

    Now, consider $\sigma_1\sigma_2$. Since $\beta G > 4$, it is trivially suboptimal to place $G(\frac{4\beta G - 1}{6\beta G})$ items together, or to place them with more than one item of size $G(\frac{2\beta G + 1}{6\beta G})$. Therefore, the best packing of $\Alg$ will consist of $n-3\ell$ bins with one item of each size, $\ell$ bins with three small items, and $3\ell$ bins with one large item. The resulting cost is $n+1.5\ell$, while the optimal cost is $n$. Therefore, the competitive ratio is at least $\frac{n+1.5\ell}{n}$.

    We can perform a similar algebra as in the previous cases to find that the \CR is minimized when $n=4.5\ell$, and substitute to recover the lower bound $CR \ge 4/3$.

\subsubsection{$\beta G >48$}
Construct three item sequences: $\sigma_1 = (G(1/6-2\epsilon))^n$, $\sigma_2 = (G(1/3 + \epsilon))^n$, $\sigma_3 = (G(1/2 + \epsilon))^n$ for $\epsilon \in (\frac{1}{2\beta G}, \frac{1}{84}-\frac{1}{14\beta G})$, where each sequence contains $n$ copies of one item type. 
This construction ensures that \Opt never uses black space because any packing takes at least $1/\beta$ black space if its filling level exceeds $G$, and such packing is \textit{trivially suboptimal} as it can be improved by moving items into a newly opened bin. For example, under the input $\sigma_1$, \Opt packs $6$ items in a bin with final fill level $G(1 - 12 \epsilon)$ and the optimal cost is $\Opt(\sigma_1) = \frac{n}{6}$ without costs from using black space. This is because packing $7$ items in a bin takes black space at least $7G(1/6 - 2\epsilon) - G = G/6 - 14\epsilon > G/6 - 14G(1/84 - 1/(14 \beta G)) = 1/\beta$. Similarly, under the input $\sigma_1\sigma_2 = (\sigma_1, \sigma_2)$, \Opt packs $2$ type-one items and $2$ type-two items in a bin; under the input $\sigma_1\sigma_2\sigma_3 = (\sigma_1, \sigma_2,\sigma_3)$, \Opt packs one item from each type in a bin. Thus, $\Opt(\sigma_1\sigma_2) = n/2$ and $\Opt(\sigma_1\sigma_2\sigma_3) = n$.

Now, consider the packing produced by any online algorithm $\Alg$, Under $\sigma_1$, each bin will have at most 6 items and take no \black space. 
If a bin had more than 6 items, it would take more than $1/\beta$ \black space, meaning it would be cheaper for $\Alg$ to open a new bin for those additional items. 

Next consider the packing produced by $\Alg$ over $\sigma_1\sigma_2$.  \autoref{table:genLB-items} lists all bin types which are not trivially suboptimal (i.e. taking more than $1/\beta$ \black space). Note in fact that each of these bin types takes zero \black space.
\begin{table}[h!]
    \centering
        \begin{tabular}{|c c c c c c c}
        \hline
    & \textbf{(1,0)} & \textbf{(2,0)} & \textbf{(3,0)}& (4,0) & (5,0)  & (6,0) \\  
    \textbf{(0,1)} & \textbf{(1,1)} & (2,1) & (3,1) & (4,1) \\ 
    (0,2) & (1,2) & (2,2) 
\end{tabular}
    \caption{Possible bin types after $\sigma_1\sigma_2$. $(i,j)$ indicates $i$ items of size $G(1/6-2\epsilon)$ and $j$ items of size $G(1/3+\epsilon)$. Bolded items are those which could accept an item of size $G(1/2+\epsilon)$. }
    \label{table:genLB-items}
\end{table}
Finally, consider the packing produced by $\Alg$ over $\sigma_1\sigma_2\sigma_3$. By construction, $\Alg$ should never place more than one item of size $G(1/2+\epsilon)$ in a bin. Additionally, these items can only be placed in the designated bins shown in \autoref{table:genLB-items}, or alone in a bin.

Thus, we have shown that on the given inputs, $\Alg$ should never pack above $G$. 
The \Opt and the bin types possible for $\Alg$ are identical to those in Yao's lower bound (Theorem~$2$ in \cite{Yao1980genLB}), thus we can now apply that result as a \black box to find a general lower bound of $3/2$ when $\beta G > 48$.

\section{Additional Experimental Results}
\label{Appendix-Experiments}

In this section, we present supplemental experimental results for inputs following a continuous uniform $[0,1]$ distribution and the GI benchmark in the BPPLIB library \cite{Gschwind2016}\cite{Dolorme2018}. The uniform distribution is a common initial benchmark of empirical performance for bin packing algorithms \cite{Bentley1984, kamali2014all}. 
Meanwhile, the GI benchmark was introduced by Gschwind and Irnich and is now part of the BPPLIB library, which is a collection of bin packing benchmarks primarily used for offline bin packing. However, the GI benchmark was used by \cite{angelopoulos2024BPpredictions} to evaluate online algorithms. 

We separately present results for $\beta G \le 1$ and $\beta G > 1$. These results follow a similar structure to those using the Weibull distribution in \autoref{sec:experiments}. One consistent phenomena is that the Weibull distribution induces a slightly higher empirical competitive ratio than the uniform or GI distributions, and that the relative performance of \NF and \HAR is swapped when $\beta G \le 1$.

\subsection{$\beta G \le 1$}
For both the continuous (\autoref{fig:unif_GBleqOne}) and GI (\autoref{fig:GI_GBleqOne}) distributions, we fix values of $\beta \in \{1,3/2,4\}$ and select 30 evenly spaced values of $G \in (0,1/\beta]$. The general shape of the results is extremely similar in both figures, and follows closely with the Weibull distribution in \autoref{fig:unif_weibullleqOne}. Unlike with the Weibull distribution, \NF performs slightly worse than \HAR for each value of $\beta,G$.

\begin{figure}[hbt!]
    \centering
    \includegraphics[width=0.99\linewidth]{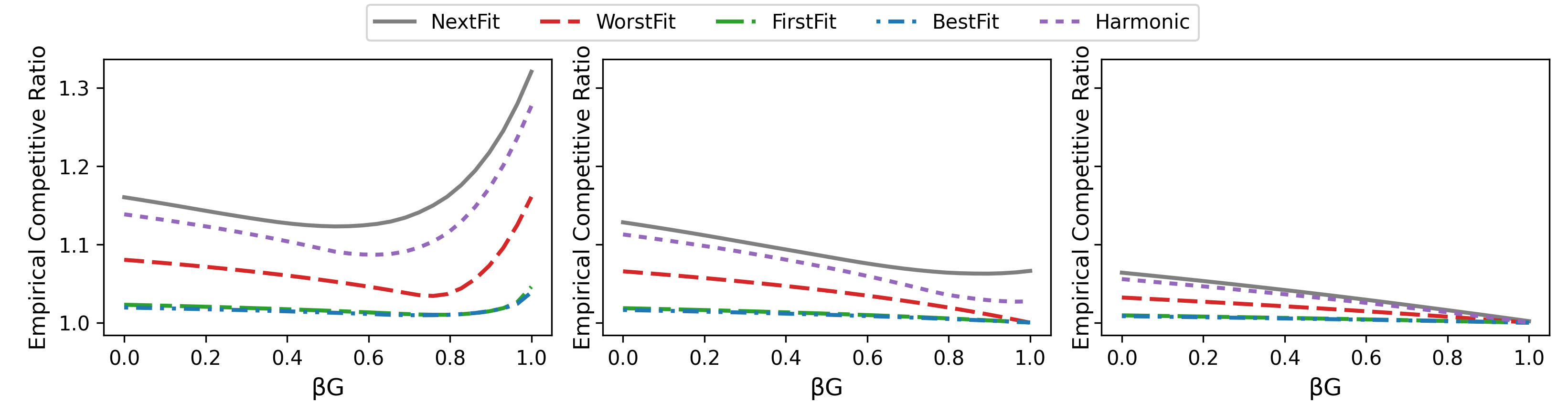}
    \makebox[0.99\textwidth][s]{ 
        \hspace{17mm}
        (a) $\beta = 1$
        \hspace{28mm}
        (b) $\beta = \nicefrac{3}{2}$
        \hspace{25mm}
        (c) $\beta = 4$
        \hspace{10mm}
    }
    \caption{Experiments on GI distribution with fixed $\beta$ for $\beta G \le 1$.}
    \Description{Plots are nearly identical to Figure \autoref{fig:unif_weibullleqOne}, but relative performance of \NF and \HAR are flipped.}
    \label{fig:GI_GBleqOne}
\end{figure}

\begin{figure}[hbt!]
    \centering
    \includegraphics[width=0.99\linewidth]{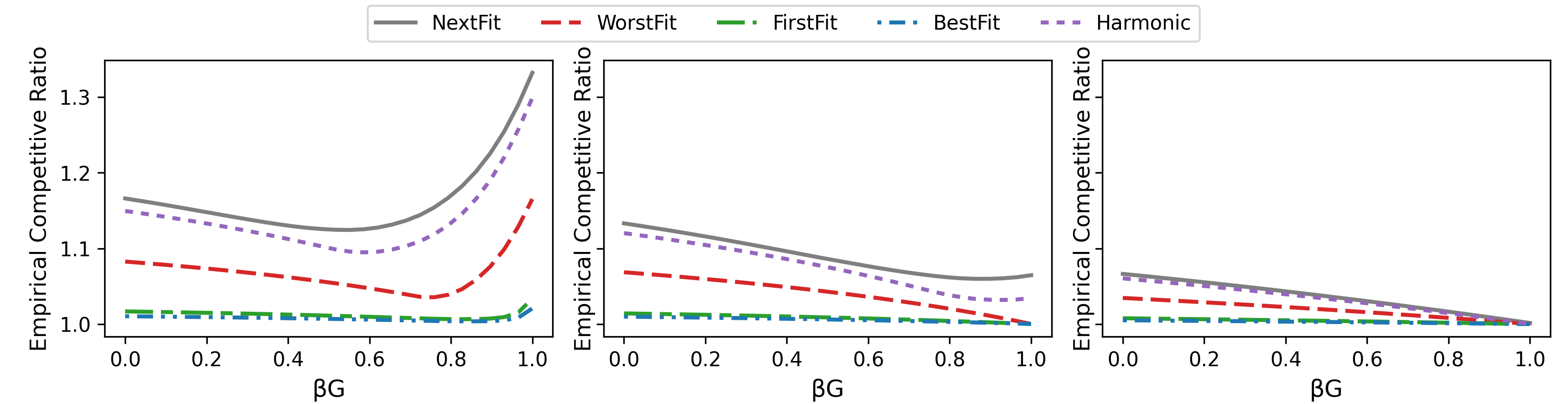}
    \makebox[0.99\textwidth][s]{ 
        \hspace{17mm}
        (a) $\beta = 1$
        \hspace{28mm}
        (b) $\beta = \nicefrac{3}{2}$
        \hspace{25mm}
        (c) $\beta = 4$
        \hspace{10mm}
    }
     \caption{Experiments on Uniform distribution with fixed $\beta$ for $\beta G \le 1$.}
    \Description{Plots are nearly identical to \autoref{fig:GI_GBleqOne}}
    \label{fig:unif_GBleqOne}
\end{figure}

\subsection{$\beta G > 1$}

The algorithms in this setting have the additional complexity of appropriate selection of the parameter $\tau$. We chose to evaluate the algorithms based on empirically-inspired threshold values. To find these values, we fixed $\beta \in \{5,10,20\}$, $ G = 0.5$, and varied $\tau$ to examine which value performs best empirically. 
We tested both Weibull (\autoref{fig:vary_tau_w}) and uniform (\autoref{fig:vary_tau_u}) distributions, selecting 150 values of $\tau$ for each. While there is some variation between the distributions and across different values of $\beta$, we observe relatively consistent values for empirical $\tau$ selection. First, we can observe that \TFF and \TBF perform best when $\tau = 0$. Meanwhile, \TWF minimizes cost when $\tau \approx \nicefrac{1}{2\beta}$. Finally, \TNF and \GHAR minimize cost when $\tau \approx \nicefrac{1}{\beta}$. We then used these values in subsequent experiments. 

\begin{figure}[hbt!]
    \centering
    \includegraphics[width=0.99\linewidth]{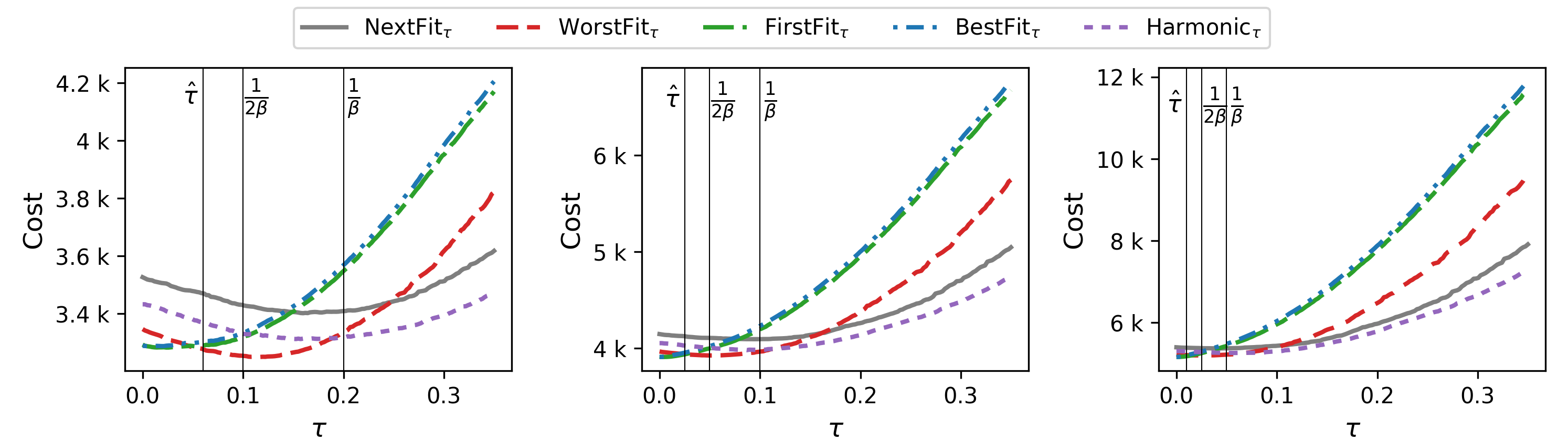}
    \makebox[0.99\textwidth][s]{ 
        \hspace{18mm}
        (a) $\beta = 5$
        \hspace{27mm}
        (b) $\beta = 10$
        \hspace{25mm}
        (c) $\beta = 20$
        \hspace{10mm}
    }
    \caption{Experiments on Weibull distribution with $G=\nicefrac{1}{2}$ to determine empirically best $\tau$ for each algorithm}
    \Description{Three plots mapping empirical performance of threshold algorithms, we find different empirical thresholds work best for different algorithms.}
    \label{fig:vary_tau_w}
\end{figure}

\begin{figure}[hbt!]
    \centering
    \includegraphics[width=0.99\linewidth]{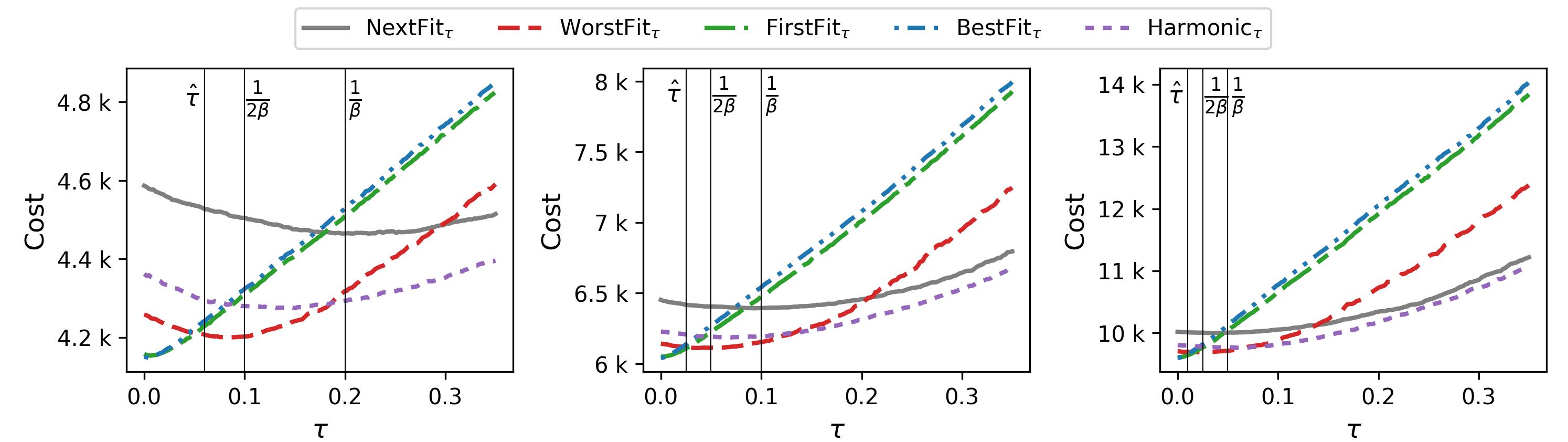}
    \makebox[0.99\textwidth][s]{ 
        \hspace{18mm}
        (a) $\beta = 5$
        \hspace{27mm}
        (b) $\beta = 10$
        \hspace{25mm}
        (c) $\beta = 20$
        \hspace{10mm}
    }
    \caption{Experiments on uniform distribution with $G= \nicefrac{1}{2}$ to determine empirically best $\tau$ for each algorithm}
    \Description{Empirical performance on Uniform distribution suggests the same thresholds as the Weibull distribution}
    \label{fig:vary_tau_u}
\end{figure}

\autoref{fig:GI_GBgeOne} and \autoref{fig:unif_GBgeOne} both use the empirically selected values of $\tau$ to evaluate the GI and uniform distributions when $G \in \{1/2, 3/4,19/20\}$ and $\beta \in (1.005/G,20/G]$. We again see very similar trends as for the Weibull distribution in \autoref{fig:WeibullgeOne}, except that the empirical competitive ratio for continuous and GI distributions is slightly smaller. We still see three clear groups of algorithmic performance, particularly for larger 
$G$, where \TFF and \TBF are close to optimal, \TWF is in the middle, and \TNF and \GHAR perform the worst.

\begin{figure}[hbt!]
    \centering
    \includegraphics[width=0.99\linewidth]{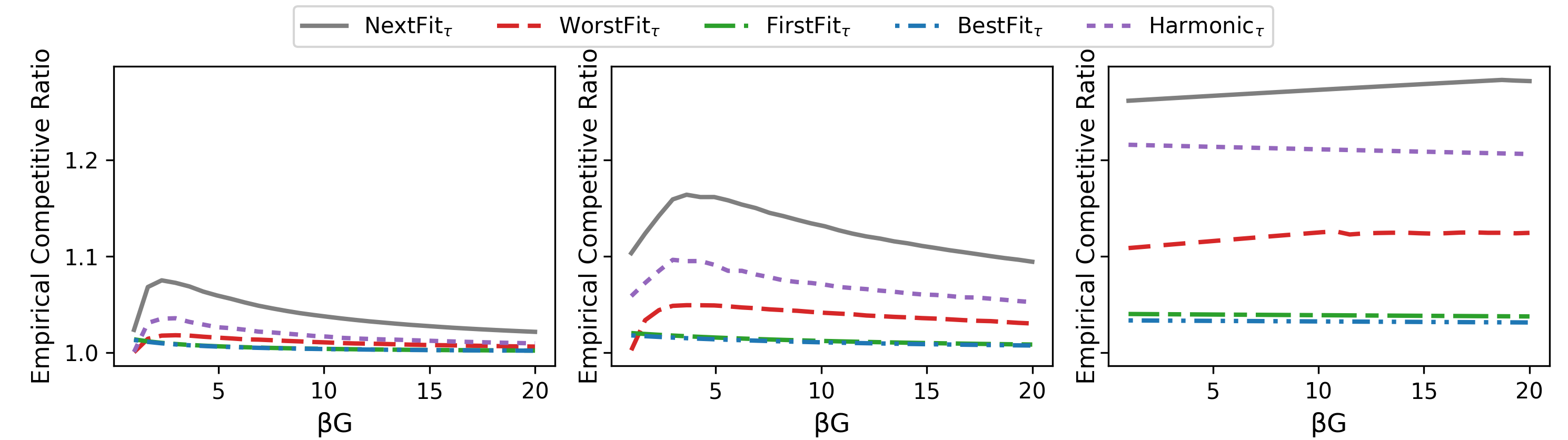}
    \makebox[0.99\textwidth][s]{ 
        \hspace{17mm}
        (a) $G = \nicefrac{1}{2}$
        \hspace{28mm}
        (b) $G= \nicefrac{3}{4}$
        \hspace{25mm}
        (c) $G= \nicefrac{19}{20}$
        \hspace{10mm}
    }
    \caption{Experiments on GI distribution with fixed $G$ using empirically determined $\tau$.}
    \Description{Plots show similar shapes to the Weibull test in \autoref{fig:WeibullgeOne}, however the competitive ratio is generally lower in these plots.}
    \label{fig:GI_GBgeOne}
\end{figure}

\begin{figure}[hbt!]
    \centering
    \includegraphics[width=0.99\linewidth]{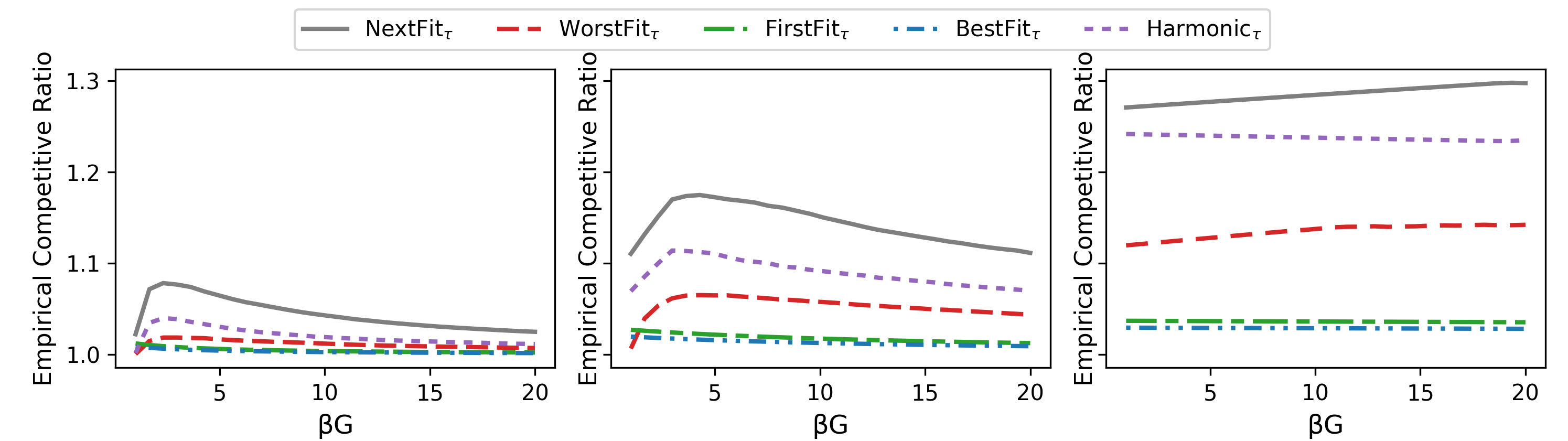}
    \makebox[0.99\textwidth][s]{ 
        \hspace{17mm}
        (a) $G = \nicefrac{1}{2}$
        \hspace{28mm}
        (b) $G= \nicefrac{3}{4}$
        \hspace{25mm}
        (c) $G= \nicefrac{19}{20}$
        \hspace{10mm}
    }
    \caption{Experiments on uniform distribution with fixed $G$ using empirically determined $\tau$.}
    \Description{Plots are nearly identical to the GI version in \autoref{fig:GI_GBgeOne}.}
    \label{fig:unif_GBgeOne}
\end{figure}

\end{document}